\pdfoutput=1

\documentclass{scrartcl}

\usepackage{cmap}
\usepackage[utf8]{inputenc}
\usepackage[T1]{fontenc}

\usepackage[nocompress]{cite}

\usepackage{amsmath}
\usepackage{amssymb}
\usepackage{mathtools}
\usepackage{mathrsfs} 
\usepackage{bm}
\usepackage{csquotes}
\usepackage{enumitem}
\usepackage{multicol}

\usepackage{everypage}

\usepackage{subcaption}
\usepackage{booktabs}

\usepackage{todonotes}

\usepackage{amsthm}
\usepackage{thmtools}

\theoremstyle{plain}
\newtheorem{theorem}{Theorem}[section]
\newtheorem{lemma}[theorem]{Lemma}
\newtheorem{proposition}[theorem]{Proposition}
\newtheorem{corollary}[theorem]{Corollary}
\newtheorem{fact}[theorem]{Fact}

\theoremstyle{definition}
\newtheorem{definition}[theorem]{Definition}
\newtheorem{example}[theorem]{Example}
\newtheorem{openproblem}[theorem]{Open Problem}

\theoremstyle{remark}
\newtheorem{remark}[theorem]{Remark}
\newtheorem*{remark*}{Remark}

\usepackage{tikz}
\usetikzlibrary{automata, matrix, positioning, calc, decorations.pathreplacing, shapes.geometric}
\newlength{\edgelength}
\newcommand{\trans}[4]{%
  \begin{tikzpicture}[auto, shorten >=1pt, >=latex, baseline=(l.base), inner sep=0pt, outer xsep=0.3333em]
    \node (l) {\ensuremath{#1}};%
    \setlength{\edgelength}{\widthof{\scriptsize\ensuremath{#2/#3}}+0.5cm}%
    \node[base right=\edgelength of l] (r) {\ensuremath{#4}};%
    \path[->] (l.mid east) edge node[inner sep=0pt] {\scriptsize\ensuremath{#2/#3}} (r.mid west);%
  \end{tikzpicture}%
}

\usepackage{calc}
\usepackage{tabularx}
\newcommand{\problem}[3][]{%
  \par\vspace{0.125cm plus 0.1cm minus 0.05cm}\begin{tabularx}{\textwidth-2\parindent}{lX}%
    \if\relax\detokenize{#1}\relax%
    \else%
    \textnormal{\textbf{Constant:}}&#1\\%
    \fi%
    \textnormal{\textbf{Input:}}&#2\\%
    \textnormal{\textbf{Question:}}&#3\\%
  \end{tabularx}\vspace{0.125cm plus 0.1cm minus 0.05cm}\par%
}

\newcommand{\changefont}[3]{\fontfamily{#1}\fontseries{#2}\fontshape{#3}\selectfont}
\newcommand*{\DecProblem}[1]{{\changefont{cmtt}{m}{sc}#1}}

\usepackage{xspace}

\newcommand*{\SAut}{$\mathscr{S}$\kern-0.4ex-\allowbreak{}au\-to\-ma\-ton\xspace}
\newcommand*{\SAuta}{$\mathscr{S}$\kern-0.4ex-\allowbreak{}au\-to\-ma\-ta\xspace}
\newcommand*{\SIAut}{$\inverse{\mathscr{S}}$\kern-0.4ex-\allowbreak{}au\-to\-ma\-ton\xspace}
\newcommand*{\SIAuta}{$\inverse{\mathscr{S}}$\kern-0.4ex-\allowbreak{}au\-to\-ma\-ta\xspace}
\newcommand*{\GAut}{$\mathscr{G}$\kern-0.2ex-\allowbreak{}au\-to\-ma\-ton\xspace}
\newcommand*{\GAuta}{$\mathscr{G}$\kern-0.2ex-\allowbreak{}au\-to\-ma\-ta\xspace}

\newcommand*{\id}{\operatorname{id}}

\usepackage{bbold}
\newcommand{\idGrp}{\mathbb{1}}

\newcommand*{\ComplexityClass}[1]{\textsc{#1}\xspace}
\newcommand*{\PSpace}{\ComplexityClass{PSpace}}

\makeatletter
\def\gnewcommand{\g@star@or@long\new@command}
\def\g@star@or@long#1{%
  \@ifstar{\let\l@ngrel@x\global#1}{\def\l@ngrel@x{\long\global}#1}}
\makeatother

\usepackage[affil-it]{authblk}

\author{Daniele D'Angeli}
\affil{Dipartimento di Ingegneria\\
  Università degli Studi Niccolò Cusano\\
  Via Don Carlo Gnocchi, 3\\
  00166 Roma, Italy}
\author{Emanuele~Rodaro}
\affil{Department of Mathematics\\
  Politecnico di Milano\\
  Piazza Leonardo da Vinci, 32\\
  20133 Milano, Italy}
\author{Jan~Philipp~Wächter}
\affil{Department of Mathematics\\
  University of Manchester\\
  Oxford Road\\
  Manchester M13 9PL, UK}

\usepackage{hyperref}

\title{The Freeness Problem for\\Automaton Semigroups}

\begin{document}
  \maketitle
  
  \begin{abstract}
    We show that the freeness problems for automaton semigroups and for automaton monoids are undecidable and, thereby, solve an open problem listed by Grigorchuk, Nekrashevych and Sush\-chansk\u{\i}i.
    We achieve this using a new technique to encode Post's Correspondence Problem into automaton semigroups and monoids and our result even holds if we restrict the alphabet of the input automata to a constant size.
    The encoding allows us to precisely control the relations in the generated semigroup/monoid and the construction is quite versatile. In fact, we obtain further undecidability results on various semigroup notions (left cancellativity, equidivisibility and extending homomorphisms).
    Our construction can also be adapted to show that the free presentation problem for automaton monoids is undecidable (and yields a weaker statement in the semigroup case).
    \\
    \textit{Note.} This paper is the full, extended journal version of these results presented at MFCS 2024.
  \end{abstract}
  
  \begin{section}{Introduction}
    \enlargethispage{\baselineskip}
    In the 1980s, Grigorchuk solved a famous question by Milnor (see \cite{grigorchuk2008groups} for a nice introduction) by presenting the first group with intermediate growth: the number of elements that can be written as a word of length at most $n$ over the generators grows sub-exponentially but super-polynomially. The group has even more noteworthy properties. It is amenable but not elementary amenable (e.\,g.\ \cite{juschenko2022amenability}) and an infinite 2-group (giving a counter-example to Burnside's problem, e.\,g.\ \cite{nekrashevych2005self,bartholdi2021groups}).
    Its peculiar properties stirred interest in Grigorchuk’s group and groups of similar form where it soon became important that Grigorchuk’s group has a nice description using what is simply called an \emph{automaton} in this context (e.\,g.\ \cite{nekrashevych2005self} or \cite{bartholdi2021groups}). The simplicity of this presentation (the automaton only uses a binary alphabet and four states -- with an additional identity state) contrasts the complex nature of the group. An \enquote{automaton} here is what more precisely is called a finite-state letter-to-letter transducer (i.\,e.\ an automaton with input and output). The idea is that, in such an automaton, every state induces a mapping of input to output words and the closure of these functions under composition forms a semigroup. If the automaton is additionally invertible, the functions are bijections and we may consider the generated group. This leads to the classes of \emph{automaton semigroups} and \emph{groups}, which contain further noteworthy examples (e.\,g.\ Gupta-Siki $p$-groups \cite{gupta1983burnside}, the lamplighter group \cite{grigorchuk2001lamplighter} and more general lamplighter-like groups \cite{silva2005class,skipper2020lamplighter}).
    
    Being able to finitely describe groups without classical finite presentations (consisting of generators and relations) additionally highlights the usefulness of considering (semi)groups generated by automata. Starting from Grigorchuk’s group, the study of automaton groups and semigroups is nowadays a thriving research field with important connections to many neighboring areas (such as geometry, dynamical systems and symbolic dynamics; see e.\,g.\ \cite{nekrashevych2005self,bartholdi2021groups} for more background information). The extensive research in Mathematics and Computer Science on the semigroup (and monoid) case (e.\,g.\ \cite{cain2009automaton,klimann2016automaton,brough2017automaton,picantin2019automatic,bartholdi2020hierarchy,structurePart}) arises naturally from the group case for example via the \emph{dual automaton} where states and input/output letters swap places. The connection between an automaton and its dual has been exploited algebraically and algorithmically (e.\,g.\ \cite{glasner2005automata,vorobets2007free,vorobets2010series,klimann2016automaton,klimann2018infinity,orbitsPart}).
    
    In this work, we look further at the algorithmic aspects of this interesting class by showing that its \emph{freeness problem} is undecidable. This problem asks whether a given automaton generates a free semigroup (or monoid). It has been studied extensively for other classes of groups and semigroups. Since freeness is a Markov property, the problem is undecidable for classical finite group presentations (see e.\,g.\ \cite{lyndon2001combinatorial}).
    A similar approach may also be used to show that it is undecidable for finitely presented monoids; interestingly, however, it turns out to be decidable for finitely presented semigroups (see \cite{nybergbrodda2024aidanRabin} for more details).
    Further important results include the undecidability of the freeness problem for matrix semigroups, originally shown using a reduction from Post’s Correspondence Problem \cite{klarner1991undecidability}, which has been improved and contrasted in many further publications (e.\,g.\ \cite{mandel1977finite,cassaigne1999undecidability,bell2008reachability}). Interestingly, matrix (semi)groups and automaton (semi)groups are connected in the sense that the former can be presented as subgroups of the latter \cite{brunner1998generation} (see also \cite{sunic2012conjugacy,decidabilityPart,waechter2020automaton}) but this does not help to prove the freeness problem undecidable for automaton (semi)groups \cite{decidabilityPartErratum}.
    
    With our result, we continue this line of research but also further contribute to the study of freeness in self-similar (i.\,e.\ generated by infinite automata) and automaton structures as well as their algorithmic aspects. For the former, we refer the reader to the survey \cite{rodaro2022selfSimilarity} and only point out that, while it is known that free groups are automaton groups \cite{vorobets2007free,vorobets2010series,steinberg2011automata}, these constructions are usually deemed rather difficult. For automaton semigroups and monoids, the situation seems to be simpler: every free semigroup of (finite) rank at least two can be generated by an automaton (see \cite{cain2009automaton} or \autoref{ex:freeSemigroups}) but the free semigroup of rank one cannot \cite{cain2009automaton}. All free monoids of finite rank are automaton semigroups, though.
    
    Regarding algorithmic questions for automaton (semi)groups, we point out that, while one may easily be misled into believing that using a finite automaton as the generating combinatorial object should be rather simple, the situation is actually quite complex and only a few natural algorithmic problems are known to be undecidable while many others notoriously remain open problems.
    An exception here seems to be that the word problem for automaton (semi)groups is $\PSpace$-complete. Interestingly, this was first known for semigroups \cite{dangeli2017complexity} and was later extended to groups \cite{waechter2020automaton}. Some subclasses have simpler word problems. For example, using finitary automata to present finite groups results in a $\ComplexityClass{coNP}$-complete word problem \cite{kotowsky2023word} and the word problem of an automaton group of polynomial activity is in polylogarithmic space \cite{bondarenko2012growth} (see \cite{waechter2024word} for more information).
    On the other hand, there is an automaton group with an undecidable conjugacy problem \cite{sunic2012conjugacy} (\enquote{are two given group elements conjugate in the group?}). The construction used there also shows that the isomorphism problem for automaton groups (\enquote{are the groups generated by two given automata isomorphic?}) and, thus, automaton semigroups is undecidable.\footnote{Unfortunately, this does not seem to be written down explicitly anywhere.} There are two constructions for an automaton group with undecidable order problem (\enquote{has a given group element finite or infinite order?}) \cite{gillibert2018automaton,bartholdi2017word}. The latter of the two even yields a contracting automaton. The undecidability was also first known for automaton semigroups \cite{gillibert2014finiteness} and the problem is decidable for bounded automaton groups \cite{bondarenko2013conjugacy}, monoids \cite{bartholdi2020hierarchy} and semigroups \cite{boundedFinitenessSemigroups}.
    
    All these constructions encoding Turing machines in automaton (semi)groups make a statement about individual (semi)group elements. Since the interaction between the generating automaton and generated algebraic structure is often surprising and still not well understood, it is much more challenging to construct reductions where the entire generated (semi)group (or monoid) has a certain property (based on whether we input a positive or negative problem instance). The only known result of this kind seems to be that the finiteness problem for automaton semigroups (\enquote{Is the semigroup generated by a given automaton finite?}) is undecidable \cite{gillibert2014finiteness}. The corresponding group problem is still open \cite{grigorchuk2000automata}.
    
    Our reduction from Post's Correspondence Problem \cite{post1946variant} to the freeness problems for automaton semigroups and for monoids in this paper is a second result of this form. It solves the corresponding open problem by Grigorchuk, Nekrashevych and Sush\-chansk\u{\i}i \cite[7.2 b)]{grigorchuk2000automata}; in fact, we show that an even stronger undecidability result holds where the input automata are guaranteed to have an alphabet size of (at most) $25$.
    Despite previous attempts \cite{decidabilityPart,decidabilityPartErratum} and a positive result for semigroups generated by invertible and reversible automata with two states \cite{klimann2016automaton} as well as a negative result on testing for relations of the form $w = \idGrp$ \cite{decidabilityPart}, the problem had remained open quite a while for groups and for semigroups.
    The main challenge seems to be that we need very precise control over the relations in the generated semigroup (which seems to be much more difficult than, e.\,g., ensuring that the semigroup is finite or infinite) while the interaction between the structure of the generating automaton and the semigroup/monoid relations is highly non-obvious.

    \enlargethispage{0.5\baselineskip}
    Our construction yields further results beyond the freeness problem(s). Namely, testing whether a given automaton generates a (left) cancellative semigroup/monoid and whether the semigroup/monoid generated by a given automaton is equidivisible (a notion strongly related to freeness by Levi's lemma, see \autoref{fct:leviLemma}) are undecidable. We also obtain that it is undecidable whether a given automaton generates a free semigroup with a given basis and whether a given map between the state sets of two given automata can be extended into an iso- or homomorphism. The latter problem is connected to the (undecidable, see above) isomorphism problem for automaton semigroups in the sense that it asks whether all relations of the first automaton semigroup also hold in the second one.
    
    Finally, the construction seems to be flexible enough to be adapted to similar problems, which gives us hope that our results could also contribute towards showing that the freeness problem is undecidable in the group case. For example, it can be adapted to show that the free presentation problem for automaton monoids is undecidable: does a given automaton generate a free monoid whose rank is equal to the number of its states (minus an identity state)? In other words, we cannot test whether a given automaton monoid contains any relations (although this is semi-decidable as the word problem is decidable, see above).
    
    Adapting our construction for this is necessary because the construction in the semigroup case always yields semigroup relations since we need to use a result on the closure of the class of automaton semigroups under (certain) free products \cite{brough2023automaton} in order to construct some kind of \enquote{partial} powers of the generating automaton. However, no details of this construction will be required to understand our results. More generally, the presentation in this work is meant to be self-contained (although the construction may be considered to be rather technical).
    
    The current version of this paper is a full journal version including all proof details of the results presented at the 49th International Symposium on Mathematical Foundations of Computer Science (MFCS 2024) \cite{conferenceVersion}. The extension of the undecidability of the freeness problems for automaton semigroups and for automaton monoids to a constant alphabet size (precisely to alphabet size $25$) is novel.
  \end{section}
  
  \begin{section}{Preliminaries}\label{sec:preliminaries}
    \paragraph*{Fundamentals, Semigroups and Monoids.}
    We use $A \uplus B$ to denote the disjoint union of two sets $A$ and $B$. We consider the set of natural numbers $\mathbb{N}$ to contain $0$.
    
    We assume the reader to be familiar with the most fundamental notions of semigroup theory (see e.\,g.\ \cite{howie1995fundamentals}).
    We denote the neutral element of a monoid $M$ by $\idGrp_M$ or, if $M$ is clear from the context, simply by $\idGrp$. For a monoid $M$, we let $M^\idGrp = M$ and, if $S$ is a semigroup but not a monoid, we may adjoin a neutral element $\idGrp \not\in S$ to $S$ by letting $\idGrp \idGrp = \idGrp$ and $\idGrp s = s = s \idGrp$ for all $s \in S$ and denote the resulting monoid by $S^\idGrp$.
    
    \paragraph*{Words, Free Semigroups and Free Monoids.}
    Let $B$ be a finite, non-empty set, which we call an \emph{alphabet}. A \emph{word} $w$ over the alphabet $B$ is a finite sequence $a_1 \dots a_n$ with $a_1, \dots, a_n \in B$, whose \emph{length} is $|w| = n$. We use $\varepsilon$ to denote the unique word of length $0$ (i.\,e.\ the \emph{empty word}). The set of all words over $B$ is denoted by $B^*$. Words have the natural operation of juxtaposition (where we let $uv = a_1 \dots a_m b_1 \dots b_n$ for $u = a_1 \dots a_m$ and $v = b_1 \dots b_n$ with $a_1, \dots, a_m, b_1, \dots, b_n \in B$), which turns $B^*$ into a monoid with $\varepsilon$ as the neutral element. This monoid $B^*$ is \emph{the free monoid} with \emph{basis} $B$ (or the free monoid \emph{over} $B$) and a monoid $M$ is \emph{free} (with basis $B$) if it is isomorphic to $B^*$ (for some alphabet $B$).
    Closely related to the free monoid is \emph{the free semigroup} $B^+$, which is formed by the set of all non-empty words (i.\,e.\ $B^+ = B^* \setminus \{ \varepsilon \}$) and (again) juxtaposition as operation. Similarly, a semigroup $S$ is \emph{free} (with basis $B$) if it is isomorphic to $B^+$ (for some alphabet $B$).
    Note that $B^*$ is (isomorphic to) $(B^+)^\idGrp$.
    Also note that the basis of a free monoid or semigroup is unique (see e.\,g.\ \cite[Proposition~7.1.3]{howie1995fundamentals}).
    The \emph{rank} of a free monoid or semigroup is the cardinality $|B|$ of its basis $B$.
    
    To lighten our notation, we will use some common conventions form formal language theory. For example, we will sometimes identify $q$ with the singleton sets $\{ q \}$; in particular, we will write $q^+$ and $q^*$ instead of $\{ q \}^+$ and $\{ q \}^*$.

    \paragraph{Properties of Free Semigroups and Monoids.}
    
    We will need some properties of free semigroups and monoids. A (general) semigroup $S$ is \emph{left cancellative} if $st = st'$ implies $t = t'$ for all $s, t, t' \in S$. Symmetrically, it is \emph{right cancellative} if $st = s't$ implies $s = s'$ for all $s, s', t \in S$ and, finally, it is \emph{cancellative} if it is both left and right cancellative. It is easy to see that $B^*$ and, thus, $B^+$ are cancellative (see, e.\,g.\ \cite[Proposition~7.1.1]{howie1995fundamentals}).
    \begin{fact}\label{fct:freeIsCancellative}
      Free semigroups and free monoids are cancellative.
    \end{fact}
    
    A \emph{length function} of a semigroup $S$ is a homomorphism $S \to \mathbb{N}_{> 0}$ where $\mathbb{N}_{> 0}$ is the additive semigroup of strictly positive natural numbers. A monoid $M$ has a \emph{proper length function} if there is a monoid homomorphism $M \to \mathbb{N}$ (where $\mathbb{N}$ is the additive monoid of the natural numbers including $0$) such that $\idGrp$ is the only pre-image of $0$ (i.\,e.\ only $\idGrp$ has length $0$, all other elements have strictly positive length). A semigroup $S$ that is not a monoid has a length function if and only if $S^\idGrp$ has a proper one and free semigroups and monoids do have (proper) length functions (mapping a word to its length).
    
    A semigroup (or monoid) $S$ is \emph{equidivisable} if, for all $s_1, s_2, s_1', s_2' \in S$ with $s_1 s_2 = s_1' s_2'$, there is some $x \in S^\idGrp$ with $s_1 = s_1' x$ and $x s_2 = s_2'$ or with $s_1 x = s_1'$ and $s_2 = x s_2'$ (see \autoref{fig:equidivisibility}).
    The idea for this definition is that two factorizations of the same semigroup element have a common subfactorization.
    \begin{figure}\centering
      \begin{subfigure}{0.5\linewidth}%
        \centering%
        \begin{tikzpicture}
          \matrix [rectangle, draw, every node/.style={inner sep=3pt}, matrix of math nodes, ampersand replacement=\&, inner sep=0pt] (s) {
            \hspace*{1.5cm} s_1\vphantom{s_1'}\hspace*{1.5cm} \& \hspace*{1cm} s_2\vphantom{s_1'} \hspace*{1cm} \\
          };
          
          \draw (s-1-1.north east) -- (s-1-1.south east);
          
          \matrix [below=1cm of s, rectangle, draw, every node/.style={inner sep=3pt}, matrix of math nodes, ampersand replacement=\&, inner sep=0pt] (s') {
            \hspace*{1cm} s_1'\hspace*{1cm} \& \hspace*{1.5cm} s_2' \hspace*{1.5cm} \\
          };
          
          \draw (s'-1-1.north east) -- (s'-1-1.south east);
          
          \path (s'-1-1.north east) |- coordinate (c1) (s-1-1.south east);
          \path (s'-1-1.north east) -| coordinate (c2) (s-1-1.south east);
          \draw[decorate, decoration=brace] ($(s-1-1.south east)-(0pt,0.5mm)$) -- ($(c1)-(0pt,0.5mm)$);
          \draw[decorate, decoration=brace] ($(s'-1-1.north east)+(0pt,0.5mm)$) -- node[above=2mm] {$x$} ($(c2)+(0pt,0.5mm)$);
        \end{tikzpicture}%
        \caption{$s_1$ is longer than $s_1'$}%
      \end{subfigure}%
      \begin{subfigure}{0.5\linewidth}%
        \centering%
        \begin{tikzpicture}
          \matrix [rectangle, draw, every node/.style={inner sep=3pt}, matrix of math nodes, ampersand replacement=\&, inner sep=0pt] (s) {
            \hspace*{1cm} s_1\vphantom{s_1'}\hspace*{1cm} \& \hspace*{1.5cm} s_2\vphantom{s_1'} \hspace*{1.5cm} \\
          };
          
          \draw (s-1-1.north east) -- (s-1-1.south east);
          
          \matrix [below=1cm of s, rectangle, draw, every node/.style={inner sep=3pt}, matrix of math nodes, ampersand replacement=\&, inner sep=0pt] (s') {
            \hspace*{1.5cm} s_1'\hspace*{1.5cm} \& \hspace*{1cm} s_2' \hspace*{1cm} \\
          };
          
          \draw (s'-1-1.north east) -- (s'-1-1.south east);
          
          \path (s'-1-1.north east) |- coordinate (c1) (s-1-1.south east);
          \path (s'-1-1.north east) -| coordinate (c2) (s-1-1.south east);
          \draw[decorate, decoration=brace] ($(c1)-(0pt,0.5mm)$) -- ($(s-1-1.south east)-(0pt,0.5mm)$);
          \draw[decorate, decoration=brace] ($(c2)+(0pt,0.5mm)$) -- node[above=2mm] {$x$} ($(s'-1-1.north east)+(0pt,0.5mm)$);
        \end{tikzpicture}%
        \caption{$s_1$ is shorter than $s_1'$}%
      \end{subfigure}%
      \caption{Graphical representation of equidivisibility \cite[Figure~2.8]{waechter2020automaton}.}\label{fig:equidivisibility}
    \end{figure}
    It is not difficult to see that free semigroups and monoids are equidivisible (see e.\,g.\ \cite[Proposition~7.1.2]{howie1995fundamentals}). Together with having a (proper) length function, this turns out to characterize free semigroups and monoids (see e.\,g.\ \cite[Proposition~7.1.8]{howie1995fundamentals}).
    \begin{fact}[Levi's Lemma]\label{fct:leviLemma}
      A semigroup (monoid) $S$ is free if and only if it is equidivisible and has a (proper) length function.
    \end{fact}
    
    \paragraph*{Free Products of Semigroups.}
    A \emph{semigroup presentation} is a pair $\langle Q \mid \mathcal{R} \rangle_\mathscr{S}$ of a set of \emph{generators} $Q$ and a (possibly infinite) set of \emph{relations} $\mathcal{R} \subseteq Q^+ \times Q^+$. We will only consider presentations where $Q$ is finite and non-empty.
    If we denote by $\mathcal{C}$ the smallest congruence $\mathcal{C} \subseteq Q^+ \times Q^+$ containing $\mathcal{R} \subseteq \mathcal{C}$, the semigroup \emph{presented} by such a presentation is $S = Q^+ / \mathcal{C}$ formed by the congruence classes $[\cdot]$ of $\mathcal{C}$ with the (well-defined!) operation $[ u ] \cdot [v] = [uv]$. We simply write $S = \langle Q \mid \mathcal{C} \rangle_{\mathscr{S}}$ in this case. Every semigroup generated by a finite, non-empty set $Q$ is presented by some semigroup presentation of this form.
    
    The free product of the semigroups $S = \langle Q \mid \mathcal{S} \rangle_\mathscr{S}$ and $T = \langle P \mid \mathcal{R} \rangle_\mathscr{S}$ is the semigroup $S \star T = \langle Q \uplus P \mid \mathcal{S} \cup \mathcal{R} \rangle_\mathscr{S}$. For example, we have $\{ p, q \}^+ = p^+ \star q^+$.
    
    \begin{remark*}
      Of course, there is also the free product of monoids (and monoid presentations). However, in this paper, we will only consider free products of semigroups. In particular, we do \textbf{not} have $\{ p, q \}^* = p^* \star q^*$.
    \end{remark*}

    \paragraph*{Automata.}
    In the context of the current paper, an \emph{automaton} is a triple $\mathcal{T} = (Q, \Sigma, \delta)$ consisting of a non-empty, finite set of \emph{states}, an \emph{alphabet} $\Sigma$ and a set $\delta \subseteq Q \times \Sigma \times \Sigma \times Q$ of \emph{transitions}.
    
    \begin{remark*}
      What we simply call an automaton here would rather be called a finite-state, letter-to-letter transducer in more general automaton-theoretic terms.
      However, simply using the term \enquote{automaton} is standard terminology in the area. We also do not use initial or final states as they do not interact nicely with the self-similar nature of the semigroups and monoids generated by automata we are about to define.
    \end{remark*}
    
    Within the context of transitions, we will use the graphical notation $\trans{p}{a}{b}{q}$ to denote $(p, a, b, q) \in Q \times \Sigma \times \Sigma \times Q$. Such a transition \emph{starts} in $p$, \emph{ends} in $q$, its \emph{input} is $a$ and its \emph{output} is $b$.
    Additionally, we use the common way of depicting automata as illustrated in \autoref{fig:exampleTransition}, which indicates that the automaton contains the transition $\trans{p}{a}{b}{q} \in \delta$.
    \begin{figure}\centering
      \begin{tikzpicture}[auto, shorten >=1pt, >=latex, baseline=(p.base)]
        \node[state] (p) {$p$};
        \node[state, right=of p] (q) {$q$};
        
        \draw[->] (p) edge node {$a / b$} (q);
      \end{tikzpicture}
      \caption{Example of depicting a transition in an automaton.}\label{fig:exampleTransition}
    \end{figure}
    When dealing with an automaton $\mathcal{T} = (Q, \Sigma, \delta)$, we are actually dealing with two alphabets ($Q$ and $\Sigma$). In order to avoid confusion, we call the elements of $Q$ \emph{states} and the elements of $Q^*$ \emph{state sequences}, while reserving the terms \emph{letters} and \emph{words} for the elements of $\Sigma$ and $\Sigma^*$, respectively.
    
    \begin{figure}[t]%
      \centering%
      \begin{subfigure}[t]{0.2\linewidth}%
        \centering%
        \begin{tikzpicture}[baseline=(m-2-1.base)]
          \matrix[matrix of math nodes, text height=1.25ex, text depth=0.25ex] (m) {
              & a & \\
            p & & q \\
              & b & \\
          };
          \foreach \i in {1} {
            \draw[->] let
              \n1 = {int(2+\i)}
            in
             (m-2-\i) -> (m-2-\n1);
            \draw[->] let
             \n1 = {int(1+\i)}
            in
              (m-1-\n1) -> (m-3-\n1);
          };
        \end{tikzpicture}%
        \caption{Single tran\-si\-tion cross di\-a\-gram.}\label{sfig:singleCrossDiagram}%
      \end{subfigure}\hfill%
      \begin{subfigure}[t]{0.5\linewidth}%
        \centering%
        \begin{tikzpicture}[baseline=(m-4-1.east)]
          \matrix[matrix of math nodes, text height=1.25ex, text depth=0.25ex, ampersand replacement=\&] (m) {
                     \& a_{0, 1}     \&          \& \dots \&              \& a_{0, m}     \&     \\
            q_{1, 0} \&              \& q_{1, 1} \& \dots \& q_{1, m - 1} \&              \& q_{1, m} \\
                     \& a_{1, 1}     \&          \&       \&              \& a_{1, m}     \&     \\
            \vdots \& \vdots       \&          \&       \&              \& \vdots       \& \vdots \\
                     \& a_{n - 1, 1} \&          \&       \&              \& a_{n - 1, m} \&     \\
            q_{n, 0} \&              \& q_{n, 1} \& \dots \& q_{n, m - 1} \&              \& q_{n, m} \\
                     \& a_{n, 1}     \&          \& \dots \&              \& a_{n, m}     \&     \\
          };
          \foreach \j in {1, 5} {
            \foreach \i in {1, 5} {
              \draw[->] let
                \n1 = {int(2+\i)},
                \n2 = {int(1+\j)}
              in
               (m-\n2-\i) -> (m-\n2-\n1);
              \draw[->] let
                \n1 = {int(1+\i)},
                \n2 = {int(2+\j)}
              in
                (m-\j-\n1) -> (m-\n2-\n1);
            };
          };
        \end{tikzpicture}%
        \caption{Multiple crosses combined in one diagram.}\label{sfig:combinedCrossDiagram}%
      \end{subfigure}\hfill%
      \begin{subfigure}[t]{0.2\linewidth}%
        \centering%
        \begin{tikzpicture}[baseline=(m-2-1.base)]
          \matrix[matrix of math nodes, text height=1.25ex, text depth=0.25ex, ampersand replacement=\&] (m) {
                   \& u \& \\
            \bm{p} \&   \& \bm{q} \\
                   \& v \& \\
          };
          \foreach \i in {1} {
            \draw[->] let
              \n1 = {int(2+\i)}
            in
              (m-2-\i) -> (m-2-\n1);
            \draw[->] let
              \n1 = {int(1+\i)}
            in
              (m-1-\n1) -> (m-3-\n1);
          };
        \end{tikzpicture}%
        \caption{Abbreviated cross diagram.}\label{sfig:abbreviatedCrossDiagram}%
      \end{subfigure}
      \caption{Combined and abbreviated cross diagrams.}
    \end{figure}
    Another somewhat graphical tool that we will make heavy use of are \emph{cross diagrams}. Here, a cross diagram as given in \autoref{sfig:singleCrossDiagram} indicates the existence of a transition $\trans{p}{a}{b}{q}$ in the automaton. Cross diagrams can be stacked together in order to create lager ones. For example, the diagram in \autoref{sfig:combinedCrossDiagram} indicates the existence of the transition $\trans{q_{i, j - 1}}{a_{i - 1, j}}{a_{i, j}}{q_{i, j}}$ for all $0 < i \leq n$ and $0 < j \leq m$. When combining cross diagrams, we will sometimes omit unnecessary states and letters. Additionally, we will also abbreviate them: for example, if we let $\bm{p} = q_{n, 0} \dots q_{1, 0}$, $u = a_{0, 1} \dots a_{0, m}$, $v = a_{n, 1} \dots a_{n, m}$ and $\bm{q} = q_{n, m} \dots q_{1, m}$, the cross diagram in \autoref{sfig:abbreviatedCrossDiagram} is an abbreviation of the cross diagram in \autoref{sfig:combinedCrossDiagram}. It is important here to note the order we write the state sequences in: in our example, $q_{1, 0}$ is the first state in the top left of the cross diagram but it is the rightmost state in the sequence $\bm{p}$.
    This order will later be more natural as we will define a left action based on cross diagrams.
    
    \begin{remark}
      It may be helpful for the reader to observe that the individual rows in a cross diagram constitutes runs\footnote{See, e.\,g.\ \cite{waechter2024word} for some introduction to automata theory and \cite[p.~275]{waechter2024word} for a precise definition of a run.} of the automaton where the output of the previous one is the input for the next one. For example, the $i$-th row of \autoref{sfig:combinedCrossDiagram} belongs to the run
      \begin{center}
        \begin{tikzpicture}[auto, shorten >=1pt, >=latex, baseline=(q0.base)]
          \node (q0) {$q_{i, 0}$};
          \node[right=2cm of q0] (q1) {$q_{i, 1}$};
          \node[right=2cm of q1] (dots) {$q_{i, 2} \ \dots\ q_{i, m - 1}$};
          \node[right=2cm of dots] (qn) {$q_{i, m}$};
          
          \path[->] (q0) edge node {\scriptsize\ensuremath{a_{i - 1, 1} / a_{i, 1}}} (q1)
                    (q1) edge node {\scriptsize\ensuremath{a_{i - 1, 2} / a_{i, 2}}} (dots)
                    (dots) edge node {\scriptsize\ensuremath{a_{i - 1, m} / a_{i, m}}} (qn)
          ;
        \end{tikzpicture}
      \end{center}
      in the automaton.
    \end{remark}
    
    An automaton $\mathcal{T} = (Q, \Sigma, \delta)$ is called \emph{complete and deterministic} if, for every $p \in Q$ and every $a \in \Sigma$, there is exactly one $q \in Q$ and exactly one $b \in \Sigma$ such that the cross diagram in \autoref{sfig:singleCrossDiagram} holds (i.\,e.\ in every state $p$ and for every letter $a \in \Sigma$, there is exactly one transition starting in $p$ with input $a$). We call such an automaton a \emph{complete \SAut} (as they naturally generate semigroups).
    
    \paragraph*{Subautomata.}
    An automaton $\mathcal{S} = (P, \Sigma, \sigma)$ is a \emph{subautomaton} of another automaton $\mathcal{T} = (Q, \Gamma, \delta)$ if $P \subseteq Q$, $\Sigma \subseteq \Gamma$ and $\sigma \subseteq \delta$. In this case, any cross diagram of $\mathcal{S}$ is also a (valid) cross diagram of $\mathcal{T}$.
    
    \paragraph*{Automaton Semigroups and Monoids.}
    Let $\mathcal{T} = (Q, \Sigma, \delta)$ be a complete \SAut. By induction, there is exactly one $v \in \Sigma^+$ and exactly one $\bm{q} \in Q^+$ for every $\bm{p} \in Q^+$ and $u \in \Sigma^+$ such that the cross diagram in \autoref{sfig:abbreviatedCrossDiagram} holds (with respect to $\mathcal{T}$). This allows us to define a left action of $Q^+$ on $\Sigma^+$ by letting $\bm{p} \circ u = v$ and to define a right action of $\Sigma^+$ on $Q^+$, called the \emph{dual action}, by letting $\bm{p} \cdot u = \bm{q}$. The reader may verify that this indeed defines well-defined actions by the way cross diagrams work.
    We may extend these into an action of $Q^*$ on $\Sigma^*$ and an action of $\Sigma^*$ on $Q^*$ by letting $\varepsilon \circ u = u$ for all $u \in \Sigma^*$, $\bm{p} \circ \varepsilon = \varepsilon$ for all $\bm{p} \in Q^*$, $\varepsilon \cdot u = \varepsilon$ again for all $u \in \Sigma^*$ and, finally, $\bm{p} \cdot \varepsilon = \bm{p}$ for (again) all $\bm{p} \in Q^*$.
    
    By the way cross diagrams work, there is an interaction between the two actions: for all $\bm{p}, \bm{q} \in Q^*$ and all $u, v \in \Sigma^*$, we have $\bm{p} \circ uv = (\bm{p} \circ u) [(\bm{p} \cdot u) \circ v]$ and $\bm{q}\bm{p} \cdot u = [\bm{q} \cdot (\bm{p} \circ u)] (\bm{p} \cdot u)$.
    
    The action $\bm{p} \circ u$ allows us to define the congruence ${=_\mathcal{T}} \subseteq Q^* \times Q^*$ as its kernel, i.\,e.
    \[
      \bm{p} =_{\mathcal{T}} \bm{q} \iff \forall u \in \Sigma^* : \bm{p} \circ u = \bm{q} \circ u \text{.}
    \]
    We denote the congruence class of $\bm{p} \in Q^*$ with respect to $=_\mathcal{T}$ by $[ \bm{p} ]_\mathcal{T}$. The set
    \[
      \mathscr{M}(\mathcal{T}) = Q^* / {=_\mathcal{T}}
    \]
    of these congruence classes forms a monoid, which is called the \emph{monoid generated} by $\mathcal{T}$.
    In other words, it is the quotient of $Q^*$ by the kernel $=_{\mathcal{T}}$, which yields a faithful action of $\mathscr{M}(\mathcal{T})$ on $\Sigma^*$.
    Note that $\varepsilon$ acts as the identity on all $u \in \Sigma^*$ and the class of $\varepsilon$, thus, forms the neutral element of $\mathscr{M}(\mathcal{T})$.
    A monoid arising in this way is called a \emph{complete automaton monoid}.
    
    Similarly, the \emph{semigroup generated} by $\mathcal{T}$ is the semigroup
    \[
      \mathscr{S}(\mathcal{T}) = Q^+ / {=_{\mathcal{T}}}
    \]
    and any semigroup arising this way is a \emph{complete automaton semigroup}. Note that the monoid and the semigroup generated by a complete \SAut coincide if and only if there is a non-empty state sequence acting as the identity.

    \begin{remark}
      We only consider complete \SAuta in this work but will make this explicit by talking about complete \SAuta and complete automaton semigroups and monoids. In the literature, these objects are often simply called \enquote{automaton semigroups} (the term \enquote{automaton monoid} is less common). This is a convention that we could also follow here but choose not to since the concepts generalize naturally also to non-complete automata, yielding (partial) automaton semigroups and monoids. It is not known whether the two classes coincide (and we refer the reader to \cite{structurePart} for more details on this question and the general concepts).
    \end{remark}
    
    \begin{remark}
      There is a subtle difference between an automaton monoid and an automaton semigroup which happens to be a monoid. In the latter, the neutral element must not necessarily belong to a state sequence acting as the identity. In fact, it is not known whether the two classes coincide (which contrasts the situations with automaton groups where it is known that every automaton semigroup that happens to be a group is an automaton group, see \cite[Proposition~3.1]{cain2009automaton} for the required construction).
    \end{remark}
    
    \paragraph*{Free Semigroups (Monoids) as Automaton Semigroups (Monoids).}
    As examples of complete automaton semigroups and monoids, we will next look at how to generate free semigroups and monoids.
    \begin{example}[The Adding Machine]\label{ex:addingMachine}
      Let
      $\mathcal{T} = (\{ q, \id \}, \{ 0, 1 \}, \delta)$
      denote the automaton given by
      
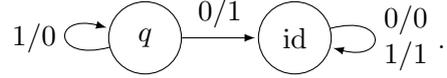
\begin{center}
        \begin{tikzpicture}[auto, shorten >=1pt, >=latex, baseline=(id.base)]
          \node[state] (q) {$q$};
          \node[state, right=of q] (id) {$\id$};
          \draw[->] (q) edge[loop left] node {$1/0$} (q)
                        edge node {$0/1$} (id)
                    (id) edge[loop right] node[align=left] {$0/0$\\$1/1$} (id);
        \end{tikzpicture}.
      \end{center}
      It is clearly a complete \SAut (and known as the \emph{adding machine}).
      
      The state $\id$ clearly acts as the identity on $\{ 0, 1 \}^*$ (justifying its name) and the action of $q$ is best understood by looking at an example:
      \begin{center}
        \begin{tikzpicture}[baseline=(m-8-1.base)]
          \matrix[matrix of math nodes, text height=1.25ex, text depth=0.25ex] (m) {
              & 0 &     & 0 &     & 0 &     \\
            q &   & \id &   & \id &   & \id \\
              & 1 &     & 0 &     & 0 &     \\
            q &   &  q  &   & \id &   & \id \\
              & 0 &     & 1 &     & 0 &     \\
            q &   & \id &   & \id &   & \id \\
              & 1 &     & 1 &     & 0 &     \\
            q &   &  q  &   &  q  &   & \id \\
              & 0 &     & 0 &     & 1 &     \\
          };
          \foreach \j in {1, 3, 5, 7} {
            \foreach \i in {1, 3, 5} {
              \draw[->] let
                \n1 = {int(2+\i)},
                \n2 = {int(1+\j)}
              in
                (m-\n2-\i) -> (m-\n2-\n1);
              \draw[->] let
                \n1 = {int(1+\i)},
                \n2 = {int(2+\j)}
              in
                (m-\j-\n1) -> (m-\n2-\n1);
            };
          };
        \end{tikzpicture}
      \end{center}
      Looking at the input and output words, we can derive that the action of $q$ can be considered as an increment of a binary number (in reverse/with the least significant bit first). In particular, we have that the actions of all $q^i$ are pair-wise different and we obtain that $\mathscr{M}(\mathcal{T})$ is isomorphic to $q^*$ (where $q^0 = \varepsilon$ belongs to $\id$). Since we have $\id =_\mathcal{T} \varepsilon$, the semigroup generated by $\mathcal{T}$ is the same as the monoid generated by it (i.\,e.\ $\mathscr{S}(\mathcal{T}) = \mathscr{M}(\mathcal{T}) \simeq q^*$).
    \end{example}
    
    The adding machine from \autoref{ex:addingMachine} shows that the free \textbf{monoid} of rank one is a complete automaton semigroup and a complete automaton monoid. The free \textbf{semigroup} of rank one, on the other hand, is neither \cite[Proposition~4.3]{cain2009automaton} (see also \cite[Theorem~15]{brough2017automaton}, \cite[Theorem~19]{structurePart} and \cite[Theorem~1.2.1.4]{waechter2020automaton}).
    
    However, free semigroups of higher rank (and their monoid counter-parts) are indeed complete automaton semigroups.
    We will present the construction from \cite[Theorem~4.1]{silva2005class} (or \cite[Proposition 4.1]{cain2009automaton}) for this next.
    \begin{figure}\centering
      \begin{tikzpicture}[auto, shorten >=1pt, >=latex, baseline=(id.base)]
        \node[state] (a) {$a$};
        \node[state, right=of a] (b) {$b$};
        \draw[->] (a) edge[loop left] node {$a/a$} (a)
                      edge[bend left] node {$b/a$} (b)
                  (b) edge[loop right] node {$b/b$} (b)
                      edge[bend left] node {$a/b$} (a);
      \end{tikzpicture}
      \caption{A complete \SAut generating $\{ a, b \}^+$.}\label{fig:freeBinarySemigroup}
    \end{figure}
    \begin{example}\label{ex:freeSemigroups}
      Let $R$ be a finite set with $|R| \geq 2$. Consider the automaton\footnote{The binary case $R = \{ a, b \}$ is depicted in \autoref{fig:freeBinarySemigroup}.} $\mathcal{R} = (R, R, \rho)$ with
      \[
        \rho = \{ \trans{a}{b}{a}{b} \mid a, b \in R \} \text{.}
      \]
      One easily verifies that $\mathcal{R}$ is a complete \SAut and we claim that it generates the semigroup $R^+$. For this, it suffices to show that, for every $\bm{p}, \bm{q} \in R^+$ with $\bm{p} \neq \bm{q}$, there is some $u \in R^*$ with $\bm{p} \circ u \neq \bm{q} \circ u$. We may assume $|\bm{p}| \geq |\bm{q}|$ and there needs to be some $a \in R$ with $\bm{p} \neq \bm{q} a^{|\bm{p}| - |\bm{q}|}$ (we just need to take $a$ different to the last letter of $\bm{p}$ if the lengths differ).
      
      Now, observe that, for all $n \geq 1$ and all $a_1, \dots, a_n, b_1, \dots, b_n \in R$, we have the cross diagram
      \begin{center}
        \begin{tikzpicture}[baseline=(m-4-1.east)]
          \matrix[matrix of math nodes, text height=1.25ex, text depth=0.25ex, ampersand replacement=\&] (m) {
                     \& b_1     \&          \& \dots \&              \& b_n     \&     \\
            a_1 \&              \& b_1 \& \dots \& b_{n - 1} \&              \& b_n \\
                     \& a_1     \&          \&       \&              \& b_{n - 1}     \&     \\
            \vdots \& \vdots       \&          \&       \&              \& \vdots       \& \vdots \\
                     \& a_{n - 1} \&          \&       \&              \& b_1 \&     \\
            a_n \&              \& a_{n - 1} \& \dots \& a_1 \&              \& b_1 \\
                     \& a_n     \&          \& \dots \&              \& a_1     \&     \\
          };
          \foreach \j in {1, 5} {
            \foreach \i in {1, 5} {
              \draw[->] let
                \n1 = {int(2+\i)},
                \n2 = {int(1+\j)}
              in
               (m-\n2-\i) -> (m-\n2-\n1);
              \draw[->] let
                \n1 = {int(1+\i)},
                \n2 = {int(2+\j)}
              in
                (m-\j-\n1) -> (m-\n2-\n1);
            };
          };
        \end{tikzpicture}%
      \end{center}
      by the construction of $\mathcal{R}$. This shows, in particular, $\bm{p} \circ a^{|\bm{p}|} = \bm{p}$ and $\bm{p} \cdot a^{|\bm{p}|} = a^{|\bm{p}|}$. By a similar cross diagram, we obtain $\bm{p} \neq_{\mathcal{R}} \bm{q}$ (since $\bm{q} \circ a^{|\bm{p}|} = (\bm{q} \circ a^{|\bm{q}|}) (a^{|\bm{q}|} \circ a^{|\bm{p}| - |\bm{q}|}) = \bm{q} a^{|\bm{p}| - |\bm{q}|} \neq \bm{p} = \bm{p} \circ a^{|\bm{p}|}$).
      
      This time, there is no non-empty state sequence which acts as the identity and this means that $\mathscr{M}(\mathcal{R})$ is $\mathscr{S}(\mathcal{R})^\idGrp \simeq R^*$, which shows that $R^*$ is a complete automaton monoid. Alternatively, we could also add a new state $\id$ with the transitions $\{ \trans{\id}{a}{a}{\id} \mid a \in R \}$ to obtain the automaton $\mathcal{R}'$. This is again a complete \SAut and we have $\mathscr{S}(\mathcal{R}') = \mathscr{M}(\mathcal{R}') = R^*$, which shows that $R^*$ is also a complete automaton semigroup (in fact, we may use this construction to show that every complete automaton monoid is a complete automaton semigroup).
    \end{example}
    
    The alphabet size for the automaton constructed in \autoref{ex:freeSemigroups} equals the rank of the generated free semigroup. It turns out, however, that already a binary alphabet is sufficient to generate free semigroups (in fact, even free groups) of arbitrary rank.
    
    \begin{proposition}\label{prop:freeComputable}
      On input of a finite set $R$ with $|R| \geq 2$, one may compute a complete \SAut $\mathcal{R} = (R, \{ 0, 1 \}, \rho)$ (i.\,e.\ one with binary alphabet) with $\mathscr{S}(\mathcal{R}) = R^+$ and $\mathscr{M}(\mathcal{R}) = R^*$.
    \end{proposition}
    \begin{figure}%
      \begin{subfigure}{0.5\linewidth}
        \resizebox*{\textwidth}{!}{%
          \begin{tikzpicture}[auto, shorten >=1pt, >=latex, node distance=1.25cm]
            \node[state] (0) {$q_0$};
            \node[state, right=of 0] (1) {$q_1$};
            \node[right=of 1] (ld) {$\dots$};
            \node[state, right=of ld] (r) {$q_r$};
            \node[state, above=of r] (r+1) {$q_{r + 1}$};
            \node[anchor=base] at ($(r+1.base-|ld.base)$) (ud) {$\dots$};
            \node[state, anchor=base] at ($(r+1.base-|1.base)$) (2r) {$q_{2r}$};
            
            \draw[->] (0) edge[loop left] node {$0/1$} (0)
                          edge node[swap] {$1/0$} (1)
                      (1) edge node[align=center, swap] {$0/0$\\$1/1$} (ld)
                      (ld) edge node[align=center, swap] {$0/0$\\$1/1$} (r)
                      (1) edge node[align=center, swap] {$0/0$\\$1/1$} (ld)
                      (r) edge node[align=center, swap] {$0/0$\\$1/1$} (r+1)
                      (r+1) edge node[align=center, swap] {$0/0$\\$1/1$} (ud)
                      (ud) edge node[align=center, swap] {$0/0$\\$1/1$} (2r)
                      (2r) edge node {$0/1$} (1)
                           edge node[sloped] {$1/0$} (0)
            ;
          \end{tikzpicture}}
        \caption{The odd case.}\label{sfig:freeAutOdd}
      \end{subfigure}
      \begin{subfigure}{0.5\linewidth}
        \resizebox*{\textwidth}{!}{%
          \begin{tikzpicture}[auto, shorten >=1pt, >=latex, node distance=1.25cm]
            \node[state] (0) {$q_0$};
            \node[state, right=of 0] (1) {$q_1$};
            \node[right=of 1] (ld) {$\dots$};
            \node[state, right=of ld] (r) {$q_r$};
            \node[state, above=of r] (r+1) {$q_{r + 1}$};
            \node[anchor=base] at ($(r+1.base-|ld.base)$) (ud) {$\dots$};
            \node[state, anchor=base] at ($(r+1.base-|1.base)$) (2r) {$q_{2r}$};
            \node[state, anchor=base, ellipse] at ($(r+1.base-|0.base)$) (2r+1) {$q_{2r +1}$};
            
            \draw[->] (0) edge[loop left] node {$0/1$} (0)
                          edge node[swap] {$1/0$} (1)
                      (1) edge node[align=center, swap] {$0/0$\\$1/1$} (ld)
                      (ld) edge node[align=center, swap] {$0/0$\\$1/1$} (r)
                      (1) edge node[align=center, swap] {$0/0$\\$1/1$} (ld)
                      (r) edge node[align=center, swap] {$0/0$\\$1/1$} (r+1)
                      (r+1) edge node[align=center, swap] {$0/0$\\$1/1$} (ud)
                      (ud) edge node[align=center, swap] {$0/0$\\$1/1$} (2r)
                      (2r) edge node[align=center, swap] {$0/1$\\$1/0$} (2r+1)
                      (2r+1) edge node[sloped] {$0/1$} (1)
                             edge node[swap] {$1/0$} (0)
            ;
          \end{tikzpicture}}
        \caption{The even case.}\label{sfig:freeAutEven}
      \end{subfigure}
      \caption{Automata over binary alphabet generating free (semi)groups.}
    \end{figure}
    \begin{proof}
      For $|R| = 2$, we may simply use the construction from \autoref{ex:freeSemigroups}. For $|R| > 2$, we distinguish between $|R|$ being odd and $|R|$ being even.
      
      In the odd case, we assume $R = \{ q_0, \dots, q_{2r} \}$ for some $r \geq 1$ and let (compare to \autoref{sfig:freeAutOdd})
      \begin{align*}
        \rho ={}& \left\{ \trans{q_0}{0}{1}{q_0}, \trans{q_0}{1}{0}{q_1} \right\} \\
        {}\cup{}& \left\{ \trans{q_i}{0}{0}{q_{i + 1}}, \trans{q_i}{1}{1}{q_{i + 1}} \middle| 1 \leq i < 2r \right\} \\
        {}\cup{}& \left\{ \trans{q_{2r}}{0}{1}{q_{1}}, \trans{q_{2r}}{1}{0}{q_{0}} \right\} \text{.}
      \end{align*}
      This (bi-reversible) automaton indeed generates a free semigroup and a free monoid of rank $|R| = 2r + 1$ \cite[Theorem~4.10]{vorobets2010series} (in fact, it even generates a free group). However, the proof is rather involved.
      
      In the even case, we assume $R = \{ q_0, \dots, q_{2r + 1} \}$ for some $r \geq 1$ and let (compare to \autoref{sfig:freeAutEven})
      \begin{align*}
        \rho ={}& \left\{ \trans{q_0}{0}{1}{q_0}, \trans{q_0}{1}{0}{q_1} \right\} \\
        {}\cup{}& \left\{ \trans{q_i}{0}{0}{q_{i + 1}}, \trans{q_i}{1}{1}{q_{i + 1}} \middle| 1 \leq i < 2r \right\} \\
        {}\cup{}& \left\{ \trans{q_{2r}}{0}{1}{q_{2r + 1}}, \trans{q_{2r}}{1}{0}{q_{2r + 1}} \right\} \\
        {}\cup{}& \left\{ \trans{q_{2r + 1}}{0}{1}{q_{1}}, \trans{q_{2r + 1}}{1}{0}{q_{0}} \right\} \text{.}
      \end{align*}
      Again, this automaton generates a free semigroup and monoid (and, again, even a group) \cite[Theorem~4.9]{steinberg2011automata}.
    \end{proof}

    \paragraph*{Automaton Operations.}
    For the following, it will be convenient to introduce some automaton constructions.
    First, the \emph{union} of two automata $\mathcal{T}_1 = (Q_1, \Sigma_1, \delta_1)$ and $\mathcal{T}_2 = (Q_2, \Sigma_2, \delta_2)$ is the automaton
    \[
      \mathcal{T}_1 \cup \mathcal{T}_2 = (Q_1 \cup Q_2, \Sigma_1 \cup \Sigma_2, \delta_1 \cup \delta_2) \text{.}
    \]
    Note that, if $\mathcal{T}_1$ and $\mathcal{T}_2$ are both complete \SAuta with non-intersecting state sets ($Q_1 \cap Q_2 = \emptyset$) but a common alphabet $\Sigma_1 = \Sigma_2$, then their union $\mathcal{T}_1 \cup \mathcal{T}_2$ is also a complete \SAut (which allows us, for example, to consider the semigroup $\mathscr{S}(\mathcal{T}_1 \cup \mathcal{T}_2)$). Similarly, the union of two complete \SAuta with the same state set but disjoint alphabets is again a complete \SAut. This operation basically adds the transitions of $\mathcal{T}_2$ to the existing transitions of $\mathcal{T}_1$.
    
    Next, there is the \emph{composition} of two automata $\mathcal{T}_2 = (Q_2, \Sigma, \delta_2)$ and $\mathcal{T}_1 = (Q_1, \Sigma, \delta_1)$ over a common alphabet $\Sigma$, which is the automaton
    \begin{align*}
      \mathcal{T}_2 \circ \mathcal{T}_1 &= (Q_2 Q_1, \Sigma, \delta_2 \circ \delta_1)
    \intertext{with the transitions}
      \delta_2 \circ \delta_1 &= \left\{ \trans{p_2 p_1}{a}{c}{q_2 q_1} \middle| \exists b \in \Sigma : \trans{p_1}{a}{b}{q_1} \in \delta_1 \text{ and } \trans{p_2}{b}{c}{q_2} \in \delta_2 \right\}
    \end{align*}
    (where $Q_2 Q_1 = \{ q_2 q_1 \mid q_1 \in Q_1, q_2 \in Q_2 \}$ is the cartesian product of $Q_2$ and $Q_1$).
    If $\mathcal{T}_2$ and $\mathcal{T}_1$ are complete \SAuta, also their composition is.
    
    The \emph{$k$-th power} $\mathcal{T}^k$ of an automaton $\mathcal{T}$ is the $k$-fold composition of $\mathcal{T}$ with itself. Here, it is important to point out that the $k$-th power of an automaton is computable and that, if $\mathcal{T}$ is a complete \SAut (which means that $\mathcal{T}^k$ is also one), then the action of some $\bm{p} \in Q^*$ of length $|\bm{p}| = k$ seen as a state of $\mathcal{T}^k$ is the same as the action of $\bm{p}$ seen as a state sequence over $\mathcal{T}$. Since an analogous statement about the dual action also holds, the notations $\bm{p} \circ u$ and $\bm{p} \cdot u$ remain unambiguous. This also shows that we have $\mathscr{S}(\mathcal{T}) = \mathscr{S}(\mathcal{T} \cup \mathcal{T}^k)$ for all $k \geq 1$, which is usually used to ensure that any fixed state sequence $\bm{p} \in Q^+$ may be assumed to be congruent to a single state under $=_{\mathcal{T}}$ (i.\,e.\ equal in the semigroup or monoid; this is basically the same as considering $\bm{p}$ as an additional generator).
    
    Finally, it will sometimes be easier to consider the \emph{dual} of an automaton $\mathcal{T} = (Q, \Sigma, \delta)$. It is the automaton $\partial \mathcal{T} = (\Sigma, Q, \partial \delta)$ with
    \[
      \partial \delta = \left\{ \trans{a}{p}{q}{b} \middle| \trans{p}{a}{b}{q} \in \delta \right\}
    \]
    (i.\,e.\ we swap the roles of the states $Q$ and the letters $\Sigma$). Clearly, the dual of a complete \SAut is again a complete \SAut.
    
    The dual automaton can make it sometimes more accessible to understand how a letter is transformed by a state sequence: we just have to follow a path in the graphical representation of the dual automaton. For example, from \autoref{sfig:T2dual}, it is obvious that the only way for $\bm{p} \circ \alpha = \bm{q} \circ \beta$ to hold is for both of them to be equal to $f$. This is not immediately clear from the original transitions depicted in \autoref{sfig:T2}.

    \paragraph*{The Freeness Problem for Automaton Semigroups and Monoids.}
    To check whether a given automaton generates a free semigroup/monoid/group is an important open problem in the algorithmic theory of automaton structures \cite[7.2 b)]{grigorchuk2000automata}. Formally, the freeness problem for automaton semigroups is the problem \DecProblem{Semigroup Freeness}
    \gnewcommand{\SemigroupFreeness}{%
      \problem{
        a (complete) \SAut $\mathcal{T}$
      }{
        is $\mathscr{S}(\mathcal{T})$ a free semigroup?
      }\noindent
    }\SemigroupFreeness
    and the freeness problem for automaton monoids is the analogous problem \DecProblem{Monoid Freeness}
    \gnewcommand{\MonoidFreeness}{%
      \problem{
        a (complete) \SAut $\mathcal{T}$
      }{
        is $\mathscr{M}(\mathcal{T})$ a free monoid?
      }\noindent
    }\MonoidFreeness
    
    A very related problem is to check whether a given automaton is a free presentation, i.\,e.\ whether the automaton generates a free semigroup or monoid where the state set forms a basis. We will only consider the monoid case for this problem and allow one of the states to represent the neutral element. Let \DecProblem{Free Monoid Presentation} be the problem:
    \gnewcommand{\FreeMonoidPresentation}{%
      \problem{
        a (complete) \SAut $\mathcal{T} = (Q, \Sigma, \delta)$ with\newline
        a dedicated state $e \in Q$ acting as the identity%
      }{
        is $\mathscr{M}(\mathcal{T}) \simeq (Q \setminus \{ e \})^*$?
      }\noindent
    }\FreeMonoidPresentation
    Note that $\mathscr{M}(\mathcal{T})$ is isomorphic to $(Q \setminus \{ e \})^*$ if and only $[ q ]_\mathcal{T} \mapsto q$ for all $q \in Q \setminus \{ e \}$ and $[ e ]_\mathcal{T} \mapsto \varepsilon$ induces a well-defined isomorphism.
    
    \paragraph*{Adding Free Generators.}
    For our results, we will need to add new free generators to existing automaton semigroups $S$ computationally (in the sense that we do not change the behavior of existing state sequences but add a new state $q$ such that the new automaton generates the (semigroup) free product $S \star q^+$). More precisely, we will use the following statement, which follows from the construction used for \cite[Theorem~6]{brough2023automaton}/\cite[Corollary~7]{brough2023automaton}.
    \begin{proposition}\label{prop:freeProduct}
      We may compute:
      \par\vspace{0.125cm plus 0.1cm minus 0.05cm}\begin{tabularx}{\textwidth-2\parindent}{lX}%
        \textnormal{\textbf{Input:}}&
          two complete \SAuta $\mathcal{S}_1 = (P_1, \Sigma_1, \delta_1)$ and $\mathcal{S}_2 = (P_2, \Sigma_2, \delta_2)$ and\newline
          a function $f\colon \Sigma_2 \to \Sigma_1^+$ that is
          guaranteed to extend into a homomorphism $\mathscr{S}(\mathcal{S}_2) \to \mathscr{S}(\mathcal{S}_1)$
        \\%
        \textnormal{\textbf{Output:}}&
          a complete \SAut $\mathcal{T} = (Q, \Gamma, \delta)$ with $Q = P_1 \uplus P_2$ such that the identity on $Q$ extends into a well-defined isomorphism $\mathscr{S}(\mathcal{T}) \to \mathscr{S}(\mathcal{S}_1) \star \mathscr{S}(\mathcal{S}_2)$ (where the free product is that of semigroups).
        \\%
      \end{tabularx}\vspace{0.125cm plus 0.1cm minus 0.05cm}\par\noindent
      Furthermore, we may assume $|\Gamma| = 3 + 3|\Sigma_1| + 3|\Sigma_2|$ if no state sequence from $P_1^+ \uplus P_2^+$ acts like the identity.
    \end{proposition}
  \end{section}
  
  \begin{section}{The Freeness Problem for Automaton Semigroups and Monoids}\label{sct:semigroupFreeness}
    We reduce Post's Correspondence Problem\footnote{Please note that Post's original statement of the problem \cite{post1946variant} is equivalent to ours. In particular, we may assume that $\varphi(i)$ and $\psi(i)$ are non-empty for all $i \in I$.} \DecProblem{PCP}
    \problem[an alphabet $\Lambda$]{
      homomorphisms $\varphi, \psi: I = \{ 1, \dots, n \} \to \Lambda^+$
    }{
      $\exists \bm{i} \in I^+: \varphi(\bm{i}) = \psi(\bm{i})$?
    }\noindent
    to (the complement of) \DecProblem{Semigroup Freeness}. For this, we fix an instance $\varphi, \psi, I$ for \DecProblem{PCP}\footnote{%
      It is worth mentioning that there are results showing that \DecProblem{PCP} remains undecidable if one restricts the cardinality of $I$ and/or $\Lambda$ (notably, \cite{neary2015undecidability} restricts them to $|I| = 5$ and $|\Lambda| = 2$). Note that we may only allow non-empty entries for our construction, however.}
    over an alphabet $\Lambda$ and describe how to map it to an \SAut $\mathcal{T} = (Q, \Sigma, \delta)$ in such a way that $\mathcal{T}$ can be computed and the \DecProblem{PCP} instance has a solution if and only if $\mathscr{S}(\mathcal{T})$ is \textbf{not} a free semigroup.
    
    Starting from the free semigroup, we will construct $\mathcal{T}$ (in steps) such that the semigroup has a relation $\#_1 \bm{i} \#_1 =_{\mathcal{T}} \#_1 \bm{i} \#_2$ for $\bm{i} \in I^+$ if and only if $\bm{i}$ belongs to a \DecProblem{PCP} solution (if there is no solution, $\mathscr{S}(\mathcal{T})$ is free).
    Throughout this process, the reader may find it convenient to refer to \autoref{tbl:semigroupSymbols} for the various symbols we are going to use.
    
    The rough idea is to add an input symbol $\iota$ whose dual action turns $\bm{i} \#_1$ into $\varphi(\bm{i})$ and $\bm{i} \#_2$ into $\psi(\bm{i})$.
    But we also have to be careful not to introduce any unwanted relations and to keep the underlying free semigroup structure intact.
    
    Without loss of generality, we may assume $|I| = n \geq 1$, $|\Lambda| \geq 2$ and $I \cap \Lambda = \emptyset$. In the following, we let $L = \max\{ |\varphi(i)|, |\psi(i)| \mid i \in I \}$, $\hat{\Lambda} = \cup_{\ell = 1}^L \Lambda^\ell$, $R = \Lambda \cup I$ and $\hat{R} = \hat{\Lambda} \cup I$.
    
    Throughout this section, the reader may find it convenient to refer to \autoref{tbl:semigroupSymbols} for a summary of (most of) the symbols we define.
    
    \paragraph{Definition of the Automaton $\hat{\mathcal{R}}$.}
    First, we compute a complete \SAut $\hat{\mathcal{R}}$ with state set $\hat{R}$ generating the free semigroup over $R$:
    \begin{proposition}\label{prop:adjoinFreeGenerator}
      On input of $I$, $\Lambda$ and $L$, one can compute a complete \SAut $\hat{\mathcal{R}} = (\hat{R}, \Gamma, \rho)$ with state set $\hat{R} = \hat{\Lambda} \cup I$ (for $\hat{\Lambda} = \cup_{\ell = 1}^L \Lambda^\ell$), alphabet size $|\Gamma| = 15$ and $\mathscr{S}(\hat{\mathcal{R}}) \simeq R^+ = (\Lambda \cup I)^+$ (where the isomorphism is given by $\hat{\lambda} \mapsto \hat{\lambda}$ for all $\hat{\lambda} \in \hat{\Lambda}$ and $i \mapsto i$ for all $i \in I$).
    \end{proposition}
    \begin{proof}
      First, we compute a complete \SAut $\mathcal{R}_1 = (\Lambda, \{ 0, 1 \}, \rho_1)$ generating the free semigroup $\Lambda^+$ (using \autoref{prop:freeComputable}). Note that only the empty state sequence can act like the identity in it.
      Then, we compute the first $L$-many powers $\mathcal{R}_1^1, \dots, \mathcal{R}_1^L$ and take their union. This results in $\hat{\mathcal{R}_1} = (\hat{\Lambda}, \{ 0, 1 \}, \hat{\rho}_1)$ for $\hat{\Lambda} = \cup_{\ell = 1}^L \Lambda^\ell$ (since neither the power nor the union construction changes the alphabet) with $\mathscr{S}(\hat{\mathcal{R}}_1) \simeq \Lambda^+$ (where an isomorphism is induced by $\hat{\Lambda} \ni \hat{\lambda} \mapsto \hat{\lambda} \in \Lambda^+$). We also still have that only the empty state sequence acts like the identity
      
      Next, we compute a complete \SAut $\mathcal{R}_2 = (I, \{ 0, 1 \}, \rho_2)$ that generates the free semigroup $I^+$ (again using \autoref{prop:freeComputable}; with only the empty state sequence acting like the identity).
      
      Finally, we use \autoref{prop:freeProduct} with $\hat{\mathcal{R}}_1$ and $\mathcal{R}_2$ (as well as $f(i) = \lambda_0$ for some $\lambda_0 \in \Lambda$ and all $i \in I$) to compute the sought automaton $\hat{\mathcal{R}} = (\hat{\Lambda} \uplus I, \Gamma, \rho)$ with $\mathscr{S}(\hat{\mathcal{R}}) = \mathscr{S}(\hat{\mathcal{R}}_1) \star \mathscr{S}(\mathcal{R}_2) \simeq \Lambda^+ \star I^+ = (\Lambda \cup I)^+$.
    \end{proof}

    The states of $\hat{\mathcal{R}}$ in $\hat{R}$ do not form a basis of the free semigroup. To simplify working with this fact, we make the following definition(s).
    \begin{definition}[natural projection]\label{def:naturalProjection}
      There is a natural projection $\pi: \hat{\Lambda}^* \to \Lambda^*$ where $\hat{\Lambda} = \bigcup_{\ell = 1}^L \Lambda^{\ell}$, which interprets a letter $\hat{\lambda} \in \hat{\Lambda}$ as the corresponding word over $\Lambda$. We extend this projection into a homomorphism $\pi: \hat{R}^* \to R^*$ by setting $\pi(i) = i$ for all $i \in I$.
      
      We say that two elements $\hat{\bm{r}}_1, \hat{\bm{r}}_2 \in \hat{R}^*$ are \emph{$R$-equivalent} and write $\hat{\bm{r}}_1 =_R \hat{\bm{r}}_2$ if they have the same image under $\pi$, i.\,e.\ we have $\hat{\bm{r}}_1 =_R \hat{\bm{r}}_2 \iff \pi(\hat{\bm{r}}_1) = \pi(\hat{\bm{r}}_2)$.
      
      Finally, we define $| \hat{\bm{r}} |_R$ for $\hat{\bm{r}} \in \hat{R}^*$ as the length of $\hat{\bm{r}}$ under $\pi$, i.\,e.\ we let $| \hat{\bm{r}} |_R = |\pi(\hat{\bm{r}})|$.
    \end{definition}
    \noindent{}Note that we have $\hat{\bm{r}}_1 =_R \hat{\bm{r}}_2$ if and only if $\hat{\bm{r}}_1 =_{\hat{\mathcal{R}}} \hat{\bm{r}}_2$ for all $\hat{\bm{r}}_1, \hat{\bm{r}}_2 \in \hat{R}^*$ since we have $\mathscr{S}(\hat{\mathcal{R}}) \simeq R^+$.
    
    \paragraph{Definition of the Automaton $\mathcal{S}$.}
    We use the automaton $\hat{\mathcal{R}} = (\hat{R}, \Gamma, \rho)$ as a building block for our target automaton $\mathcal{T} = (Q, \Sigma, \delta)$ for the reduction. We fix some arbitrary element $\lambda_\# \in \Lambda \subseteq \hat{R}$. To compute $\mathcal{S}$ from $\hat{\mathcal{R}}$, we duplicate the state $\lambda_\#$ twice and call these two copies $\#_1$ and $\#_2$. Formally, we have $\mathcal{S} = (Q, \Gamma, \sigma)$ where $Q = \hat{R} \uplus \{ \#_1, \#_2 \}$ for the new symbols $\#_1$ and $\#_2$ and
    \[
      \sigma = \rho \cup \{ \trans{\#_1}{c}{d}{q}, \trans{\#_2}{c}{d}{q} \mid \trans{\lambda_\#}{c}{d}{q} \in \rho \} \text{.}
    \]
    
    Note that the new states $\#_1$ and $\#_2$ act in the same way as $\lambda_\#$. Accordingly, we have $\mathscr{S}(\mathcal{S}) = \mathscr{S}(\hat{\mathcal{R}}) \simeq R^+$.
    
    \paragraph{Definition of the Automaton $\mathcal{T}$.}
    The next step is to fix another $\lambda_R \in \Lambda \subseteq Q$ arbitrarily but different to $\lambda_\#$. Then, we add two new letters $a, b$ to the alphabet and the transitions depicted in \autoref{sfig:T1transitions}. This way, we obtain the complete \SAut $\mathcal{T}_1 = (Q, \Gamma \cup \{ a, b \}, \delta_1)$. Note, in particular, that we have the transitions $\trans{\lambda_{\#}}{a}{a}{\lambda_R}$ and the self-loops $\trans{\lambda_R^\ell}{a}{a}{\lambda_R^\ell}$ for all $1 \leq \ell \leq L$ in $\mathcal{T}_1$.
    \begin{figure}\centering
      \subcaptionbox{new transitions\label{sfig:T1transitions}}{%
        \begin{tikzpicture}[auto, shorten >=1pt, >=latex, baseline=(s.base)]
          \node[state] (r) {$\hat{r}$};
          \node[state, right=of r, dashed] (lambda) {$\lambda_R^{|\hat{r}|_R}$};
          \node[state, below=1.5cm of r.center, anchor=center] (s) {$\#_x$};
          \node[state, below=1.5cm of lambda.center, anchor=center, dashed] (lambdas) {$\lambda_\#$};
          
          \path[->] (r) edge[loop left] node {$b / b$} (r)
                        edge node {$a/a$} (lambda)
                    (s) edge[loop left] node {$b / b$} (s)
                        edge node[swap] {$a/b$} (lambdas)
                        ;
        \end{tikzpicture}
      }
      \subcaptionbox{dual transitions\label{sfig:T1dual}}{%
        \begin{tikzpicture}[auto, shorten >=1pt, >=latex, baseline=(a.base)]
          \node[state] (a) {$a$};
          \node[state, right=1.5cm of a] (b) {$b$};
          
          \path[->] (a) edge[loop left] node {$\hat{r}/\lambda_R^{|\hat{r}|_R}$} (a)
                        edge node {$\#_x/\lambda_\#$} (b)
                    (b) edge[loop right] node {$\id_Q$} (b)
          ;
        \end{tikzpicture}
      }
      \caption{The new transitions for $\mathcal{T}_1$. The transitions exist for all $\hat{r} \in \hat{R}$ and $x \in \{ 1, 2 \}$. The transitions for dashed states are implicitly defined and $\id_Q$ indicates that we have $q/q$ transitions for all $q \in Q$.}\label{fig:T1transitions}
    \end{figure}
    
    The idea for this part is that we may factorize a state sequence $\bm{q} \in Q^*$ into blocks from $\hat{R}^*$ and symbols $\#_1$ and $\#_2$ and then remove the blocks one after another using the letter $a$. We will explain this precisely later in \autoref{fct:removeBlocks}.
    
    Finally, we let $\mathcal{T} = (Q, \Sigma, \delta) = \mathcal{T}_1 \cup \mathcal{T}_2$ where $\mathcal{T}_2$ is given by \autoref{sfig:T2}.
    \begin{figure}
      \begin{subfigure}{\linewidth}\centering
        \begin{tikzpicture}[auto, shorten >=1pt, >=latex, baseline=(e.base)]
          \node[state, dashed] (e) {$\lambda_\#$};
          \node[state, left=3cm of e] (s1) {$\#_1$};
          \node[state, right=3cm of e] (s2) {$\#_2$};
          
          \path[->] (s1) edge node[align=center] {$\alpha/f_\alpha,\; f_\alpha/f_\alpha$\\$\beta/f_\beta,\; f_\beta/f_\beta$\\$\alpha'/f,\; \beta'/f,\; f/f$}
                      node[below] {$\iota/\alpha$} (e)
                    (s2) edge node[align=center, above] {$\alpha/f_\alpha,\; f_\alpha/f_\alpha$\\$\beta/f_\beta,\; f_\beta/f_\beta$\\$\alpha'/f,\; \beta'/f,\; f/f$}
                    node[below] {$\iota/\beta$} (e)
          ;
        \end{tikzpicture}\\[0.5cm]
        \begin{tikzpicture}[auto, shorten >=1pt, >=latex, baseline=(e.base)]
          \node[state] (i) {$i$};
          \node[state, above left=0.5cm and 2cm of i, dashed] (phi) {$\varphi(i)$};
          \node[state, below left=0.5cm and 2cm of i, dashed] (psi) {$\psi(i)$};
          \node[state, right=3cm of i, dashed] (e) {$\lambda_R$};
          
          \path[->] (i) edge node {$\id_{\{ \iota, f_\alpha, f_\beta, f \}}$} (e)
                        edge node[swap] {$\alpha/\alpha',\; \alpha'/\alpha'$} (phi)
                        edge node {$\beta/\beta',\; \beta'/\beta'$} (psi)
          ;
        \end{tikzpicture}\\
        \begin{tikzpicture}[auto, shorten >=1pt, >=latex, baseline=(e.base)]
          \node[state] (l) {$\hat{\lambda}$};
          \node[state, right=3cm of l, dashed] (e) {$\lambda_R^{|\hat{r}|_R}$};
          
          \path[->] (l) edge node {$\id_{\{ \iota, f_\alpha, f_\beta, f \}}$}
                             node[below, align=center] {$\alpha/f_\alpha,\; \alpha'/f_\alpha$\\ $\beta/f_\beta,\; \beta'/f_\beta$} (e)
          ;
        \end{tikzpicture}
        \caption{Schematic depiction of $\mathcal{T}_2$ over the (new) alphabet $\{ \iota, \alpha, \alpha', f_\alpha, \beta, \beta', f_\beta, f \}$. The transitions exist for all $i \in I$ and $\hat{\lambda} \in \hat{\Lambda}$.}\label{sfig:T2}
      \end{subfigure}\\[0.5cm]
      \begin{subfigure}{\linewidth}\centering
        \begin{tikzpicture}[auto, shorten >=1pt, >=latex, baseline=(i.base)]
          \node[state] (i) {$\iota$};
          \node[state, above right=0.5cm and 2cm of i] (a) {$\alpha$};
          \node[state, right=2cm of a] (a') {$\alpha'$};
          \node[state, above=of a'] (fa) {$f_\alpha$};
          \node[state, below right=0.5cm and 2cm of i] (b) {$\beta$};
          \node[state, right=2cm of b] (b') {$\beta'$};
          \node[state, below=of b'] (fb) {$f_\beta$};
          \node[state, below right=0.5cm and 2cm of a'] (f) {$f$};
          
          \path[->] (i) edge[loop left] node {$\hat{r}/\lambda_R^{|\hat{r}|_R}$} (i)
                        edge node {$\#_1/\lambda_\#$} (a)
                        edge node[swap] {$\#_2/\lambda_\#$} (b)
                    (a) edge node[swap] {$i/\varphi(i)$} (a')
                        edge node {$\begin{aligned}
                            \#_x &/ \lambda_\# \\
                            \hat{\lambda} &/ \lambda_R^{|\hat{\lambda}|_R}
                          \end{aligned}$} (fa)
                    (a') edge[loop right] node {$i / \varphi(i)$} (a')
                         edge node[swap] {$\hat{\lambda}/\lambda_R^{|\hat{\lambda}|_R}$} (fa)
                         edge node[swap] {$\#_x/\lambda_\#$} (f)
                    (fa) edge[loop right] node {$\begin{aligned}
                        \#_x &/ \lambda_\# \\
                        \hat{r} &/ \lambda_R^{|\hat{r}|_R}
                      \end{aligned}$} (fa)
                    (b) edge node {$i/\psi(i)$} (b')
                        edge node[swap] {$\begin{aligned}
                            \hat{\lambda} &/ \lambda_R^{|\hat{\lambda}|_R} \\
                            \#_x &/ \lambda_\# \\
                          \end{aligned}$} (fb)
                    (b') edge[loop right] node {$i / \psi(i)$} (b')
                         edge node {$\hat{\lambda}/\lambda_R^{|\hat{\lambda}|_R}$} (fb)
                         edge node {$\#_x/\lambda_\#$} (f)
                    (fb) edge[loop right] node {$\begin{aligned}
                        \hat{r} &/ \lambda_R^{|\hat{r}|_R} \\
                        \#_x &/ \lambda_\#
                      \end{aligned}$} (fb)
                    (f) edge[loop right] node {$\begin{aligned}
                        \#_x &/ \lambda_\# \\
                        \hat{r} &/ \lambda_R^{|\hat{r}|_R}
                      \end{aligned}$} (f)
          ;
        \end{tikzpicture}
        \caption{The dual $\partial \mathcal{T}_2$. The transitions exist for all $i \in I, \hat{r} \in \hat{R}, \hat{\lambda} \in \hat{\Lambda}$ and $x \in \{ 1, 2 \}$.}\label{sfig:T2dual}
      \end{subfigure}
      \caption{The automaton $\mathcal{T}_2$ and its dual.}\label{fig:T2andDual}
    \end{figure}
    Note, in particular, that we have $\varphi(i), \psi(i) \in \bigcup_{\ell = 1}^L \Lambda^\ell = \hat{\Lambda} \subseteq \hat{R}$.
    
    In other words, we obtain $\mathcal{T}$ from $\mathcal{T}_1$ by adding new symbols to the alphabet resulting in $\Sigma = \Gamma \cup \{ a, b \} \cup \{ \iota, \alpha, \alpha', f_\alpha, \beta, \beta', f_\beta, f \}$ and adding the transitions depicted in \autoref{fig:T2andDual} for all $i \in I$ and $\hat{\lambda} \in \hat{\Lambda}$.
    Clearly, $\mathcal{T}$ can be computed and is a complete \SAut.
    \begin{remark}
      We have $|\Gamma| = 15$ (by \autoref{prop:adjoinFreeGenerator}) and, thus, $|\Sigma| = 25$ for the alphabet size of $\mathcal{T}$.
    \end{remark}
    
    \paragraph*{The Role of $a$ and $b$ in $\mathcal{T}$.}
    As already mentioned above, we may use the letter $a$ to remove a block from a certain factorization of a state sequence (the letter $b$ is simply used to ignore remaining parts of the factorization). We will apply this multiple times below and, therefore, state this as its own fact here.
    \begin{fact}\label{fct:removeBlocks}
      Let $\bm{p} \in Q^*$ and factorize it as
      \[
        \bm{p} = (\bm{p}_s \#_{x_s}) \dots (\bm{p}_1 \#_{x_1}) \, \bm{p}_0
      \]
      for $\bm{p}_0, \dots, \bm{p}_s \in \hat{R}^*$ and $x_1, \dots, x_s \in \{ 1, 2 \}$.
      
      Then, for any $1 \leq \mu \leq s$, we have (in $\mathcal{T}$):
      \[
        \bm{p} \cdot a^{\mu} = (\bm{p}_s \#_{x_s}) \ldots (\bm{p}_{\mu + 1} \#_{x_{\mu + 1}}) \, \bm{p}_\mu \, \lambda_{\#} \lambda_R^{\mu - 1 + | \bm{p}_{\mu - 1} \ldots \bm{p}_0 |_R}
      \]
    \end{fact}
    \begin{proof}
      Note that we are only making any statement for $1 \leq s$.
      Write $\bm{p} = \bm{p}' \#_{x_1} \bm{p}_0$ for $\bm{p}' = (\bm{p}_s \#_{x_s}) \dots (\bm{p}_2 \#_{x_2}) \, \bm{p}_1$. Using an index shift by setting $\bm{p}_{\nu}' = \bm{p}_{\nu}$ for $0 \leq \nu < s$ and $x_{\nu}' = x_{\nu + 1}$ for $0 < \nu < s$, we obtain
      \[
        \bm{p}' = (\bm{p}'_{s - 1} \#_{x_{s - 1}}) \dots (\bm{p}'_1 \#_{x_1}) \, \bm{p}'_0
      \]
      and may apply induction and re-substitution (for $\mu > 1$) to obtain
      \begin{align*}
        \bm{p}' \cdot a^{\mu - 1} &= (\bm{p}'_{s - 1} \#_{x_{s - 1}}) \ldots (\bm{p}'_{\mu} \#_{x'_{\mu}}) \, \bm{p}'_{\mu - 1} \, \lambda_{\#} \lambda_R^{\mu - 2 + | \bm{p}'_{\mu - 2} \ldots \bm{p}'_0 |_R} \\
        &= (\bm{p}_s \#_{x_s}) \ldots (\bm{p}_{\mu + 1} \#_{x_{\mu + 1}}) \, \bm{p}_\mu \, \lambda_{\#} \lambda_R^{\mu - 2 + | \bm{p}_{\mu - 1} \ldots \bm{p}_1 |_R} = \bm{p}'' \text{.}
      \end{align*}
      This yields the cross diagram (compare to the transitions in \autoref{fig:T1transitions})
      \begin{center}
        \begin{tikzpicture}[baseline=(m-4-1.base)]
          \matrix[matrix of math nodes, text height=1.25ex, text depth=0.25ex, ampersand replacement=\&] (m) {
              \& a \& \& a^{\mu - 1} \\
            \bm{p}_0 \& \& \lambda_R^{|\bm{p}_0|_R} \& \& \lambda_R^{|\bm{p}_0|_R} \\
              \& a \& \& a^{\mu - 1} \\
            \#_{x_1} \& \& \lambda_{\#} \& \& \lambda_R \\
              \& b \& \& a^{\mu - 1} \\
            \bm{p}' \& \& \bm{p}' \& \& \bm{p}'' \\
              \& b \& \& {} \\
          };
          \foreach \i in {1,3,5} {
            \foreach \j in {1,3} {
              \draw[->] let
                \n1 = {int(2+\i)},
                \n2 = {int(1+\j)}
              in
                (m-\i-\n2) -> (m-\n1-\n2);
              \draw[->] let
                \n1 = {int(1+\i)},
                \n2 = {int(2+\j)}
              in
                (m-\n1-\j) -> (m-\n1-\n2);
            };
          };
        \end{tikzpicture}
      \end{center}
      where the crosses in the first column prove the case $\mu = 1$ immediately and the crosses in the second column only exists for $\mu > 1$. In that case, observe that, as desired,
      \begin{align*}
        \bm{p} \cdot a^{\mu} = \bm{p}'' \lambda_R^{1 + |\bm{p}_0|_R} &= (\bm{p}_s \#_{x_s}) \ldots (\bm{p}_{\mu + 1} \#_{x_{\mu + 1}}) \, \bm{p}_\mu \, \lambda_{\#} \lambda_R^{\mu - 2 + | \bm{p}_{\mu - 1} \ldots \bm{p}_1 |_R} \lambda_R^{1 + |\bm{p}_0|_R} \\
        &= (\bm{p}_s \#_{x_s}) \ldots (\bm{p}_{\mu + 1} \#_{x_{\mu + 1}}) \, \bm{p}_\mu \, \lambda_{\#} \lambda_R^{\mu - 1 + | \bm{p}_{\mu - 1} \ldots \bm{p}_1 \bm{p}_0 |_R}
        \qedhere
      \end{align*}
    \end{proof}
    
    \begin{table}[p]\centering
      \begin{tabular}{r@{}l}\toprule
        \textbf{symbol}\phantom{${}:{}$}& \textbf{usage} \\\midrule
        $\Lambda : {}$ & PCP base alphabet, $|\Lambda| \geq 2$\\
        $I : {}$ & PCP index set, $|I| \geq 1$, $I \cap \Lambda = \emptyset$\\
        $\varphi, \psi : {}$ & $I \to \Lambda^+$ PCP homomorphisms\\
        $L = {}$ & $\max\{ |\varphi(i)|, |\psi(i)| \mid i \in I \}$\\
        $\hat{\Lambda} = {}$ & $\bigcup_{\ell = 1}^L \Lambda^\ell$\\
        $R = {}$ & $\Lambda \cup I$\\
        $\hat{R} = {}$ & $\hat{\Lambda} \cup I : {}$state set of $\mathcal{R}$\\
        $\hat{\mathcal{R}} = {}$ & $(\hat{R}, \Gamma, \rho) : {}$ complete \SAut generating $R^+ = (\Lambda \cup I)^+$\\
        $\rho : {}$ & transition set of $\hat{\mathcal{R}}$\\
        $\Gamma : {} $ & alphabet of $\hat{\mathcal{R}}$ and $\mathcal{S}$, $|\Gamma| = 15$\\
        $\pi : {}$ & $\hat{\Lambda}^* \to \Lambda$, $\hat{R}^* \to R^*$ natural projection with $\pi(i) = i$ for all $i \in I$\\
        $|\hat{\bm{r}}|_R : {}$ & length of $\pi(\hat{\bm{r}})$ for $\hat{\bm{r}} \in \hat{R}^*$\\
        $\lambda_{\#} \in {}$ & $\Lambda \subseteq \hat{R} : {}$arbitrarily chosen element\\
        $\#_1, \#_2 : {}$ & copies of $\lambda_{\#}$\\
        $\mathcal{S} = {}$ & $(Q, \Gamma, \sigma) : {}$complete \SAut, extension of $\hat{\mathcal{R}}$ still generating $R^+$\\
        $Q = {}$ & $\hat{R} \uplus \{ \#_1, \#_2 \} : {}$state set of $\mathcal{S}$ and $\mathcal{T}$\\
        $\sigma : {}$ & transition set of $\mathcal{S}$\\
        $\lambda_R \in {}$ & $\Lambda \subseteq Q : {}$arbitrarily chosen element with $\lambda_R \neq \lambda_{\#}$\\
        $a, b \not\in \Gamma : {}$ & new letters for $\mathcal{T}_1$\\
        $\mathcal{T}_1 = {}$ & $(Q, \Gamma \uplus \{ a, b \}, \delta_1) : {}$complete \SAut, extension of $\mathcal{S}$, see \autoref{fig:T1transitions}\\
        $\delta_1 : {}$ & transition set of $\mathcal{T}_1$, see \autoref{fig:T1transitions}\\
        $\mathcal{T} = {}$ & $(Q, \Sigma, \delta) = \mathcal{T}_1 \cup \mathcal{T}_2 : {}$complete \SAut, result of the reduction\\
        $\mathcal{T}_2 : {}$ & complete \SAut with new transitions for $\mathcal{T}$, see \autoref{fig:T2andDual}\\
        $\Sigma = {}$ & $\Gamma \uplus \{ a, b \} \uplus \{ \iota, \alpha, \alpha', f_\alpha, \beta, \beta', f_\beta, f \} : {}$alphabet of $\mathcal{T}$, $|\Sigma| = 25$\\
        $\pi_{\#} : {}$ & $Q^* \to \{ \#_1, \#_2 \}^*$ homomorphism with $\pi_{\#}(\#_x) = \#_x$ but $\pi_{\#}(\hat{r}) = \varepsilon$ for $\hat{r} \in \hat{R}$\\
        $\pi' :{}$ & $Q^* \to (R \cup \{ \#_1, \#_2 \})^*$ homomorphism extending $\pi$ with $\pi'(\#_x) = \#_x$ for $x \in \{ 1, 2 \}$\\\bottomrule
      \end{tabular}
      \caption{Various symbols in the order of their definition.}\label{tbl:semigroupSymbols}
    \end{table}

    \paragraph{Correctness.}
    It remains to show that the \DecProblem{PCP} instance $\varphi, \psi, I$ has a solution if and only if $\mathscr{S}(\mathcal{T})$ is \textbf{not} a free semigroup. We start with the (easier) \enquote{only if} direction and show that the additional transitions from $\mathcal{T}_1$ and $\mathcal{T}_2$ do not affect the subautomaton $\hat{\mathcal{R}}$: if two state sequences are $R$-equivalent, they are also equal with respect to $\mathcal{T}$.
    \begin{lemma}\label{lem:REquivalentImpliesTEquivalent}
      Let $\hat{\bm{r}}_1, \hat{\bm{r}}_2 \in \hat{R}^*$ with $\hat{\bm{r}}_1 =_R \hat{\bm{r}}_2$. Then, we have $\hat{\bm{r}}_1 =_{\mathcal{T}} \hat{\bm{r}}_2$.
    \end{lemma}
    \begin{proof}
      We can only have $\hat{\bm{r}}_1 =_R \hat{\bm{r}}_2 = \varepsilon$ if $\hat{\bm{r}}_1 = \hat{\bm{r}}_2 = \varepsilon$, which trivially implies $\hat{\bm{r}}_1 =_{\mathcal{T}} \hat{\bm{r}}_2$.
      
      Therefore, assume $\hat{\bm{r}}_1, \hat{\bm{r}}_2 \neq_R \varepsilon$. We show $\hat{\bm{r}}_1 \circ u = \hat{\bm{r}}_2 \circ u$ for all $u \in \Sigma^*$ by induction on $u$. For $u = \varepsilon$, there is nothing to show. Thus, write $u = c u'$ for some $c \in \Sigma = \Gamma \cup  \{ a, b \} \cup \{ \iota, \alpha, \alpha', f_\alpha, \beta, \beta', f_\beta, f \}$ and $u' \in \Sigma^*$. For $c \in \Gamma$ (the alphabet of $\hat{\mathcal{R}}$), recall that we have $\mathscr{S}(\hat{\mathcal{R}}) \simeq R^+$. Therefore, $\hat{\bm{r}}_1 =_R \hat{\bm{r}}_2$ implies $\hat{\bm{r}}_1 =_{\hat{\mathcal{R}}} \hat{\bm{r}}_2$ and we have the cross diagrams
      \vspace*{-\baselineskip}
      \begin{center}
        \begin{tikzpicture}[baseline=(m-2-1.base)]
          \matrix[matrix of math nodes, text height=1.25ex, text depth=0.25ex, ampersand replacement=\&] (m) {
                     \& c \& \\
            \hat{\bm{r}}_1 \&   \& \hat{\bm{r}}_1' \\
                     \& d \& \\
          };
          \foreach \i in {1} {
            \foreach \j in {1} {
              \draw[->] let
                \n1 = {int(2+\i)},
                \n2 = {int(1+\j)}
              in
                (m-\i-\n2) -> (m-\n1-\n2);
              \draw[->] let
                \n1 = {int(1+\i)},
                \n2 = {int(2+\j)}
              in
                (m-\n1-\j) -> (m-\n1-\n2);
            };
          };
        \end{tikzpicture}%
        and
        \begin{tikzpicture}[baseline=(m-2-1.base)]
          \matrix[matrix of math nodes, text height=1.25ex, text depth=0.25ex, ampersand replacement=\&] (m) {
                     \& c \& \\
            \hat{\bm{r}}_2 \&   \& \hat{\bm{r}}_2' \\
                     \& d \& \\
          };
          \foreach \i in {1} {
            \foreach \j in {1} {
              \draw[->] let
                \n1 = {int(2+\i)},
                \n2 = {int(1+\j)}
              in
                (m-\i-\n2) -> (m-\n1-\n2);
              \draw[->] let
                \n1 = {int(1+\i)},
                \n2 = {int(2+\j)}
              in
                (m-\n1-\j) -> (m-\n1-\n2);
            };
          };
        \end{tikzpicture}%
      \end{center}
      in $\hat{\mathcal{R}}$ for some $d \in \Gamma$ and $\hat{\bm{r}}_1', \hat{\bm{r}}_2' \in \hat{R}^+$ with $\hat{\bm{r}}_1' =_{\hat{\mathcal{R}}} \hat{\bm{r}}_2'$ and, equivalently, $\hat{\bm{r}}_1' =_R \hat{\bm{r}}_2'$. Since $\hat{\mathcal{R}}$ is a subautomaton of $\mathcal{T}$, we have the same cross diagrams in $\mathcal{T}$ and are done by induction.
      
      For $c \in \{ a \} \cup \{ \iota, f_\alpha, f_\beta, f \}$, we have
      \begin{center}
        \begin{tikzpicture}[baseline=(m-2-1.base)]
          \matrix[matrix of math nodes, text height=1.25ex, text depth=0.25ex, ampersand replacement=\&] (m) {
            \& a/\iota/f_\alpha/f_\beta/f \& \\
            \hat{\bm{r}}_1 \&   \& \lambda_R^{|\hat{\bm{r}}_1|_R} \\
            \& a/\iota/f_\alpha/f_\beta/f \& \\
          };
          \foreach \i in {1} {
            \foreach \j in {1} {
              \draw[->] let
                \n1 = {int(2+\i)},
                \n2 = {int(1+\j)}
              in
                (m-\i-\n2) -> (m-\n1-\n2);
              \draw[->] let
                \n1 = {int(1+\i)},
                \n2 = {int(2+\j)}
              in
                (m-\n1-\j) -> (m-\n1-\n2);
            };
          };
        \end{tikzpicture}%
        and
        \begin{tikzpicture}[baseline=(m-2-1.base)]
          \matrix[matrix of math nodes, text height=1.25ex, text depth=0.25ex, ampersand replacement=\&] (m) {
            \& a/\iota/f_\alpha/f_\beta/f \& \\
            \hat{\bm{r}}_2 \&   \& \lambda_R^{|\hat{\bm{r}}_2|_R} \\
            \& a/\iota/f_\alpha/f_\beta/f \& \\
          };
          \foreach \i in {1} {
            \foreach \j in {1} {
              \draw[->] let
                \n1 = {int(2+\i)},
                \n2 = {int(1+\j)}
              in
                (m-\i-\n2) -> (m-\n1-\n2);
              \draw[->] let
                \n1 = {int(1+\i)},
                \n2 = {int(2+\j)}
              in
                (m-\n1-\j) -> (m-\n1-\n2);
            };
          };
        \end{tikzpicture}%
      \end{center}
      and are done since $\hat{\bm{r}}_1 =_R \hat{\bm{r}}_2$ implies $|\hat{\bm{r}}_1|_R = |\hat{\bm{r}}_2|_R$. For $c = b$, we have
      \begin{center}
        \begin{tikzpicture}[baseline=(m-2-1.base)]
          \matrix[matrix of math nodes, text height=1.25ex, text depth=0.25ex, ampersand replacement=\&] (m) {
            \& b \& \\
            \hat{\bm{r}}_1 \&   \& \hat{\bm{r}}_1 \\
            \& b \& \\
          };
          \foreach \i in {1} {
            \foreach \j in {1} {
              \draw[->] let
                \n1 = {int(2+\i)},
                \n2 = {int(1+\j)}
              in
                (m-\i-\n2) -> (m-\n1-\n2);
              \draw[->] let
                \n1 = {int(1+\i)},
                \n2 = {int(2+\j)}
              in
                (m-\n1-\j) -> (m-\n1-\n2);
            };
          };
        \end{tikzpicture}%
        and
        \begin{tikzpicture}[baseline=(m-2-1.base)]
          \matrix[matrix of math nodes, text height=1.25ex, text depth=0.25ex, ampersand replacement=\&] (m) {
            \& b \& \\
            \hat{\bm{r}}_2 \&   \& \hat{\bm{r}}_2 \\
            \& b \& \\
          };
          \foreach \i in {1} {
            \foreach \j in {1} {
              \draw[->] let
                \n1 = {int(2+\i)},
                \n2 = {int(1+\j)}
              in
                (m-\i-\n2) -> (m-\n1-\n2);
              \draw[->] let
                \n1 = {int(1+\i)},
                \n2 = {int(2+\j)}
              in
                (m-\n1-\j) -> (m-\n1-\n2);
            };
          };
        \end{tikzpicture}%
      \end{center}
      and are done by induction.
      
      The remaining cases are $c \in \{ \alpha, \alpha', \beta, \beta' \}$. For these, we factorize $\hat{\bm{r}}_1 = \hat{\bm{s}}_1 \hat{\lambda}_1 \bm{i}_1$ with $\bm{i}_1 \in I^*$ maximal, $\hat{\lambda}_1 \in \hat{\Lambda} \cup \{ \varepsilon \}$ and $\hat{\bm{s}}_1 \in \hat{R}^*$ with $\lambda_1 = \varepsilon \implies \hat{\bm{s}}_1 = \varepsilon$. Analogously, we factorize $\hat{\bm{r}}_2 = \hat{\bm{s}}_2 \hat{\lambda}_2 \bm{i}_2$. Observe that, since we have $\hat{\bm{r}}_1 =_R \hat{\bm{r}}_2$, we must have $\bm{i}_1 = \bm{i}_2 = \bm{i}$, $\hat{\bm{s}}_1 \hat{\lambda}_1 =_R \hat{\bm{s}}_2 \hat{\lambda}_2$ and $\hat{\lambda}_1 = \varepsilon \iff \hat{\lambda}_2 = \varepsilon$. This yields the cross diagrams
      \begin{center}
        \begin{tikzpicture}[baseline=(m-4-1.base)]
          \matrix[matrix of math nodes, text height=2.25ex, text depth=0.25ex, ampersand replacement=\&] (m) {
              \& \alpha\textcolor{lightgray}{'} \& \\
            \bm{i} \&   \& \varphi(\bm{i}) \\
              \& \alpha\textcolor{lightgray}{'} \& \\
            \hat{\lambda}_1 \& \& \lambda_R^{|\hat{\lambda}_1|_R} \\
              \& f_\alpha \& \\
            \hat{\bm{s}}_1 \& \& \lambda_R^{|\hat{\bm{s}}_1|_R} \\
              \& f_\alpha \& \\
          };
          \foreach \i in {1,3,5} {
            \foreach \j in {1} {
              \draw[->] let
                \n1 = {int(2+\i)},
                \n2 = {int(1+\j)}
              in
                (m-\i-\n2) -> (m-\n1-\n2);
              \draw[->] let
                \n1 = {int(1+\i)},
                \n2 = {int(2+\j)}
              in
                (m-\n1-\j) -> (m-\n1-\n2);
            };
          };
          
          \path[fill=gray, opacity=0.2, rounded corners] (m-4-3.north -| m-4-1.west) rectangle (m-7-2.south -| m-6-3.east);
        \end{tikzpicture}%
        and
        \begin{tikzpicture}[baseline=(m-4-1.base)]
          \matrix[matrix of math nodes, text height=2.25ex, text depth=0.25ex, ampersand replacement=\&] (m) {
              \& \alpha\textcolor{lightgray}{'} \& \\
            \bm{i} \&   \& \varphi(\bm{i}) \\
              \& \alpha\textcolor{lightgray}{'} \& \\
            \hat{\lambda}_2 \& \& \lambda_R^{|\hat{\lambda}_2|_R} \\
              \& f_\alpha \& \\
            \hat{\bm{s}}_2 \& \& \lambda_R^{|\hat{\bm{s}}_2|_R} \\
              \& f_\alpha \& \\
          };
          \foreach \i in {1,3,5} {
            \foreach \j in {1} {
              \draw[->] let
                \n1 = {int(2+\i)},
                \n2 = {int(1+\j)}
              in
                (m-\i-\n2) -> (m-\n1-\n2);
              \draw[->] let
                \n1 = {int(1+\i)},
                \n2 = {int(2+\j)}
              in
                (m-\n1-\j) -> (m-\n1-\n2);
            };
          };
          
          \path[fill=gray, opacity=0.2, rounded corners] (m-4-3.north -| m-4-1.west) rectangle (m-7-2.south -| m-6-3.east);
        \end{tikzpicture}%
      \end{center}
      where the shaded parts only exist if $\hat{\lambda}_1, \hat{\lambda}_2 \neq \varepsilon$ and where we have $\alpha'$ after applying $\bm{i}$ if $\bm{i} \neq \varepsilon$. In both diagrams, we have the same state sequence on the right hand side (because of $\hat{\bm{s}}_1 \hat{\lambda}_1 =_R \hat{\bm{s}}_2 \hat{\lambda}_2$) and, thus, we are done. The case $c \in \{ \beta, \beta' \}$ is analogous with $\psi$ instead of $\varphi$.
    \end{proof}
    
    Finally, we show that a solution for the \DecProblem{PCP} instance implies a proper relation in the semigroup generated by $\mathcal{T}$ and, thus, that it is not free.
    \begin{lemma}\label{lem:solutionImpliesRelation}
      If $\bm{i} \in I^+$ is a solution for the \DecProblem{PCP} instance, then we have 
      \[
        \#_1 \bm{i} \#_1 =_{\mathcal{T}} \#_1 \bm{i} \#_2 \text{.}
      \]
    \end{lemma}
    \begin{proof}
      We show $\#_1 \bm{i} \#_1 \circ u = \#_1 \bm{i} \#_2 \circ u$ for all $u \in \Sigma^*$. For $u = \varepsilon$, there is nothing to show. So, let $u = c u'$ for some $c \in \Sigma = \Gamma \cup  \{ a, b \} \cup \{ \iota, \alpha, \alpha', f_\alpha, \beta, \beta', f_\beta, f \}$ and $u' \in \Sigma^*$. For $c \in \Gamma$ (the alphabet of $\hat{\mathcal{R}}$ and $\mathcal{S}$), we have the cross diagram
      \begin{center}
        \begin{tikzpicture}[baseline=(m-4-1.base)]
          \matrix[matrix of math nodes, text height=1.25ex, text depth=0.25ex] (m) {
                   & c   &      \\
              \#_x &     & \lambda_\# \cdot c \\
                   & c'  &      \\
            \bm{i} &     & \bm{i} \cdot c'    \\
                   & d   &      \\
              \#_1 &     & \lambda_\# \cdot d \\
                   & d'  &      \\
          };
          \foreach \i in {1,3,5} {
            \foreach \j in {1} {
              \draw[->] let
                \n1 = {int(2+\i)},
                \n2 = {int(1+\j)}
              in
                (m-\i-\n2) -> (m-\n1-\n2);
              \draw[->] let
                \n1 = {int(1+\i)},
                \n2 = {int(2+\j)}
              in
                (m-\n1-\j) -> (m-\n1-\n2);
            };
          };
        \end{tikzpicture}
      \end{center}
      for both, $x = 1$ and $x = 2$ with the same $c', d, d' \in \Gamma$. This shows $\#_1 \bm{i} \#_1 \circ cu' = \#_1 \bm{i} \#_2 \circ cu'$ for all $c \in \Gamma$ and $u' \in \Sigma^*$. The cases $c \in \{ a \} \cup \{ \alpha, \beta, \alpha', \beta', f_\alpha, f_\beta, f \}$ are similar; they are depicted in \autoref{fig:solutionImplesRelation}. The case $c = b$, requires induction but is still similar; it is depicted in \autoref{sfig:solutionImpliesRelationB}. Finally, the case $c = \iota$ is the most interesting one. Writing $\bm{i} = i_K \dots i_2 i_1$ for $i_1, \dots, i_K \in I$, we obtain
      \begin{center}
        \begin{tikzpicture}[baseline=(m-6-1.base)]
          \matrix[matrix of math nodes, text height=1.25ex, text depth=0.25ex] (m) {
                   & \iota   &      \\
              \#_1 &         & \lambda_\#   \\
                   & \alpha  &      \\
               i_1 &         & \varphi(i_1) \\
                   & \alpha' &      \\
               i_2 &         & \varphi(i_2) \\
                   & \alpha' &      \\
              \vdots &       & \vdots \\
                   & \alpha' &      \\
              i_K &       & \varphi(i_K) \\
                   & \alpha' &      \\
              \#_1 &         & \lambda_\#   \\
                   & f       &      \\
          };
          \foreach \i in {1,3,5,9,11} {
            \foreach \j in {1} {
              \draw[->] let
                \n1 = {int(2+\i)},
                \n2 = {int(1+\j)}
              in
                (m-\i-\n2) -> (m-\n1-\n2);
              \draw[->] let
                \n1 = {int(1+\i)},
                \n2 = {int(2+\j)}
              in
                (m-\n1-\j) -> (m-\n1-\n2);
            };
          };
        \end{tikzpicture}
        and
        \begin{tikzpicture}[baseline=(m-6-1.base)]
          \matrix[matrix of math nodes, text height=1.25ex, text depth=0.25ex] (m) {
                 & \iota  &      \\
            \#_2 &        & \lambda_\# \\
                 & \beta  &      \\
             i_1 &        & \psi(i_1) \\
                 & \beta' &      \\
             i_2 &        & \psi(i_2) \\
                 & \beta' &      \\
            \vdots &      & \vdots \\
                 & \beta' &      \\
            i_K &      & \psi(i_K) \\
                 & \beta' &      \\
            \#_1 &        & \lambda_\# \\
                 & f      &      \\
          };
          \foreach \i in {1,3,5,9,11} {
            \foreach \j in {1} {
              \draw[->] let
                \n1 = {int(2+\i)},
                \n2 = {int(1+\j)}
              in
                (m-\i-\n2) -> (m-\n1-\n2);
              \draw[->] let
                \n1 = {int(1+\i)},
                \n2 = {int(2+\j)}
              in
                (m-\n1-\j) -> (m-\n1-\n2);
            };
          };
        \end{tikzpicture}.
      \end{center}
      Since $\bm{i} = i_K \dots i_2 i_1$ is a solution, we have $\varphi(i_K) \dots \varphi(i_2) \varphi(i_1) =_R \psi(i_K) \dots \psi(i_2) \psi(i_1)$. Thus, \autoref{lem:REquivalentImpliesTEquivalent} implies $\lambda_\# \varphi(i_K) \dots \varphi(i_2) \varphi(i_1) \lambda_\# =_{\mathcal{T}} \lambda_\# \psi(i_K) \dots \psi(i_2) \psi(i_1) \lambda_\#$ and we are done.
    \end{proof}
    \begin{figure}[b]\centering
      \subcaptionbox{$c \in \{ \alpha, \beta \}$\label{sfig:solutionImpliesRelationAlphaBeta}}{%
        \begin{tikzpicture}[baseline=(m-4-1.base)]
          \matrix[matrix of math nodes, text height=1.25ex, text depth=0.25ex, ampersand replacement=\&] (m) {
                 \& \alpha/\beta \&      \\
            \#_x \&              \& \lambda_\#    \\
                 \& f_{\alpha}/f_{\beta} \& \\
           \bm{i}\&   \& \lambda_R^{|\bm{i}|} \\
                 \& f_\alpha/f_\beta \&      \\
            \#_1 \&   \& \lambda_\#    \\
                 \& f_\alpha/f_\beta \&      \\
          };
          \foreach \i in {1,3,5} {
            \foreach \j in {1} {
              \draw[->] let
                \n1 = {int(2+\i)},
                \n2 = {int(1+\j)}
              in
                (m-\i-\n2) -> (m-\n1-\n2);
              \draw[->] let
                \n1 = {int(1+\i)},
                \n2 = {int(2+\j)}
              in
                (m-\n1-\j) -> (m-\n1-\n2);
            };
          };
        \end{tikzpicture}%
      }
      \subcaptionbox{$c \in \{ \alpha', \beta' \}$\label{sfig:solutionImpliesRelationAlphaBetaPrime}}{%
        \begin{tikzpicture}[baseline=(m-4-1.base)]
          \matrix[matrix of math nodes, text height=1.25ex, text depth=0.25ex, ampersand replacement=\&] (m) {
               \& \alpha'/\beta' \&      \\
          \#_x \&                \& \lambda_\# \\
               \& f \& \\
            \bm{i}\&   \& \lambda_R^{|\bm{i}|} \\
               \& f \&      \\
            \#_1 \&   \& \lambda_\#    \\
               \& f \&      \\
          };
          \foreach \i in {1,3,5} {
            \foreach \j in {1} {
              \draw[->] let
                \n1 = {int(2+\i)},
                \n2 = {int(1+\j)}
              in
                (m-\i-\n2) -> (m-\n1-\n2);
              \draw[->] let
                \n1 = {int(1+\i)},
                \n2 = {int(2+\j)}
              in
                (m-\n1-\j) -> (m-\n1-\n2);
            };
          };
        \end{tikzpicture}%
      }
      \subcaptionbox{$c \in \{ f_\alpha, f_\beta, f \}$\label{sfig:solutionImpliesRelationF}}{%
        \begin{tikzpicture}[baseline=(m-4-1.base)]
          \matrix[matrix of math nodes, text height=1.25ex, text depth=0.25ex, ampersand replacement=\&] (m) {
              \& f_\alpha/f_\beta/f \&      \\
          \#_x \&              \& \lambda_\# \\
              \& f_\alpha/f_\beta/f \& \\
            \bm{i}\&   \& \lambda_R^{|\bm{i}|} \\
              \& f_\alpha/f_\beta/f \&      \\
            \#_1 \&   \& \lambda_\#    \\
              \& f_\alpha/f_\beta/f \&      \\
          };
          \foreach \i in {1,3,5} {
            \foreach \j in {1} {
              \draw[->] let
                \n1 = {int(2+\i)},
                \n2 = {int(1+\j)}
              in
                (m-\i-\n2) -> (m-\n1-\n2);
              \draw[->] let
                \n1 = {int(1+\i)},
                \n2 = {int(2+\j)}
              in
                (m-\n1-\j) -> (m-\n1-\n2);
            };
          };
        \end{tikzpicture}%
      }
      \subcaptionbox{$c = a$\label{sfig:solutionImpliesRelationA}}{%
        \begin{tikzpicture}[baseline=(m-4-1.base)]
          \matrix[matrix of math nodes, text height=1.25ex, text depth=0.25ex, ampersand replacement=\&] (m) {
              \& a \&           \\
          \#_x \& \& \lambda_\# \\
              \& b \& \\
            \bm{i}\&   \& \bm{i} \\
              \& b \&      \\
            \#_1 \&   \& \#_1    \\
              \& b \&      \\
          };
          \foreach \i in {1,3,5} {
            \foreach \j in {1} {
              \draw[->] let
                \n1 = {int(2+\i)},
                \n2 = {int(1+\j)}
              in
                (m-\i-\n2) -> (m-\n1-\n2);
              \draw[->] let
                \n1 = {int(1+\i)},
                \n2 = {int(2+\j)}
              in
                (m-\n1-\j) -> (m-\n1-\n2);
            };
          };
        \end{tikzpicture}%
      }
      \subcaptionbox{$c = b$\label{sfig:solutionImpliesRelationB}}{%
        \begin{tikzpicture}[baseline=(m-4-1.base)]
          \matrix[matrix of math nodes, text height=1.25ex, text depth=0.25ex, ampersand replacement=\&] (m) {
              \& b \&      \\
          \#_x \& \& \#_x    \\
              \& b \& \\
            \bm{i}\&   \& \bm{i} \\
              \& b \&      \\
            \#_1 \&   \& \#_1    \\
              \& b \&      \\
          };
          \foreach \i in {1,3,5} {
            \foreach \j in {1} {
              \draw[->] let
                \n1 = {int(2+\i)},
                \n2 = {int(1+\j)}
              in
                (m-\i-\n2) -> (m-\n1-\n2);
              \draw[->] let
                \n1 = {int(1+\i)},
                \n2 = {int(2+\j)}
              in
                (m-\n1-\j) -> (m-\n1-\n2);
            };
          };
        \end{tikzpicture}%
      }
      \caption{Various cases for $c \in \Sigma$. The cross diagrams hold for $x \in \{ 1, 2 \}$.}\label{fig:solutionImplesRelation}
    \end{figure}
    
    \begin{proposition}\label{prop:solutionImpliesNotFree}
      If the \DecProblem{PCP} instance has a solution, $\mathscr{S}(\mathcal{T})$ is not (left) cancellative and, thus, not a free semigroup.
    \end{proposition}
    \begin{proof}
      Let $\bm{i} \in I^+$ be a solution for the \DecProblem{PCP} instance but suppose that $\mathscr{S}(\mathcal{T})$ is left cancellative. By \autoref{lem:solutionImpliesRelation}, we have the relation $\#_1 \bm{i} \#_1 =_{\mathcal{T}} \#_1 \bm{i} \#_2$, which implies $\#_1 =_{\mathcal{T}} \#_2$. This, however, constitutes a contradiction since we have $\#_1 \circ \iota = \alpha$ but $\#_2 \circ \iota = \beta$.
      
      Finally, if $\mathscr{S}(\mathcal{T})$ is not (left) cancellative it cannot be free (as every free semigroup is (left and right) cancellative, see \autoref{fct:freeIsCancellative}).
    \end{proof}

    \paragraph{Converse Direction.}
    To show that the \DecProblem{PCP} instance has a solution if the monoid is not free, we first introduce another definition.
    \begin{definition}[compatible state sequences]\label{def:compatible}
      We may factorize any pair $\bm{p}, \bm{q} \in Q^*$ (uniquely) as
      \[
        \bm{p} = \left( \bm{p}_s \#_{x_s} \right) \dots \left( \bm{p}_1 \#_{x_1} \right) \bm{p}_0 \quad \text{and} \quad \bm{q} = \left( \bm{q}_t \#_{y_t} \right) \dots \left( \bm{q}_1 \#_{y_1} \right) \bm{q}_0
      \]
      with $\bm{p}_0, \dots, \bm{p}_s, \bm{q}_0, \dots, \bm{q}_t \in \hat{R}^*$ and $x_1, \dots, x_s, y_1, \dots, y_t \in \{ 1, 2 \}$. We define:
      \[
        \bm{p} \text{ and } \bm{q} \text{ are \emph{compatible}} \iff s = t \text{ and } \forall\, 0 \leq i \leq s = t: \bm{p}_i =_R \bm{q}_i
      \]
    \end{definition}
    
    Any relation in the monoid is compatible:
    \begin{lemma}\label{lem:relationIsCompatible}
      Let $\bm{p}, \bm{q} \in Q^*$ with $\bm{p} =_{\mathcal{T}} \bm{q}$. Then, we have that $\bm{p}$ and $\bm{q}$ are compatible.
    \end{lemma}
    \begin{proof}
      We factorize $\bm{p}$ and $\bm{q}$ in the same way as in \autoref{def:compatible} and show the statement by induction on $s + t$.
      For $s = t = 0$, we have $\bm{p}_0 = \bm{p} =_{\mathcal{T}} \bm{q} = \bm{q}_0$. Since $\hat{\mathcal{R}}$ is a subautomaton of $\mathcal{T}$, this implies $\bm{p}_0 =_{\hat{\mathcal{R}}} \bm{q}_0$ and, equivalently, $\bm{p} = \bm{p}_0 =_R \bm{q}_0 = \bm{q}$.
      
      For the inductive step ($s + t > 0$), we may assume $s > 0$ (due to symmetry) or, in other words, that $\bm{p}$ contains at least one $\#_1$ or $\#_2$. We have $\bm{p} \circ a = b$ (compare to \autoref{sfig:T1dual}) and, thus, due to $\bm{p} =_{\mathcal{T}} \bm{q}$, also $\bm{q} \circ a = \bm{p} \circ a = b$. This is only possible (again, compare to \autoref{sfig:T1dual}) if $\bm{q}$ also contains at least one $\#_1$ or $\#_2$, i.\,e.\ if $t > 0$.
      
      From \autoref{fct:removeBlocks} (with $\mu = 1$), we obtain (for both $\bm{p}$ and $\bm{q}$):
      \begin{align*}
        \bm{p} \cdot a ={}& \bm{p}' \lambda_{\#} \lambda_R^{|\bm{p}_0|_R}\\
        &\text{for } \bm{p}' = \left( \bm{p}_s \#_{x_s} \right) \dots \left( \bm{p}_2 \#_{x_2} \right) \bm{p}_1 \text{ and}\\
        \bm{q} \cdot a ={}& \bm{q}' \lambda_{\#} \lambda_R^{|\bm{q}_0|_R}\\
        &\text{for } \bm{q}' = \left( \bm{q}_t \#_{x_t} \right) \dots \left( \bm{q}_2 \#_{x_2} \right) \bm{q}_1
      \end{align*}
      Now, $\bm{p} =_{\mathcal{T}} \bm{q}$ implies $\bm{p}' \lambda_{\#} \lambda_R^{|\bm{p}_0|_R} = \bm{p} \cdot a =_{\mathcal{T}} \bm{q} \cdot a = \bm{q}' \lambda_{\#} \lambda_R^{|\bm{q}_0|_R}$
      and we may apply the induction hypothesis, which yields that $\bm{p}' \lambda_\# \lambda_R^{|\bm{p}_0|_R}$ and $\bm{q}' \lambda_\# \lambda_R^{|\bm{q}_0|_R}$ are compatible.
      This means that we have $s = t$, $\bm{p}_{\mu} =_R \bm{q}_{\mu}$ for all $2 \leq \mu \leq s = t$ and $\bm{p}_1 \lambda_\# \lambda_R^{|\bm{p}_0|_R} =_R \bm{q}_1 \lambda_\# \lambda_R^{|\bm{q}_0|_R}$.
      Observe that the latter implies $\bm{p}_1 =_R \bm{q}_1$ (as we have chosen $\lambda_\#$ and $\lambda_R$ as different elements of $\Lambda$).
      In particular, we also obtain $\bm{p}_s \lambda_\# \bm{p}_{s - 1} \ldots \lambda_\# \bm{p}_1 =_R \bm{q}_t \lambda_\# \bm{q}_{t - 1} \ldots \lambda_\# \bm{q}_1$.
      
      Since $\mathcal{S}$ is a subautomaton of $\mathcal{T}$, $\bm{p} =_{\mathcal{T}} \bm{q}$ implies $\bm{p} =_{\mathcal{S}} \bm{q}$. As $\#_1$ and $\#_2$ act in the same way as $\lambda_\#$ in $\mathcal{S}$ by construction, this shows $\bm{p}_s \lambda_\# \ldots \bm{p}_1 \lambda_\# \bm{p}_0 =_{\mathcal{S}} \bm{q}_t \lambda_\# \ldots \bm{q}_1 \lambda_\# \bm{q}_0$ and, because of $\mathscr{S}(\mathcal{S}) \simeq R^+$, also $\bm{p}_s \lambda_\# \ldots \bm{p}_1 \lambda_\# \bm{p}_0 =_R \bm{q}_t \lambda_\# \ldots \bm{q}_1 \lambda_\# \bm{q}_0$. Now, because $R^*$ as a free monoid is cancellative (see \autoref{fct:freeIsCancellative}) and because we have $\bm{p}_s \lambda_\# \bm{p}_{s - 1} \ldots \lambda_\# \bm{p}_1 =_R \bm{q}_t \lambda_\# \bm{q}_{t - 1} \ldots \lambda_\# \bm{q}_1$ (from above), we obtain $\lambda_\# \bm{p}_0 =_R \lambda_\# \bm{q}_0$ and, finally, $\bm{p}_0 =_R \bm{q}_0$, which concludes the proof that $\bm{p}$ and $\bm{q}$ are compatible.
    \end{proof}
    
    On the other hand, not every compatible pair forms a semigroup relation. However, this is true by \autoref{lem:REquivalentImpliesTEquivalent} if, additionally, the subsequence containing only $\#_1$ and $\#_2$ is the same in both entries. To formalize this, we introduce the following definition.
    \begin{definition}[projection on $\{ \#_1, \#_2 \}$]\label{def:sharpProjection}
      Let $\pi_\# : Q^* \to \{ \#_1, \#_2 \}^*$ be the homomorphism given by $\pi_\#(\#_x) = \#_x$ for both $x \in \{ 1, 2 \}$ and $\pi_\#(\hat{r}) = \varepsilon$ for all other $\hat{r} \in Q \setminus \{ \#_1, \#_2 \} = \hat{R}$.
    \end{definition}
    \begin{lemma}\label{lem:sharpCompatibleImpliesRelation}
      Let $\bm{p}, \bm{q} \in Q^*$ such that $\bm{p}$ and $\bm{q}$ are compatible and we have $\pi_\#(\bm{p}) = \pi_\#(\bm{q})$. Then, we have $\bm{p} =_{\mathcal{T}} \bm{q}$.
    \end{lemma}
    \begin{proof}
      Factorize $\bm{p}$ and $\bm{q}$ in the same way as in \autoref{def:compatible}. Since $\bm{p}$ and $\bm{q}$ are compatible, we have $\bm{p}_\mu =_R \bm{q}_\mu$ for all $0 \leq \mu \leq s = t$. This implies $\bm{p}_\mu =_{\mathcal{T}} \bm{q}_\mu$ by \autoref{lem:REquivalentImpliesTEquivalent}. Finally, $\pi_\#(\bm{p}) = \pi_\#(\bm{q})$ implies $\#_{x_\mu} = \#_{y_\mu}$ for all $1 \leq \mu \leq s = t$ and we obtain $\bm{p} =_{\mathcal{T}} \bm{q}$ because ${=_{\mathcal{T}}}$ is a congruence.
    \end{proof}
    
    Combining the last two lemmas, we obtain that $\mathscr{S}(\mathcal{T})$ is a free semigroup if all its relations have the same projection under $\pi_{\#}$. Most importantly, we will later on apply the contraposition of the \enquote{only if} direction of the following lemma to obtain a relation with different images under the projection if the semigroup is not free.
    \begin{lemma}\label{lem:freeIfRelationsSharp}
      Let $\pi': Q^* \to (R \cup \{ \#_1, \#_2 \})^*$ be the extension of the natural projection $\pi$ (from \autoref{def:naturalProjection}) with $\pi'(\#_x) = \#_x$ for $x \in \{ 1, 2 \}$.
      Then, the following statements are equivalent:
      \begin{enumerate}[label=(\arabic*)]
        \item
          \label{itm:relationImpliesSharp}
          For all $\bm{p}, \bm{q} \in Q^+$ with $\bm{p} =_{\mathcal{T}} \bm{q}$, we have $\pi_\#(\bm{p}) = \pi_\#(\bm{q})$.
        \item
          \label{itm:wellDefHom}
          The map $\pi'$ induces a well-defined homomorphism $\mathscr{S}(\mathcal{T}) \to (R \cup \{ \#_1, \#_2 \})^+$.
        \item
          \label{itm:wellDefIso}
          The map $\pi'$ induces a well-defined isomorphism $\mathscr{S}(\mathcal{T}) \to (R \cup \{ \#_1, \#_2 \})^+$.
      \end{enumerate}
      
      In particular, $\mathscr{S}(\mathcal{T})$ is isomorphic to $(R \cup \{ \#_1, \#_2 \})^+$ if we have $\pi_\#(\bm{p}) = \pi_\#(\bm{q})$ for all $\bm{p}, \bm{q} \in Q^+$ with $\bm{p} =_{\mathcal{T}} \bm{q}$.
    \end{lemma}
    \begin{proof}
      Note that we have $\pi'(\bm{p}) = \pi'(\bm{q})$ for $\bm{p}, \bm{q} \in Q^+$ if and only if $\bm{p}$ and $\bm{q}$ are compatible and $\pi_{\#}(\bm{p}) = \pi_{\#}(\bm{q})$ holds.
      
      For the implication \ref{itm:relationImpliesSharp} $\implies$ \ref{itm:wellDefHom}, suppose we have $\pi_\#(\bm{p}) = \pi_\#(\bm{q})$ for all $\bm{p}, \bm{q} \in Q^+$ with $\bm{p} =_{\mathcal{T}} \bm{q}$. We want to show that $\pi'$ induces a well-defined homomorphism $\mathscr{S}(\mathcal{T}) \to (R \cup \{ \#_1, \#_2 \})^+$. If it is well-defined, it is clearly a homomorphism. Thus, we only need to show that it is well-defined. Let $\bm{p}, \bm{q} \in Q^+$ with $\bm{p} =_{\mathcal{T}} \bm{q}$. By \autoref{lem:relationIsCompatible}, we have that $\bm{p}$ and $\bm{q}$ are compatible. By hypothesis, we also obtain $\pi_{\#}(\bm{p}) = \pi_{\#}(\bm{q})$.
      
      For the implication \ref{itm:wellDefHom} $\implies$ \ref{itm:wellDefIso}, not that, if $\pi'$ induces a well-defined homomorphism $\mathscr{S}(\mathcal{T}) \to (R \cup \{ \#_1, \#_2 \})^+$, it is clearly surjective. Therefore, we only need to show that it is injective. For this, let $\bm{p}, \bm{q} \in Q^+$ be compatible with $\pi_{\#}(\bm{p}) = \pi_{\#}(\bm{q})$. This implies $\bm{p} =_{\mathcal{T}} \bm{q}$ by \autoref{lem:sharpCompatibleImpliesRelation}.
      
      Finally, for the implication \ref{itm:wellDefIso} $\implies$ \ref{itm:relationImpliesSharp}, suppose that $\pi'$ is a well-defined isomorphism. In particular, $\bm{p} =_{\mathcal{T}} \bm{q}$ implies $\pi'(\bm{p}) = \pi'(\bm{q})$ for all $\bm{p}, \bm{q} \in Q^+$ and, thus, $\pi_\#(\bm{p}) = \pi_\#(\bm{q})$.
    \end{proof}
    
    If we have found a relation whose sides have different images under $\pi_\#$, we obtain a solution for the \DecProblem{PCP} instance.
    \begin{lemma}\label{lem:sharpRelationImpliesSolution}
      If there are $\bm{p}, \bm{q} \in Q^+$ with $\bm{p} =_{\mathcal{T}} \bm{q}$ but $\pi_\#(\bm{p}) \neq \pi_\#(\bm{q})$, then the \DecProblem{PCP} instance has a solution.
    \end{lemma}
    \begin{proof}
      We factorize these $\bm{p}$ and $\bm{q}$ in the same way as in \autoref{def:compatible} and observe that $\bm{p}$ and $\bm{q}$ are compatible by \autoref{lem:relationIsCompatible}. We may assume that there is some $1 \leq \mu_0 \leq s = t$ with $\#_{x_{\mu_0}} = \#_1$ but $\#_{y_{\mu_0}} = \#_2$ (due to symmetry).
      
      We may assume $\mu_0 = 1$ without loss of generality. This is because we may substitute $\bm{p}$ by $\bm{p}' = \bm{p} \cdot a^{\mu_0 - 1}$ and $\bm{q}' = \bm{q} \cdot a^{\mu_0 - 1}$, for which we still have $\bm{p}' =_{\mathcal{T}} \bm{q}'$ and, by \autoref{fct:removeBlocks} (for $\mu_0 > 1$),
      \begin{align*}
        \bm{p}' &= (\bm{p}_s \#_{x_s}) \ldots (\bm{p}_{\mu_0} \#_{x_{\mu_0}}) \, \bm{p}_{\mu_0 - 1} \, \lambda_{\#} \lambda_R^{\mu - 2 + | \bm{p}_{\mu - 2} \ldots \bm{p}_0 |_R} \text{ and}\\
        \bm{q}' &= (\bm{q}_t \#_{y_t}) \ldots (\bm{q}_{\mu_0} \#_{y_{\mu_0}}) \, \bm{q}_{\mu_0 - 1} \, \lambda_{\#} \lambda_R^{\mu - 2 + | \bm{q}_{\mu - 2} \ldots \bm{q}_0 |_R} \text{.}
      \end{align*}
      
      With these assumptions, we apply $\bm{p}$ and $\bm{q}$ to $\iota$ and obtain the cross diagrams (see \autoref{fig:T2andDual})
      \begin{center}
        \begin{tikzpicture}[baseline=(m-6-1.base)]
          \matrix[matrix of math nodes, text height=1.25ex, text depth=0.25ex] (m) {
                     & \iota  &                \\
            \bm{p}_0 &        & \lambda_R^{|\bm{p}_0|_R} \\
                     & \iota  &                \\
            \#_1     &        & \lambda_\#     \\
                     & \alpha &                \\
            \bm{p}_1 &        & \bm{p}_1'      \\
                     & c_1    &                \\
            \#_{x_2} &        & p_2'           \\
                     & c_2    &                \\
            \tilde{\bm{p}} &  & \tilde{\bm{p}}'\\
                     & c      &                \\
          };
          \foreach \i in {1,3,5,7,9} {
            \foreach \j in {1} {
              \draw[->] let
                \n1 = {int(2+\i)},
                \n2 = {int(1+\j)}
              in
                (m-\i-\n2) -> (m-\n1-\n2);
              \draw[->] let
                \n1 = {int(1+\i)},
                \n2 = {int(2+\j)}
              in
                (m-\n1-\j) -> (m-\n1-\n2);
            };
          };
        \end{tikzpicture}
        and
        \begin{tikzpicture}[baseline=(m-6-1.base)]
          \matrix[matrix of math nodes, text height=1.25ex, text depth=0.25ex] (m) {
                     & \iota &                \\
            \bm{q}_0 &       & \lambda_R^{|\bm{q}_0|_R} \\
                     & \iota &                \\
            \#_2     &       & \lambda_\#     \\
                     & \beta &                \\
            \bm{q}_1 &       & \bm{q}_1'      \\
                     & d_1   &                \\
            \#_{y_2} &       & q_2'           \\
                     & d_2   &                \\
            \tilde{\bm{q}} & & \tilde{\bm{q}}'\\
                     & d     &                \\
          };
          \foreach \i in {1,3,5,7,9} {
            \foreach \j in {1} {
              \draw[->] let
                \n1 = {int(2+\i)},
                \n2 = {int(1+\j)}
              in
                (m-\i-\n2) -> (m-\n1-\n2);
              \draw[->] let
                \n1 = {int(1+\i)},
                \n2 = {int(2+\j)}
              in
                (m-\n1-\j) -> (m-\n1-\n2);
            };
          };
        \end{tikzpicture}
      \end{center}
      for $\tilde{\bm{p}} = \bm{p}_s \#_{x_s} \ldots \bm{p}_3 \#_{x_3} \bm{p}_2$, $\tilde{\bm{q}} = \bm{q}_t \#_{y_t} \ldots \bm{q}_3 \#_{y_3} \bm{q}_2$ and some $\bm{p}_1', \tilde{\bm{p}}', \bm{q}_1', \tilde{\bm{q}}' \in Q^*$, $p_2', q_2' \in Q$ and $c_1, c_2, c, d_1, d_2, d \in \Gamma$. Since we have $\bm{p} =_{\mathcal{T}} \bm{q}$, we must have $c = d$ and, by the construction of $\mathcal{T}$, this is only possible if $c = f = d$ (see \autoref{sfig:T2dual}). This, in turn, is only possible if we have $\bm{p}_1 = \bm{i} \in I^+$ and $\bm{q}_1 = \bm{j} \in I^+$. Since $\bm{p}$ and $\bm{q}$ are compatible, we must even have $\bm{i} = \bm{p}_1 =_R \bm{q}_1 = \bm{j}$, which implies $\bm{i} = \bm{j}$. Additionally, we also obtain $\bm{p}_1' =_R \varphi(\bm{i})$, $c_1 = \alpha'$, $p_2' = \lambda_\#$, $c_2 = f$, $\bm{q}_1' =_R \psi(\bm{i})$, $d_1 = \beta'$, $q_2' = \lambda_\#$, $d_2 = f$ and $\tilde{\bm{p}}' = \lambda_R^{|\bm{p}_s|_R} \lambda_{\#} \ldots \lambda_R^{|\bm{p}_3|_R} \lambda_{\#} \lambda_R^{|\bm{p}_2|_R}$ as well as $\tilde{\bm{q}}' = \lambda_R^{|\bm{q}_t|_R} \lambda_{\#} \ldots \lambda_R^{|\bm{q}_3|_R} \lambda_{\#} \lambda_R^{|\bm{q}_2|_R}$ from the construction of $\mathcal{T}$.

      This shows that we have
      \begin{align*}
        &\lambda_R^{|\bm{p}_s|_R} \lambda_{\#} \ldots \lambda_R^{|\bm{p}_3|_R} \lambda_{\#} \lambda_R^{|\bm{p}_2|_R} \,
        \lambda_\# \varphi(\bm{i}) \lambda_\# \, \lambda_R^{|\bm{p}_0|_R}\\
        =_{\mathcal{T}}{}&
        \lambda_R^{|\bm{q}_t|_R} \lambda_{\#} \ldots \lambda_R^{|\bm{q}_3|_R} \lambda_{\#} \lambda_R^{|\bm{q}_2|_R} \,
        \lambda_\# \psi(\bm{i}) \lambda_\# \, \lambda_R^{|\bm{q}_0|_R}
      \end{align*}
      and, by \autoref{lem:relationIsCompatible}, also that both sides are $R$-equivalent. Since $\bm{p}$ and $\bm{q}$ are compatible, we have $\lambda_R^{|\bm{p}_\mu|_R} =_R \lambda_R^{|\bm{q}_\mu|_R}$ for all $0 \leq \mu \leq s = t$. Combining this with the cancellativity of $R^*$, we obtain $\varphi(\bm{i}) =_R \psi(\bm{i})$ and, thus, that $\bm{i}$ is a solution for the \DecProblem{PCP} instance.
    \end{proof}
    
    We have now basically shown that the \DecProblem{PCP} instance has a solution if the semigroup generated by $\mathcal{T}$ is not free. However, we have shown even more, which we will state in \autoref{prop:equivalences}. For one part of this statement, however, we will first look at another consequence of \autoref{lem:relationIsCompatible}, namely that the semigroup $\mathscr{S}(\mathcal{T})$ has a length function.
    \begin{proposition}\label{prop:lengthFunction}
      The function
      \begin{alignat*}{2}
        Q = \hat{R} \uplus \{ \#_1, \#_2 \} &\to \mathbb{N} \\
        \hat{r} &\mapsto |\hat{r}|_R &&\text{ for } \hat{r} \in \hat{R} \text{ and} \\
        \#_x&\mapsto 1 &&\text{ for } x \in \{ 1, 2 \}
      \end{alignat*}
      induces a well-defined proper length function of $\mathscr{M}(\mathcal{T})$ (and, thus, a well-defined length function of $\mathscr{S}(\mathcal{T})$).
    \end{proposition}
    \begin{proof}
      Consider two state sequences $\bm{p}, \bm{q} \in Q^*$ with $\bm{p} =_{\mathcal{T}} \bm{q}$ and let $m$ and $n$ be the lengths they get mapped to, respectively.
      We have to show $m = n$. By \autoref{lem:relationIsCompatible}, we have that $\bm{p}$ and $\bm{q}$ are compatible and we may factorize them in the same way as in \autoref{def:compatible}. Then, we have
      \begin{align*}
        m &= s + \sum_{\mu = 1}^{s} |\bm{p}_\mu|_R
        = t + \sum_{\mu = 1}^{t} |\bm{q}_\mu|_R = n
      \end{align*}
      since $s = t$ and $\bm{p}_\mu =_R \bm{q}_\mu$ (as $\bm{p}$ and $\bm{q}$ are compatible).
    \end{proof}
    \begin{proposition}\label{prop:equivalences}
      The following statements are equivalent:
      \begin{enumerate}[label=(\arabic*)]
        \item\label{itm:solution}
          The \DecProblem{PCP} instance has a solution $\bm{i} \in I^+$.
        \item\label{itm:iRelation}
          We have $\#_1 \bm{i} \#_1 =_{\mathcal{T}} \#_1 \bm{i} \#_2$ for some $\bm{i} \in I^+$.
        \item\label{itm:sharpRelation}
          There are $\bm{p}, \bm{q} \in Q^+$ with $\bm{p} =_{\mathcal{T}} \bm{q}$ but $\pi_\#(\bm{p}) \neq \pi_\#(\bm{q})$.
      \end{enumerate}%
      \begin{multicols}{2}
        \begin{enumerate}[resume*]
          \item\label{itm:notFree}
            $\mathscr{S}(\mathcal{T})$ is not a free semigroup.
          \item\label{itm:notIsomorphic}
            $\mathscr{S}(\mathcal{T})$ is not isomorphic to\\$(R \cup \{ \#_1, \#_2 \})^+$.
          \item\label{itm:notCancellative}
            $\mathscr{S}(\mathcal{T})$ is not (left\footnote{Recall that we defined automaton semigroups by a left action here.}) cancellative.
          \item\label{itm:notEquidivisible}
            $\mathscr{S}(\mathcal{T})$ is not equidivisible.
        \end{enumerate}
        \begin{enumerate}[start=4, label=(\arabic*')]
          \item\label{itm:MnotFree}
          $\mathscr{M}(\mathcal{T})$ is not a free monoid.
          \item\label{itm:MnotIsomorphic}
          $\mathscr{M}(\mathcal{T})$ is not isomorphic to\\$(R \cup \{ \#_1, \#_2 \})^*$.
          \item\label{itm:MnotCancellative}
          $\mathscr{M}(\mathcal{T})$ is not (left) cancellative.
          \item\label{itm:MnotEquidivisible}
          $\mathscr{M}(\mathcal{T})$ is not equidivisible.
        \end{enumerate}
      \end{multicols}
    \end{proposition}
    \begin{proof}
      We first show that \ref{itm:solution}, \ref{itm:iRelation} and \ref{itm:sharpRelation} are equivalent. The implication \ref{itm:solution} $\implies$ \ref{itm:iRelation} is \autoref{lem:solutionImpliesRelation}, the implication \ref{itm:iRelation} $\implies$ \ref{itm:sharpRelation} is trivial and the implication \ref{itm:sharpRelation} $\implies$ \ref{itm:solution} is \autoref{lem:sharpRelationImpliesSolution}.
      
      The implications \ref{itm:solution} $\implies$ \ref{itm:notFree} and \ref{itm:solution} $\implies$ \ref{itm:notCancellative} are given by \autoref{prop:solutionImpliesNotFree}. The implication \ref{itm:notFree} $\implies$ \ref{itm:notIsomorphic} is trivial and the implication \ref{itm:notIsomorphic} $\implies$ \ref{itm:sharpRelation} follows from \autoref{lem:freeIfRelationsSharp} (as $\pi'$ cannot be a well-defined isomorphism in this case). The implication \ref{itm:notCancellative} $\implies$ \ref{itm:notFree} is trivial again.
      
      Finally, \ref{itm:notEquidivisible} is equivalent to \ref{itm:notFree} by \autoref{fct:leviLemma} since $\mathscr{S}(\mathcal{T})$ has a length function by \autoref{prop:lengthFunction}.
      
      For the monoid statements, observe that $\bm{p} \neq_{\mathcal{T}} \varepsilon$ for all $\bm{p} \in Q^+$ (which can, for example, be seen by observing $\bm{p} \circ \alpha \in \{ \alpha', f_\alpha \}$; compare to \autoref{sfig:T2dual}). Thus, we have $\mathscr{M}(\mathcal{T}) = \mathscr{S}(\mathcal{T})^\idGrp$ for $\idGrp \not\in \mathscr{S}(\mathcal{T})$. This shows that \ref{itm:notFree} and \ref{itm:MnotFree} as well as \ref{itm:notIsomorphic} and \ref{itm:MnotIsomorphic} are equivalent, respectively (which is most easily seen using the negations of the statements). That \ref{itm:MnotFree} and \ref{itm:MnotEquidivisible} are equivalent follows again from the existence of a proper length function (\autoref{prop:lengthFunction}) and \autoref{fct:leviLemma}. Finally, if $\mathscr{M}(\mathcal{T})$ is free, it is also cancellative (by \autoref{fct:freeIsCancellative}, \ref{itm:MnotCancellative} $\implies$ \ref{itm:MnotFree}) and, in particular, left cancellative, which, in turn, is then also trivially true for $\mathscr{S}(\mathcal{T})$ (\ref{itm:notCancellative} $\implies$ \ref{itm:MnotCancellative}).
    \end{proof}

    \paragraph{Main Theorem and other Consequences.}
    \autoref{prop:equivalences} shows that we have reduced \DecProblem{PCP} to (the complement of) \DecProblem{Semigroup Freeness} and \DecProblem{Monoid Freeness} (as the construction of $\mathcal{T}$ is computable). In fact, we even have a reduction to a stronger version of the problem(s) where the alphabet size is fixed. Since \DecProblem{PCP} is undecidable \cite{post1946variant}, this shows our main result.
    \begin{theorem}\label{thm:mainTheorem}
      The problem
        \problem
          [an alphabet $\Sigma$ of size $|\Sigma| = 25$]
          {a (complete) \SAut $\mathcal{T}$ with input/output alphabet $\Sigma$}
          {if $\mathscr{S}(\mathcal{T})$ a free semigroup?}
      and the problem
        \problem
          [an alphabet $\Sigma$ of size $|\Sigma| = 25$]
          {a (complete) \SAut $\mathcal{T}$ with input/output alphabet $\Sigma$}
          {if $\mathscr{M}(\mathcal{T})$ a free monoid?}
      are undecidable.
    \end{theorem}
    \begin{corollary}
      In particular, the freeness problem for automaton semigroups
        \SemigroupFreeness
      and the freeness problem for automaton monoids
        \MonoidFreeness
      are undecidable.
    \end{corollary}
    
    We also get the undecidability of a weaker form of the free presentation problem for automaton semigroups.
    \begin{theorem}\label{thm:partialSemigroupPresentation}
      The problem
      \problem
      [
        an alphabet $\Sigma$ of size $|\Sigma| = 25$
      ]
      {
        a (complete) \SAut $\mathcal{T} = (Q, \Sigma, \delta)$ and\newline
        a subset $P \subseteq Q$
      }{
        is $\mathscr{S}(\mathcal{T}) \simeq P^+$?
      }\noindent
      is undecidable.
    \end{theorem}
    \begin{proof}
      We can use the same reduction and choose $P = R \cup \{ \#_1, \#_2 \}$. The correctness of the reduction is then stated in \autoref{prop:equivalences}.
    \end{proof}
    \begin{remark*}
      Of course, we also get a corresponding monoid result but, for monoids, we will prove something stronger in \autoref{sec:monoidPresentation} anyway.
    \end{remark*}
    
    Additionally, we also get from our construction that it is not decidable whether a given \SAut generates a (left) cancellative or an equidivisible semigroup/monoid. This, again, follows form \autoref{prop:equivalences}.
    \begin{theorem}\label{thm:cancellativityEquidivisibilityIsUndecidable}
      The problems
      \problem
        [an alphabet $\Sigma$ of size $|\Sigma| = 25$]
        {a (complete) \SAut $\mathcal{T}$ with input/output alphabet $\Sigma$}
        {is $\mathscr{S}(\mathcal{T})$ (left) cancellative/equidivisible?}\noindent
      and the problems
      \problem
        [an alphabet $\Sigma$ of size $|\Sigma| = 25$]
        {a (complete) \SAut $\mathcal{T}$ with input/output alphabet $\Sigma$}
        {is $\mathscr{M}(\mathcal{T})$ (left) cancellative/equidivisible?}\noindent
      are undecidable.
    \end{theorem}
    
    Finally, we obtain that it is undecidable whether a given map on the generators induces a homomorphism (or an isomorphism) between two automaton semigroups. Note that the isomorphism problem for automaton groups (and, thus, also for automaton semigroups and monoids) is known to be undecidable (as it follows from \cite{sunic2012conjugacy}).
    \begin{theorem}
      The problems
      \problem
      [
        alphabets $\Sigma_1$ and $\Sigma_2$ of size $|\Sigma_1| = 25$ and $|\Sigma_2| = 2$
      ]{
        two (complete) \SAuta $\mathcal{T}_1 = (Q_1, \Sigma_1, \delta_1)$ and $\mathcal{T}_2 = (Q_2, \Sigma_2, \delta_2)$\newline
        and
        a map $f: Q_1 \to Q_2$
      }{
        does $f$ extend into a homomorphism $\mathscr{S}(\mathcal{T}_1) \to \mathscr{S}(\mathcal{T}_2)$?
      }\noindent
      and
      \problem[
        alphabets $\Sigma_1$ and $\Sigma_2$ of size $|\Sigma_1| = 25$ and $|\Sigma_2| = 2$
      ]{
        two (complete) \SAuta $\mathcal{T}_1 = (Q_1, \Sigma_1, \delta_1)$ and $\mathcal{T}_2 = (Q_2, \Sigma_2, \delta_2)$\newline
        and
        a map $f: Q_1 \to Q_2$
      }{
        does $f$ extend into an isomorphism $\mathscr{S}(\mathcal{T}_1) \to \mathscr{S}(\mathcal{T}_2)$?
      }\noindent
      are undecidable.
    \end{theorem}
    \begin{proof}
      We can use the same reduction from (the complement of) \DecProblem{PCP} for both problems. For $\mathcal{T}_1$, we use the automaton $\mathcal{T}$ constructed above and, for $\mathcal{T}_2$, we use an \SAut with $\mathscr{S}(\mathcal{T}_2) \simeq (R \cup \{ \#_1, \#_2 \})^+$ (which is computable by \autoref{prop:freeComputable}). For the map $f$, we can restrict $\pi'$ from \autoref{lem:freeIfRelationsSharp} into a map $Q \to Q_2$ (potentially using a union of appropriate powers of $\mathcal{T}_2$). Now, by \autoref{prop:equivalences}, the \DecProblem{PCP} instance has a solution if and only if we have $\pi_\#(\bm{p}) \neq \pi_\#(\bm{q})$ for some $\bm{p}, \bm{q} \in Q^+$ with $\bm{p} =_\mathcal{T} \bm{q}$. By \autoref{lem:freeIfRelationsSharp}, this is the case if and only if $\pi'$ does not induce a well-defined homomorphism/isomorphism $\mathscr{S}(\mathcal{T}) \to (R \cup \{ \#_1, \#_2 \})^+$.
    \end{proof}
  \end{section}
  
  \begin{section}{Free Presentations of Monoids}\label{sec:monoidPresentation}
    In this section, we show that \DecProblem{Free Monoid Presentation} is undecidable (which is stronger than the result for semigroups stated in \autoref{thm:partialSemigroupPresentation}) using a reduction similar to the one presented in \autoref{sct:semigroupFreeness}. This time we use a variant of \DecProblem{PCP} where we pad the components of the tiles to have the same length. To this end, let $\Lambda$ be an alphabet,\footnote{Again, the reader may find it convenient to refer to \autoref{tbl:monoidSymbols} for a summary of the symbols defined in this section.} choose some padding symbol $e \not\in \Lambda$ and define $\hat{\Lambda} = \Lambda \cup \{ e \}$.
    \begin{definition}[natural projection]
      Let $\pi: \hat{\Lambda}^* \to \Lambda^*$ be the natural projection given by $\pi(\lambda) = \lambda$ for all $\lambda \in \Lambda$ and $\pi(e) = \varepsilon$.
      
      We call two words $u, v \in \hat{\Lambda}^*$ \emph{$e$-equivalent} and write $u =_e v$ if they have the same image under $\pi$, i.\,e.\ we have $u =_e v \iff \pi(u) = \pi(v)$.
    \end{definition}
    
    With this definition at hand, let \DecProblem{ePCP} be the problem
    \problem[
      an alphabet $\Lambda$ and\newline
      a padding symbol $e \not\in \Lambda$
    ]{
      a number $L \in \mathbb{N}$ and\newline
      homomorphisms $\varphi, \psi: I = \{ 1, \dots, n \} \to (\Lambda \cup \{ e \})^*$ with\newline
      $|\varphi(i)| = |\psi(i)| = L$ for all $i \in I$.
    }{
      $\exists \bm{i} \in I^+: \varphi(\bm{i}) =_e \psi(\bm{i})$?
    }\noindent
    Clearly, \DecProblem{ePCP} is undecidable as we can reduce \DecProblem{PCP} to it by padding all $\varphi(i)$ and $\psi(i)$ to the same length $L$ using the padding symbol $e$.
    
    As in \autoref{sct:semigroupFreeness}, we fix an \DecProblem{ePCP} instance $\varphi, \psi, L, I$ with padding symbol $e$ and alphabet $\Lambda$ and compute from it an \SAut $\mathcal{T} = (Q, \Sigma, \delta)$ with $e \in Q$ acting as the identity in such a way that the \DecProblem{ePCP} instance has a solution if and only if $\mathscr{M}(\mathcal{T})$ is \textbf{not} (isomorphic to) $(Q \setminus \{ e \})^*$.
    
    The reduction steps are similar as before. First, we assume without loss of generality that we have $L \geq 2$, $|I| + |\Lambda| \geq 2$, $I \cap \Lambda = \emptyset$ and $e \not\in \Lambda, I$, and let $R = \Lambda \cup I$ as well as $\hat{R} = \hat{\Lambda} \cup I = \Lambda \cup I \cup \{ e \}$. Naturally, we can also extend the natural projection $\pi$ into a homomorphism $\pi: \hat{R}^* \to R^*$ by letting $\pi(i) = i$ for all $i \in I$. This yields that $\hat{\bm{r}}_1, \hat{\bm{r}}_2 \in \hat{R}^*$ are $e$-equivalent if they become equal when we remove all letters $e$ (compare to \autoref{def:naturalProjection}).
    
    \paragraph{Definition of the Automaton $\hat{\mathcal{R}}$.}
    Then, we compute an \SAut $\hat{\mathcal{R}}$ with state set $\hat{R}$ which generates the free monoid over $R$ and in which $e$ acts as the identity. This time, we do not use a power automaton as this would always create relations in the generated monoid (which we need to avoid for a free presentation).
    \begin{fact}
      On input $I$, one can compute an \SAut $\hat{\mathcal{R}} = (\hat{R}, \Gamma, \rho)$ with state set $\hat{R} = \Lambda \cup I \cup \{ e \}$ and $\mathscr{M}(\mathcal{T}) = \mathscr{S}(\mathcal{T}) \simeq R^*$ where $e$ acts as the identity (i.\,e.\ $e =_{\hat{\mathcal{R}}} \varepsilon$).
    \end{fact}
    \begin{proof}
      We can compute an \SAut $\mathcal{R} = (R, \Gamma, \rho_1)$ with $\mathscr{S}(\mathcal{R}) \simeq R^+$ by \autoref{prop:freeComputable}. To obtain $\hat{\mathcal{R}}$, we simply add the new state $e$ with $c / c$ self-loops for all $c \in \Gamma$.
    \end{proof}
    \begin{remark*}
      From \autoref{prop:freeComputable}, we obtain that $\Gamma$ is a binary alphabet. However, we will not use this fact as the alphabet of the eventual automaton $\mathcal{T}$ constructed for the reduction will depend on $L$ and $|I|$ and, thus, on the input instance anyway.
    \end{remark*}
    
    By the construction of $\hat{\mathcal{R}}$, we have $\hat{\bm{r}}_1 =_e \hat{\bm{r}}_2$ if and only if $\hat{\bm{r}}_1 =_{\hat{\mathcal{R}}} \hat{\bm{r}}_2$ for all $\hat{\bm{r}}_1, \hat{\bm{r}}_2 \in \hat{R}^*$.
    
    \paragraph{Definition of the Automaton $\mathcal{S}$.}
    The next step is to compute $\mathcal{S}$ from $\hat{\mathcal{R}}$. Here, we introduce two new states $\#_1, \#_2 \not\in \hat{R}$ that act as the identity. We let $Q = \hat{R} \cup \{ \#_1, \#_2 \}$ and $\mathcal{S} = (Q, \Gamma, \sigma)$ with the transitions
    \[
      \sigma = \rho \cup \{ \trans{\#_x}{c}{c}{e} \mid x \in \{ 1, 2 \}, c \in \Gamma \} \text{.}
    \]
    Clearly, we still have $\mathscr{M}(\mathcal{S}) = \mathscr{M}(\hat{\mathcal{R}}) \simeq R^*$.
    
    \paragraph{Definition of the Automaton $\mathcal{T}$.}
    Before we can finally define $\mathcal{T}$, we first add two new letters $a$ and $b$ to the alphabet and the transitions depicted in \autoref{sfig:monoidT1transitions} (compare to \autoref{fig:T1transitions}). This yields the complete \SAut $\mathcal{T}_1 = (Q, \Gamma \cup \{ a, b \}, \delta_1)$.
    \begin{figure}\centering
      \subcaptionbox{new transitions\label{sfig:monoidT1transitions}}{%
        \begin{tikzpicture}[auto, shorten >=1pt, >=latex, baseline=(s.base)]
          \node[state] (r) {$\hat{r}$};
          \node[state, right=of r, dashed] (lambda) {$e$};
          \node[state, below=1.5cm of r.center, anchor=center] (s) {$\#_x$};
          \node[state, below=1.5cm of lambda.center, anchor=center, dashed] (lambdas) {$e$};
          
          \path[->] (r) edge[loop left] node {$b / b$} (r)
                        edge node {$a/a$} (lambda)
                    (s) edge[loop left] node {$b / b$} (s)
                        edge node[swap] {$a/b$} (lambdas)
          ;
        \end{tikzpicture}
      }
      \subcaptionbox{dual transitions\label{sfig:monoidT1dual}}{%
        \begin{tikzpicture}[auto, shorten >=1pt, >=latex, baseline=(a.base)]
          \node[state] (a) {$a$};
          \node[state, right=of a] (b) {$b$};
          
          \path[->] (a) edge[loop left] node {$\hat{r}/e$} (a)
                        edge node {$\#_x/e$} (b)
                    (b) edge[loop right] node {$\id_Q$} (b)
          ;
        \end{tikzpicture}
      }
      \caption{The new transitions for $\mathcal{T}_1$. The transitions exist for all $\hat{r} \in \hat{R}$ and $x \in \{ 1, 2 \}$; in particular, $e$ still acts as the identity also on $a$ and $b$. The transitions for dashed states are implicitly defined and $\id_Q$ indicates that we have $q/q$ transitions for all $q \in Q$.}\label{fig:monoidT1transitions}
    \end{figure}
    
    Finally, we take $\mathcal{T} = \mathcal{T}_1 \cup \mathcal{T}_2$ where $\mathcal{T}_2$ is given by \autoref{sfig:monoidT2} (compare to \autoref{fig:T2andDual}).
    This time $\mathcal{T}_2$ is a bit more complicated than in the semigroup case because we cannot assume that $\varphi(i)$ and $\psi(i)$ are single states in our automaton.
    \AddThispageHook{\thispagestyle{empty}}
    \begin{figure}
      \begin{subfigure}{\linewidth}\centering
        \begin{tikzpicture}[auto, shorten >=1pt, >=latex, baseline=(e.base)]
          \node[state, dashed] (e) {$e$};
          \node[state, left=3cm of e] (s1) {$\#_1$};
          \node[state, right=3cm of e] (s2) {$\#_2$};
        
          \path[->] (s1) edge node[align=center] {%
            $\alpha_0/f_\alpha,\; \alpha_{i, \ell}/f_\alpha\; f_\alpha/f_\alpha$\\%
            $\beta_0/f_\beta,\; \beta_{i, \ell}/f_\beta,\; f_\beta/f_\beta$\\%
            $\alpha_L/f,\; \beta_L/f,\; f/f$}
                              node[below] {$\iota/\alpha_0$} (e)
                    (s2) edge node[above, align=center] {%
            $\alpha_0/f_\alpha,\; \alpha_{i, \ell}/f_\alpha\; f_\alpha/f_\alpha$\\%
            $\beta_0/f_\beta,\; \beta_{i, \ell}/f_\beta,\; f_\beta/f_\beta$\\%
            $\alpha_L/f,\; \beta_L/f,\; f/f$}
                              node[below] {$\iota/\beta_0$} (e)
        ;
        \end{tikzpicture}\\
        \begin{tikzpicture}[auto, shorten >=1pt, >=latex, baseline=(e.base)]
          \node[state] (i) {$i$};
          \node[state, ellipse, above left=0cm and 4cm of i, dashed] (phi1) {$\varphi_1(i)$};
          \node[state, ellipse, below left=0cm and 4cm of i, dashed] (phil+1) {$\varphi_{\ell + 1}(i)$};
          \node[state, below=2cm of i, dashed] (e) {$e$};
          \node[state, ellipse, above right=0cm and 4cm of i, dashed] (psi1) {$\psi_1(i)$};
          \node[state, ellipse, below right=0cm and 4cm of i, dashed] (psil+1) {$\psi_{\ell + 1}(i)$};
        
          \path[->] (i) edge node[base right, pos=0.8] {$\id_{\{ \iota, f_\alpha, f_\beta, f \}}$}
                             node[base left, pos=0.8, align=center] {$\alpha_{j, \ell}/f_\alpha$\\$\beta_{j, \ell}/f_\beta$} (e)
                        edge node[above, sloped] {$\alpha_0/\alpha_{i, 1},\; \alpha_L/\alpha_{i, 1}$} (phi1)
                        edge node[below, sloped] {$\alpha_{i, \ell}/\alpha_{i, \ell + 1}$} (phil+1)
                        edge node[above, sloped] {$\beta_0/\beta_{i, 1},\; \beta_L/\beta_{i, 1}$} (psi1)
                        edge node[below, sloped] {$\beta_{i, \ell}/\beta_{i, \ell + 1}$} (psil+1)
          ;
        \end{tikzpicture}\\
        \begin{tikzpicture}[auto, shorten >=1pt, >=latex, baseline=(e.base)]
          \node[state] (l) {$\lambda$};
          \node[state, right=4cm of l] (e) {$e$};
          
          \path[->] (l) edge node {$\id_{\{ \iota, f_\alpha, f_\beta, f \}}$}
                             node[below, align=center] {%
                               $\alpha_0/f_\alpha,\; \alpha_{i, \ell}/f_\alpha,\; \alpha_L/f_\alpha$\\%
                               $\beta_0/f_\beta,\; \beta_{i, \ell}/f_\beta,\; \beta_L/f_\beta$} (e)
                    (e) edge[loop right] node {$\id_\Sigma$} (e)
          ;
        \end{tikzpicture}
        \caption{Schematic depiction of $\mathcal{T}_2$ over the (new) alphabet $\{ \alpha_0, \alpha_L, \beta_0, \beta_L \} \cup \{ \alpha_{i, \ell}, \beta_{i, \ell} \mid i \in I, 1 \leq \ell < L \}$. The transitions exist for all $i \in I$, $j \in I \setminus \{ i \}$, $\lambda \in \Lambda$ and $1 \leq \ell < L$; transitions labeled by $\id_X$ for $X \subseteq \Sigma$ indicate that we have an $x/x$ transition for all $x \in X$ and we use $\varphi_\ell(i)$ and $\psi_\ell(i)$ to denote the letters of $\varphi(i)$ and $\psi(i)$, respectively, where we have $\varphi(i) = \varphi_L(i) \dots \varphi_1(i)$ and $\psi(i) = \psi_L(i) \dots \psi_1(i)$. Note that also the transitions at the dashed states are defined.}\label{sfig:monoidT2}
      \end{subfigure}\\
      \begin{subfigure}{\linewidth}\centering
        \begin{tikzpicture}[auto, shorten >=1pt, >=latex, baseline=(i.base)]
          \node[state] (i) {$\iota$};
          \node[state, above=0.5cm of i] (a0) {$\alpha_0$};
          \node[state, right=1.5cm of a0] (a1) {$\alpha_{i, 1}$};
          \node[right=1.5cm of a1] (adots) {$\dots$};
          \node[state, ellipse, right=1.5cm of adots] (a-1) {$\alpha_{i, L - 1}$};
          \node[state, right=1.5cm of a-1] (aL) {$\alpha_L$};
          \node[state, above=1.5cm of adots] (fa) {$f_\alpha$};
          \node[state, anchor=base] at ($(i.base-|aL.base)$) (f) {$f$};
          
          \node[state, below=0.5cm of i] (b0) {$\beta_0$};
          \node[state, right=1.5cm of b0] (b1) {$\beta_{i, 1}$};
          \node[right=1.5cm of b1] (bdots) {$\dots$};
          \node[state, ellipse, right=1.5cm of bdots] (b-1) {$\beta_{i, L - 1}$};
          \node[state, right=1.5cm of b-1] (bL) {$\beta_L$};
          \node[state, below=1.5cm of bdots] (fb) {$f_\beta$};

          \path[->] (i) edge[loop left] node {$r/e$} (i)
                        edge node {$\#_1/e$} (a0)
                        edge node[swap] {$\#_2/e$} (b0)
                    (a0) edge node {$i/\varphi_1(i)$} (a1)
                         edge[bend left] node {$\begin{aligned}
                           \#_x, \lambda &/ e
                           \end{aligned}$} (fa)
                    (a1) edge node {$i/\varphi_2(i)$} (adots)
                         edge[bend left] node[below right, pos=0.6] {$\begin{aligned}
                           \#_x, \lambda, j &/ e
                           \end{aligned}$} (fa)
                    (adots) edge node {$i/\varphi_{L - 1}(i)$} (a-1)
                    (a-1) edge node {$i/\varphi_{L}(i)$} (aL)
                            edge[bend right] node[left, pos=0.2] {$\begin{aligned}
                              \#_x, \lambda, j &/ e
                            \end{aligned}$} (fa)
                    (aL) edge[bend left] node[swap] {$i/\varphi_1(i)$} (a1)
                         edge[bend right] node[swap, pos=0.2] {$\lambda / e$} (fa)
                         edge node {$\#_x / e$} (f)
                    (fa) edge[loop above] node[pos=0.25, left] {$q/e$} (fa)
                    (f) edge[loop right] node {$q/e$} (f)
                    (b0) edge node[swap] {$i/\psi_1(i)$} (b1)
                         edge[bend right] node[swap] {$\begin{aligned}
                           \#_x, \lambda &/ e
                         \end{aligned}$} (fb)
                    (b1) edge node[swap] {$i/\psi_2(i)$} (bdots)
                         edge[bend right] node[above right, pos=0.6] {$\begin{aligned}
                           \#_x, \lambda, j &/ e
                         \end{aligned}$} (fb)
                    (bdots) edge node[swap] {$i/\psi_{L - 1}(i)$} (b-1)
                    (b-1) edge node[swap] {$i/\psi_{L}(i)$} (bL)
                          edge[bend left] node[left, pos=0.2] {$\begin{aligned}
                            \#_x, \lambda, j &/ e
                          \end{aligned}$} (fb)
                    (bL) edge[bend right] node[] {$i/\psi_1(i)$} (b1)
                         edge[bend left] node[pos=0.2] {$\lambda / e$} (fb)
                         edge node[swap] {$\#_x / e$} (f)
                    (fb) edge[loop below] node[pos=0.25, right] {$q/e$} (fb)
          ;
        \end{tikzpicture}
        \caption{The dual $\partial \mathcal{T}_2$. The transitions (and states) exist for all $i \in I, j \in I \setminus \{ i \}, x \in \{ 1, 2 \}, q \in Q, r \in R$ and $\lambda \in \Lambda$. Transitions with multiple inputs exist for all these inputs (with the same output). We write $\varphi(i) = \varphi_L(i) \dots \varphi_1(i)$ and $\psi(i) = \psi_L(i) \dots \psi_1(i)$ for the individual letters $\varphi_1(i), \dots, \varphi_L(i), \psi_1(i), \dots, \psi_L(i) \in \Lambda$. The $e/e$ loops at all states are omitted.}\label{sfig:monoidT2dual}
      \end{subfigure}
      \caption{The automaton $\mathcal{T}_2$ and its dual.}\label{fig:monoidT2}
    \end{figure}
    
    This means that $\mathcal{T} = (Q, \Sigma, \delta)$ is obtained from $\mathcal{T}_1$ by letting $\Sigma = \Gamma \cup \{ a, b \} \cup \{ \alpha_0, \alpha_L, \beta_0, \beta_L \} \cup \{ \alpha_{i, \ell}, \beta_{i, \ell} \mid i \in I, 1 \leq \ell < L \}$ (where all symbols are new) and adding the transitions depicted in \autoref{fig:monoidT2}. Note that the resulting automaton is a complete \SAut and can be computed. Note also that $e$ continues to act as an identity.
    
    \begin{remark}\label{rem:noFixedAlphabetForMonoidPresentation}
      In contrast to the reduction in \autoref{sct:semigroupFreeness}, the automaton $\mathcal{T}$ does not have an alphabet of fixed size. In fact (since we may assume $|\Gamma| = 2$), we have $|\Sigma| = 8 + |I| \cdot L$ (where $L$ and $I$ both are part of/depend on the input to the reduction function).
    \end{remark}
    
    \paragraph*{The Role of $a$ and $b$ in $\mathcal{T}$.}
    The role of $a$ and $b$ this time is very much the same as in the semigroup case above: we may use $a$ to remove a block from a certain factorization of a state sequence. Formally, this is stated in the following fact, which is an analogue of \autoref{fct:removeBlocks} and which may be proved using a similar (this time even simpler) induction.
    \begin{fact}\label{fct:monoidRemoveBlocks}
      Let $\bm{p} \in Q^*$ and factorize it as
      \[
        \bm{p} = (\bm{p}_s \#_{x_s}) \dots (\bm{p}_1 \#_{x_1}) \, \bm{p}_0
      \]
      for $\bm{p}_0, \dots, \bm{p}_s \in \hat{R}^*$ and $x_1, \dots, x_s \in \{ 1, 2 \}$.
      
      Then, for any $0 \leq \mu \leq s$, we have (in $\mathcal{T}$):
      \[
        \bm{p} \cdot a^{\mu} = (\bm{p}_s \#_{x_s}) \ldots (\bm{p}_{\mu + 1} \#_{x_{\mu + 1}}) \, \bm{p}_\mu \, e^{\mu + | \bm{p}_{\mu - 1} \ldots \bm{p}_0 |_R}
      \]
    \end{fact}
    
      \begin{table}[p]\centering
        \begin{tabular}{r@{}l}\toprule
          \textbf{symbol}\phantom{${}:{}$}& \textbf{usage} \\\midrule
          $\Lambda : {}$ & PCP base alphabet, $|I| + |\Lambda| \geq 2$, $I \cap \Lambda = \emptyset$\\
          $e \in {}$ & $\hat{\Lambda} \setminus \Lambda \subseteq \hat{R} \subseteq Q : {}$PCP padding symbol and identity state in $\hat{\mathcal{R}}$, $\mathcal{S}$ and $\mathcal{T}$\\
          $\hat{\Lambda} = {}$ & $\Lambda \uplus \{ e \}$\\
          $\pi : {}$ & $\hat{\Lambda}^* \to \Lambda$, $\hat{R}^* \to R^*$ natural projection with $\pi(e) = \varepsilon$ and $\pi(i) = i$ for all $i \in I$\\
          $I : {}$ & PCP index set, $|I| + |\Lambda| \geq 2$, $I \cap \Lambda = \emptyset$\\
          $\varphi, \psi : {}$ & $I \to (\Lambda \uplus \{ e \})^*$ PCP homomorphisms\\
          $L : {}$ & length of $\varphi(i)$ and $\psi(i)$ for all $i \in I$, $L \geq 2$, \textbf{also}:\\
          $L : {}$ & $I^* \to I^*$ homomorphism with $L(i) = i^L$ for all $i \in I$\\
          $R = {}$ & $\Lambda \uplus I$\\
          $\hat{R} = {}$ & $\hat{\Lambda} \uplus I = \Lambda \uplus I \uplus \{ e \} : {}$state set of $\mathcal{R}$\\
          $\hat{\mathcal{R}} = {}$ & $(\hat{R}, \Gamma, \rho) : {}$ complete \SAut generating $R^* = (\Lambda \cup I)^+$ with $e =_{\hat{\mathcal{R}}} \varepsilon$\\
          $\rho : {}$ & transition set of $\hat{\mathcal{R}}$\\
          $\Gamma : {} $ & alphabet of $\hat{\mathcal{R}}$ and $\mathcal{S}$\\
          
          $\#_1, \#_2 \not\in {}$ & $\hat{R} : {}$new states acting as the identity in $\mathcal{S}$\\
          $\mathcal{S} = {}$ & $(Q, \Gamma, \sigma) : {}$complete \SAut, extension of $\hat{\mathcal{R}}$ still generating $R^*$ with $e =_{\mathcal{S}} \varepsilon$\\
          $Q = {}$ & $\hat{R} \uplus \{ \#_1, \#_2 \} : {}$state set of $\mathcal{S}$ and $\mathcal{T}$\\
          $\sigma : {}$ & transition set of $\mathcal{S}$\\
          $a, b \not\in \Gamma : {}$ & new letters for $\mathcal{T}_1$\\
          $\mathcal{T}_1 = {}$ & $(Q, \Gamma \uplus \{ a, b \}, \delta_1) : {}$complete \SAut, extension of $\mathcal{S}$, see \autoref{fig:monoidT1transitions}\\
          $\delta_1 : {}$ & transition set of $\mathcal{T}_1$, see \autoref{fig:monoidT1transitions}\\
          $\mathcal{T} = {}$ & $(Q, \Sigma, \delta) = \mathcal{T}_1 \cup \mathcal{T}_2 : {}$complete \SAut with $e =_{\mathcal{T}} \varepsilon$, result of the reduction\\
          $\mathcal{T}_2 : {}$ & complete \SAut with new transitions for $\mathcal{T}$, see \autoref{fig:monoidT2}\\
          
          $\Sigma = {}$ & $\Gamma \uplus \{ a, b \} \uplus \{ \alpha_0, \alpha_L, \beta_0, \beta_L \} \uplus \{ \alpha_{i, \ell}, \beta_{i, \ell} \mid i \in I, 1 \leq \ell < L \} : {}$alphabet of $\mathcal{T}$\\
          $\pi_{\#} : {}$ & $Q^* \to \{ \#_1, \#_2 \}^*$ homomorphism with $\pi_{\#}(\#_x) = \#_x$ but $\pi_{\#}(\hat{r}) = \varepsilon$ for $\hat{r} \in \hat{R}$\\
          $\pi' :{}$ & $Q^* \to (R \cup \{ \#_1, \#_2 \})^*$ homomorphism extending $\pi$ with $\pi'(\#_x) = \#_x$ for $x \in \{ 1, 2 \}$\\
          \bottomrule
        \end{tabular}
        \caption{Various symbols in the order of their definition.}\label{tbl:monoidSymbols}
      \end{table}
    
    \paragraph{Correctness.}
    We have described how we may compute the complete \SAut $\mathcal{T} = (Q, \Sigma, \delta)$ where $e \in Q$ acts as the identity from the \DecProblem{ePCP} instance $\varphi, \psi, I, L$. Now we need to show that there is a solution for the \DecProblem{ePCP} instance if and only if $\mathscr{M}(\mathcal{T})$ is \textbf{not} (isomorphic to) $(Q \setminus \{ e \})^*$.
    
    Again, we start with the (easier) \enquote{only if} direction. Since $e$ acts as the identity (with respect to $\mathcal{T}$), we immediately obtain the following fact (compare to \autoref{lem:REquivalentImpliesTEquivalent}).
    \begin{fact}\label{fct:EEquivalentImpliesTEquivalent}
      Let $\hat{\bm{r}}_1, \hat{\bm{r}}_2 \in \hat{R}^*$ with $\hat{\bm{r}}_1 =_e \hat{\bm{r}}_2$. Then, we have $\hat{\bm{r}}_1 =_{\mathcal{T}} \hat{\bm{r}}_2$.
    \end{fact}
    
    Next, we show that a solution $\bm{i} \in I^+$ implies a (proper) relation in $\mathscr{M}(\mathcal{T})$ (compare to \autoref{lem:solutionImpliesRelation}), which shows that $\mathscr{M}(\mathcal{T})$ cannot be free with basis $Q \setminus \{ e \}$. Since we do not have powers in $\hat{\mathcal{R}}$ this time (and due to the construction in \autoref{fig:monoidT2}), we need to repeat each index/letter $L$ many times in the solution $\bm{i}$ to obtain the relation. In a slight abuse of notation, we define the homomorphism $L: I^* \to I^*$ by setting $L(i) = i^L$ for all $i \in I$.
    
    Before we prove the relation, we first show that, using this homomorphism, our construction behaves similar to the one we used in \autoref{sct:semigroupFreeness}.
    \begin{fact}\label{fct:monoidPhiPsiCrossDiagrams}
      For all $\alpha \in \{ \alpha_0, \alpha_L \}$, $\beta \in \{ \beta_0, \beta_L \}$ and $i \in I$, we have the cross diagrams
      \begin{center}
        \begin{tikzpicture}[baseline=(m-2-1.base)]
          \matrix[matrix of math nodes, text height=1.25ex, text depth=0.25ex] (m) {
                 & \alpha   &        \\
            L(i) &          & \varphi(i) \\
                 & \alpha_L &        \\
          };
          \foreach \i in {1} {
            \foreach \j in {1} {
              \draw[->] let
                \n1 = {int(2+\i)},
                \n2 = {int(1+\j)}
              in
                (m-\i-\n2) -> (m-\n1-\n2);
              \draw[->] let
                \n1 = {int(1+\i)},
                \n2 = {int(2+\j)}
              in
                (m-\n1-\j) -> (m-\n1-\n2);
            };
          };
        \end{tikzpicture}\quad
        and\quad
        \begin{tikzpicture}[baseline=(m-2-1.base)]
          \matrix[matrix of math nodes, text height=1.25ex, text depth=0.25ex] (m) {
                 & \beta   &         \\
            L(i) &         & \psi(i) \text{.} \\
                 & \beta_L &         \\
          };
          \foreach \i in {1} {
            \foreach \j in {1} {
              \draw[->] let
                \n1 = {int(2+\i)},
                \n2 = {int(1+\j)}
              in
                (m-\i-\n2) -> (m-\n1-\n2);
              \draw[->] let
                \n1 = {int(1+\i)},
                \n2 = {int(2+\j)}
              in
                (m-\n1-\j) -> (m-\n1-\n2);
            };
          };
        \end{tikzpicture}
      \end{center}
    \end{fact}
    \begin{proof}
      The fact follows from the construction of $\mathcal{T}$ (see \autoref{sfig:monoidT2dual}) since we have
      \\\vspace{\parskip}{\hspace*{\fill}
        \begin{tikzpicture}[baseline=(m-4-1.base)]
          \matrix[matrix of math nodes, text height=1.25ex, text depth=0.25ex] (m) {
              & \alpha            &        \\
            i &                   & \varphi_1(i) \\
              & \alpha_{i, 1}     &        \\
            \vdots & \vdots       & \vdots \\
              & \alpha_{i, L - 1} &        \\
            i &                   & \varphi_L(i) \\
              & \alpha_L          &        \\
          };
          \foreach \i in {1,5} {
            \foreach \j in {1} {
              \draw[->] let
                \n1 = {int(2+\i)},
                \n2 = {int(1+\j)}
              in
                (m-\i-\n2) -> (m-\n1-\n2);
              \draw[->] let
                \n1 = {int(1+\i)},
                \n2 = {int(2+\j)}
              in
                (m-\n1-\j) -> (m-\n1-\n2);
            };
          };
          \draw[decorate, decoration=brace] (m-6-1.south west) -- node[base left, anchor=base, rotate=90, yshift=1ex] {$L$ times} (m-2-1.north west);
          \draw[decorate, decoration=brace] (m-2-3.north east) -- node[right] {${}= \varphi(i)$} (m-6-3.south east);
        \end{tikzpicture}\quad
        and\quad
        \begin{tikzpicture}[baseline=(m-4-1.base)]
          \matrix[matrix of math nodes, text height=1.25ex, text depth=0.25ex] (m) {
              & \beta            &        \\
            i &                  & \psi_1(i) \\
              & \beta_{i, 1}     &        \\
            \vdots & \vdots      & \vdots \\
              & \beta_{i, L - 1} &        \\
            i &                  & \psi_L(i) \\
              & \beta_L          &        \\
          };
          \foreach \i in {1,5} {
            \foreach \j in {1} {
              \draw[->] let
                \n1 = {int(2+\i)},
                \n2 = {int(1+\j)}
              in
                (m-\i-\n2) -> (m-\n1-\n2);
              \draw[->] let
                \n1 = {int(1+\i)},
                \n2 = {int(2+\j)}
              in
                (m-\n1-\j) -> (m-\n1-\n2);
            };
          };
          \draw[decorate, decoration=brace] (m-6-1.south west) -- node[base left, anchor=base, rotate=90, yshift=1ex] {$L$ times} (m-2-1.north west);
          \draw[decorate, decoration=brace] (m-2-3.north east) -- node[right] {${}= \psi(i)$} (m-6-3.south east);
        \end{tikzpicture}.}
    \end{proof}
    
    \begin{proposition}\label{prop:monoidSolutionImpliesRelation}
      If $\bm{i} \in I^+$ is a solution for the \DecProblem{ePCP} instance, then we have
      \[
        \#_1 L(\bm{i}) \#_1 =_{\mathcal{T}} \#_1 L(\bm{i}) \#_2 \text{.}
      \]
      In particular, we have $\mathscr{M}(\mathcal{T}) \not\simeq (Q \setminus \{ e \})^*$.
    \end{proposition}
    \begin{proof}
      Our proof is very similar to the one for \autoref{lem:solutionImpliesRelation}.\footnote{In fact, it is easier because we now have a neutral element.} We show that the relation holds by showing that both sides act in the same way on all $u \in \Sigma^*$ using an induction. Thus, let $u = c u'$ for some $c \in \Sigma = \Gamma \cup \{ a, b \} \cup \{ \alpha_0, \alpha_L, \beta_0, \beta_L \} \cup \{ \alpha_{i, \ell}, \beta_{i, \ell} \mid i \in I, 1 \leq \ell < L \}$. For $c \in \Gamma$ (the alphabet of $\hat{\mathcal{R}}$), we have the cross diagram (see definition of $\mathcal{S}$)
      \begin{center}
        \begin{tikzpicture}[baseline=(m-4-1.base)]
          \matrix[matrix of math nodes, text height=1.25ex, text depth=0.25ex] (m) {
                   & c &      \\
              \#_x &   & e \\
                   & c &      \\
            L(\bm{i}) &   & L(\bm{i}) \cdot c    \\
                   & d &      \\
              \#_1 &   & e \\
                   & d &      \\
          };
          \foreach \i in {1,3,5} {
            \foreach \j in {1} {
              \draw[->] let
                \n1 = {int(2+\i)},
                \n2 = {int(1+\j)}
              in
                (m-\i-\n2) -> (m-\n1-\n2);
              \draw[->] let
                \n1 = {int(1+\i)},
                \n2 = {int(2+\j)}
              in
                (m-\n1-\j) -> (m-\n1-\n2);
            };
          };
        \end{tikzpicture}
      \end{center}
      for both, $x = 1$ and $x = 2$ with the same $d \in \Gamma$. Since the state sequence on the right is the same in both cases, there is nothing more to show.
      
      The cases $c \in \{ a \} \cup \{ \alpha_0, \alpha_L, \beta_0, \beta_L \} \cup \{ \alpha_{i, \ell}, \beta_{i, \ell} \mid i \in I, 1 \leq \ell < L \}$ are similar; they are depicted in \autoref{fig:monoidSolutionImplesRelation}. The case $c = b$ requires induction but is still similar; it is depicted in \autoref{sfig:monoidSolutionImpliesRelationB}.
      
      Finally, the case $c = \iota$ is again the most interesting one. Writing $\bm{i} = i_K \dots i_2 i_1$ for $i_1, \dots, i_K \in I$, we obtain $L(\bm{i}) = i_K^L \dots i_2^L i_1^L$ and \autoref{fct:monoidPhiPsiCrossDiagrams} yields the cross diagrams
      \begin{center}
        \begin{tikzpicture}[baseline=(m-6-1.base)]
          \matrix[matrix of math nodes, text height=1.25ex, text depth=0.25ex] (m) {
                   & \iota   &      \\
              \#_1 &         & e    \\
                   & \alpha_0  &      \\
               i_1^L &         & \varphi(i_1) \\
                   & \alpha_L &      \\
               i_2^L &         & \varphi(i_2) \\
                   & \alpha_L &      \\
              \vdots &       & \vdots \\
                   & \alpha_L &      \\
              i_K^L &       & \varphi(i_K) \\
                   & \alpha_L &      \\
              \#_1 &         & e    \\
                   & f       &      \\
          };
          \foreach \i in {1,3,5,9,11} {
            \foreach \j in {1} {
              \draw[->] let
                \n1 = {int(2+\i)},
                \n2 = {int(1+\j)}
              in
                (m-\i-\n2) -> (m-\n1-\n2);
              \draw[->] let
                \n1 = {int(1+\i)},
                \n2 = {int(2+\j)}
              in
                (m-\n1-\j) -> (m-\n1-\n2);
            };
          };
        \end{tikzpicture}
        and
        \begin{tikzpicture}[baseline=(m-6-1.base)]
          \matrix[matrix of math nodes, text height=1.25ex, text depth=0.25ex] (m) {
                 & \iota  &      \\
            \#_2 &        & e    \\
                 & \beta_0  &      \\
             i_1^L &        & \psi(i_1) \\
                 & \beta_L &      \\
             i_2^L &        & \psi(i_2) \\
                 & \beta_L &      \\
            \vdots &      & \vdots \\
                 & \beta_L &      \\
            i_K^L &      & \psi(i_K) \\
                 & \beta_L &      \\
            \#_1 &        & e    \\
                 & f      &      \\
          };
          \foreach \i in {1,3,5,9,11} {
            \foreach \j in {1} {
              \draw[->] let
                \n1 = {int(2+\i)},
                \n2 = {int(1+\j)}
              in
                (m-\i-\n2) -> (m-\n1-\n2);
              \draw[->] let
                \n1 = {int(1+\i)},
                \n2 = {int(2+\j)}
              in
                (m-\n1-\j) -> (m-\n1-\n2);
            };
          };
        \end{tikzpicture}.
      \end{center}
      Since $\bm{i} = i_K \dots i_2 i_1$ is a solution, we have $\varphi(i_K) \dots \varphi(i_2) \varphi(i_1) =_e \psi(i_K) \dots \psi(i_2) \psi(i_1)$. Thus, \autoref{fct:EEquivalentImpliesTEquivalent} implies $e \varphi(i_K) \dots \varphi(i_2) \varphi(i_1) e =_{\mathcal{T}} e \psi(i_K) \dots \psi(i_2) \psi(i_1) e$ and we are done.
    \end{proof}
    \begin{figure}\centering
      \subcaptionbox{$c \in \{ \alpha_0 \} \cup \{ \alpha_{i, \ell} \mid i \in I, 1 \leq \ell < L \}$\label{sfig:monoidSolutionImpliesRelationAlphaBeta}}{%
        \begin{tikzpicture}[baseline=(m-4-1.base)]
          \matrix[matrix of math nodes, text height=1.25ex, text depth=0.25ex, ampersand replacement=\&] (m) {
                 \& \alpha_0/\alpha_{i, \ell} \&      \\
            \#_x \&             \& e    \\
                 \& f_{\alpha}  \& \\
           L(\bm{i})\&   \& e^{L |\bm{i}|} \\
                 \& f_\alpha\&      \\
            \#_1 \&   \& e    \\
                 \& f_\alpha \&      \\
          };
          \foreach \i in {1,3,5} {
            \foreach \j in {1} {
              \draw[->] let
                \n1 = {int(2+\i)},
                \n2 = {int(1+\j)}
              in
                (m-\i-\n2) -> (m-\n1-\n2);
              \draw[->] let
                \n1 = {int(1+\i)},
                \n2 = {int(2+\j)}
              in
                (m-\n1-\j) -> (m-\n1-\n2);
            };
          };
        \end{tikzpicture}%
      }
      \subcaptionbox{$c = \alpha_L$\label{sfig:monoidSolutionImpliesRelationAlphaBetaL}}{%
        \begin{tikzpicture}[baseline=(m-4-1.base)]
          \matrix[matrix of math nodes, text height=1.25ex, text depth=0.25ex, ampersand replacement=\&] (m) {
               \& \alpha_L \&      \\
          \#_x \&                \& e \\
               \& f \& \\
            L(\bm{i})\&   \& e^{L |\bm{i}|} \\
               \& f \&      \\
            \#_1 \&   \& e  \\
               \& f \&      \\
          };
          \foreach \i in {1,3,5} {
            \foreach \j in {1} {
              \draw[->] let
                \n1 = {int(2+\i)},
                \n2 = {int(1+\j)}
              in
                (m-\i-\n2) -> (m-\n1-\n2);
              \draw[->] let
                \n1 = {int(1+\i)},
                \n2 = {int(2+\j)}
              in
                (m-\n1-\j) -> (m-\n1-\n2);
            };
          };
        \end{tikzpicture}%
      }
      \subcaptionbox{$c \in \{ f_\alpha, f_\beta, f \}$\label{sfig:monoidSolutionImpliesRelationF}}{%
        \begin{tikzpicture}[baseline=(m-4-1.base)]
          \matrix[matrix of math nodes, text height=1.25ex, text depth=0.25ex, ampersand replacement=\&] (m) {
              \& f_\alpha/f \&      \\
          \#_x \&              \& e \\
              \& f_\alpha/f \& \\
            L(\bm{i})\&   \& e^{L |\bm{i}|} \\
              \& f_\alpha/f \&      \\
            \#_1 \&   \& e \\
              \& f_\alpha/f \&      \\
          };
          \foreach \i in {1,3,5} {
            \foreach \j in {1} {
              \draw[->] let
                \n1 = {int(2+\i)},
                \n2 = {int(1+\j)}
              in
                (m-\i-\n2) -> (m-\n1-\n2);
              \draw[->] let
                \n1 = {int(1+\i)},
                \n2 = {int(2+\j)}
              in
                (m-\n1-\j) -> (m-\n1-\n2);
            };
          };
        \end{tikzpicture}%
      }\newline
      \subcaptionbox{$c = a$\label{sfig:monoidSolutionImpliesRelationA}}{%
        \begin{tikzpicture}[baseline=(m-4-1.base)]
          \matrix[matrix of math nodes, text height=1.25ex, text depth=0.25ex, ampersand replacement=\&] (m) {
              \& a \&           \\
          \#_x \& \& e \\
              \& b \& \\
            L(\bm{i})\&   \& L(\bm{i}) \\
              \& b \&      \\
            \#_1 \&   \& \#_1    \\
              \& b \&      \\
          };
          \foreach \i in {1,3,5} {
            \foreach \j in {1} {
              \draw[->] let
                \n1 = {int(2+\i)},
                \n2 = {int(1+\j)}
              in
                (m-\i-\n2) -> (m-\n1-\n2);
              \draw[->] let
                \n1 = {int(1+\i)},
                \n2 = {int(2+\j)}
              in
                (m-\n1-\j) -> (m-\n1-\n2);
            };
          };
        \end{tikzpicture}%
      }
      \subcaptionbox{$c = b$\label{sfig:monoidSolutionImpliesRelationB}}{%
        \begin{tikzpicture}[baseline=(m-4-1.base)]
          \matrix[matrix of math nodes, text height=1.25ex, text depth=0.25ex, ampersand replacement=\&] (m) {
              \& b \&      \\
          \#_x \& \& \#_x  \\
              \& b \& \\
            L(\bm{i})\&   \& L(\bm{i}) \\
              \& b \&      \\
            \#_1 \&   \& \#_1    \\
              \& b \&      \\
          };
          \foreach \i in {1,3,5} {
            \foreach \j in {1} {
              \draw[->] let
                \n1 = {int(2+\i)},
                \n2 = {int(1+\j)}
              in
                (m-\i-\n2) -> (m-\n1-\n2);
              \draw[->] let
                \n1 = {int(1+\i)},
                \n2 = {int(2+\j)}
              in
                (m-\n1-\j) -> (m-\n1-\n2);
            };
          };
        \end{tikzpicture}%
      }
      \caption{Various cases for $c \in \Sigma$. The cross diagrams hold for $x \in \{ 1, 2 \}$. The cross diagrams for $c \in \{ \beta_0, \beta_L, f_\beta \} \cup \{ \beta_{i, \ell} \mid i \in I, 1 \leq \ell < L \}$ are symmetric to their $\alpha$ analogues.}\label{fig:monoidSolutionImplesRelation}
    \end{figure}

    Although we do not strictly require it for our current proof, we also obtain that the generated monoid is not free if a solution exists. This can be shown using the same proof as for \autoref{prop:solutionImpliesNotFree} (only using \autoref{prop:monoidSolutionImpliesRelation} and the relation given there).
    \begin{corollary}
      If the \DecProblem{ePCP} instance has a solution, $\mathscr{M}(\mathcal{T})$ is not a free monoid.
    \end{corollary}
    
    \paragraph{Converse Direction.}
    For showing that the \DecProblem{ePCP} instance has a solution if $\mathscr{M}(\mathcal{T})$ is not (isomorphic to) $(Q \setminus \{ e \})^*$, we modify the definition of compatible state sequences from \autoref{def:compatible} by requiring only $e$-equivalence in the individual parts.
    \begin{definition}[compatible state sequences]\label{def:monoidCompatible}
      Let $\bm{p}, \bm{q} \in Q^*$ and factorize them (uniquely) as
      \[
        \bm{p} = \left( \bm{p}_s \#_{x_s} \right) \dots \left( \bm{p}_1 \#_{x_1} \right) \bm{p}_0 \quad \text{and} \quad \bm{q} = \left( \bm{q}_t \#_{y_t} \right) \dots \left( \bm{q}_1 \#_{y_1} \right) \bm{q}_0
      \]
      with $\bm{p}_0, \dots, \bm{p}_s, \bm{q}_0, \dots, \bm{q}_t \in \hat{R}^*$ and $x_1, \dots, x_s, y_1, \dots, y_t \in \{ 1, 2 \}$. We define:
      \[
        \bm{p} \text{ and } \bm{q} \text{ are \emph{compatible}} \iff s = t \text{ and } \forall\, 0 \leq i \leq s = t: \bm{p}_i =_e \bm{q}_i
      \]
    \end{definition}
    
    Similarly to \autoref{lem:relationIsCompatible} (and using a simplified proof), we still have that every relation with respect to $\mathcal{T}$ is compatible.
    \begin{lemma}\label{lem:monoidRelationIsCompatible}
      Let $\bm{p}, \bm{q} \in Q^*$ with $\bm{p} =_{\mathcal{T}} \bm{q}$. Then, we have that $\bm{p}$ and $\bm{q}$ are compatible.
    \end{lemma}
    \begin{proof}
      We factorize $\bm{p}$ and $\bm{q}$ in the same way as in \autoref{def:monoidCompatible} and show the statement by induction on $s + t$.
      For $s = t = 0$, we have $\bm{p}_0 = \bm{p} =_{\mathcal{T}} \bm{q} = \bm{q}_0$. Since $\hat{\mathcal{R}}$ is a subautomaton of $\mathcal{T}$, this implies $\bm{p}_0 =_{\hat{\mathcal{R}}} \bm{q}_0$ and, equivalently, $\bm{p} = \bm{p}_0 =_e \bm{q}_0 = \bm{q}$.
      
      For the inductive step ($s + t > 0$), we may assume $s > 0$ (due to symmetry) or, in other words, that $\bm{p}$ contains at least one $\#_1$ or $\#_2$. We have $\bm{p} \circ a = b$ (compare to \autoref{sfig:monoidT1dual}) and, thus, due to $\bm{p} =_{\mathcal{T}} \bm{q}$, also $\bm{q} \circ a = \bm{p} \circ a = b$. This is only possible (again, compare to \autoref{sfig:monoidT1dual}) if $\bm{q}$ also contains at least one $\#_1$ or $\#_2$, i.\,e.\ if $t > 0$.
      
      From \autoref{fct:monoidRemoveBlocks} (with $\mu = 1$), we obtain (for both $\bm{p}$ and $\bm{q}$):
      \begin{align*}
        \bm{p} \cdot a ={}& \bm{p}' e e^{|\bm{p}_0|_R}\\
        &\text{for } \bm{p}' = \left( \bm{p}_s \#_{x_s} \right) \dots \left( \bm{p}_2 \#_{x_2} \right) \bm{p}_1 \text{ and}\\
        \bm{q} \cdot a ={}& \bm{q}' e e^{|\bm{q}_0|_R}\\
        &\text{for } \bm{q}' = \left( \bm{q}_t \#_{x_t} \right) \dots \left( \bm{q}_2 \#_{x_2} \right) \bm{q}_1
      \end{align*}
      Thus, $\bm{p} =_{\mathcal{T}} \bm{q}$ implies $\bm{p}' =_{\mathcal{T}} \bm{p}' e e^{|\bm{p}_0|} =_{\mathcal{T}} \bm{q}' e e^{|\bm{q}_0|} =_{\mathcal{T}} \bm{q}'$ and we can apply the induction hypothesis to obtain that $\bm{p}'$ and $\bm{q}'$ are compatible, which implies $s = t$ and $\bm{p}_{\mu} =_e \bm{q}_{\mu}$ for all $1 \leq \mu \leq s = t$.
      In particular, we also obtain $\bm{p}_s \bm{p}_{s - 1} \ldots \bm{p}_1 =_e \bm{q}_t \bm{q}_{t - 1} \ldots \bm{q}_1$.
      
      Note that $\mathcal{S}$ is a subautomaton of $\mathcal{T}$ and that, therefore, $\bm{p} =_{\mathcal{T}} \bm{q}$ implies $\bm{p} =_{\mathcal{S}} \bm{q}$. Since $\#_1$ and $\#_2$ act as the identity in $\mathcal{S}$ by construction, this shows $\bm{p}_s \ldots \bm{p}_1 \bm{p}_0 =_{\mathcal{S}} \bm{q}_t \ldots \bm{q}_1 \bm{q}_0$ and, because of $\mathscr{M}(\mathcal{S}) \simeq R^*$, also $\bm{p}_s \ldots \bm{p}_1 \bm{p}_0 =_e \bm{q}_t \ldots \bm{q}_1 \bm{q}_0$. Since $R^*$ as a free monoid is cancellative, this (together with $\bm{p}_s \bm{p}_{s - 1} \ldots \bm{p}_1 =_e \bm{q}_t \bm{q}_{t - 1} \ldots \bm{q}_1$) yields $\bm{p}_0 =_e \bm{q}_0$, which concludes the proof.
    \end{proof}
    
    That two state sequences from $Q^*$ form a relation with respect to $\mathcal{T}$ if their projections are equal in $(Q \setminus \{ e \})^*$ (the analogue of \autoref{lem:sharpCompatibleImpliesRelation}) follows because $e$ (the only letter which is changed/removed by the projection) acts as the identity (compare to \autoref{fct:EEquivalentImpliesTEquivalent}). In order to make this statement formally, we define $\pi_\#$ in our setting in the same way as before in \autoref{def:sharpProjection} (i.\,e.\ $\pi_{\#}(\#_{x}) = \#_x$ for $x \in \{ 1, 2 \}$ and $\pi_{\#}(\hat{r}) = \varepsilon$ for $\hat{r} \in \hat{R}$).
    
    \begin{fact}\label{fct:monoidSharpCompatibleImpliesRelation}
      Let $\bm{p}, \bm{q} \in Q^*$ such that $\bm{p}$ and $\bm{q}$ are compatible and we have $\pi_\#(\bm{p}) = \pi_\#(\bm{q})$. Then, we have $\bm{p} =_{\mathcal{T}} \bm{q}$.
    \end{fact}

    Next, we show the analogue of \autoref{lem:freeIfRelationsSharp}: $\mathscr{M}(\mathcal{T})$ is isomorphic to $(Q \setminus \{ e \})^*$ unless we have a relation whose sides only differ in their images under $\pi_\#$. The isomorphism is given by the extension of the natural projection $\pi$ keeping the letters $\{ \#_1, \#_2 \}$ fixed.
    \begin{lemma}\label{lem:monoidIsomorphicIfRelationsSharp}
      We have $\pi_\#(\bm{p}) = \pi_\#(\bm{q})$ for all $\bm{p}, \bm{q} \in Q^*$ with $\bm{p} =_{\mathcal{T}} \bm{q}$ if and only if $\pi': Q^* \to (Q \setminus \{ e \})^* = (R \cup \{ \#_1, \#_2 \})^*$ with $q \mapsto q$, $e \mapsto \varepsilon$ and $\#_x \mapsto \#_x$ for $x \in \{ 1, 2 \}$ induces a well-defined isomorphism $\mathscr{M}(\mathcal{T}) \to (Q \setminus \{ e \})^*$.
      
      In particular, $\mathscr{M}(\mathcal{T})$ is isomorphic to $(Q \setminus \{ e \})^*$ if we have $\pi_\#(\bm{p}) = \pi_\#(\bm{q})$ for all $\bm{p}, \bm{q} \in Q^*$ with $\bm{p} =_{\mathcal{T}} \bm{q}$.
    \end{lemma}

    \begin{proof}
      Clearly, we have $\pi'(\bm{p}) = \pi'(\bm{q})$ for $\bm{p}, \bm{q} \in Q^+$ if and only if $\bm{p}$ and $\bm{q}$ are compatible and $\pi_{\#}(\bm{p}) = \pi_{\#}(\bm{q})$ holds.
      
      First, we show that $\pi'$ induces a well-defined isomorphism if we have $\pi_\#(\bm{p}) = \pi_\#(\bm{q})$ for all $\bm{p}, \bm{q} \in Q^*$ with $\bm{p} =_\mathcal{T} \bm{q}$.
      To show that the isomorphism is well-defined, let $\bm{p}, \bm{q} \in Q^*$ with $\bm{p} =_{\mathcal{T}} \bm{q}$. By \autoref{lem:monoidRelationIsCompatible}, we have that $\bm{p}$ and $\bm{q}$ are compatible. By hypothesis, we also obtain $\pi_{\#}(\bm{p}) = \pi_{\#}(\bm{q})$.
      To show that the isomorphism is indeed injective, let $\bm{p}, \bm{q} \in Q^*$ be compatible with $\pi_{\#}(\bm{p}) = \pi_{\#}(\bm{q})$. Then, by \autoref{fct:monoidSharpCompatibleImpliesRelation}, this implies $\bm{p} =_{\mathcal{T}} \bm{q}$. Finally, surjectivity and the homomorphism property are trivial.
      
      For the other direction, assume that $\pi'$ induces a well-defined isomorphism $\mathscr{M}(\mathcal{T}) \to (Q \setminus \{ e \})^*$.
      Then, $\bm{p} =_\mathcal{T} \bm{q}$ implies $\pi'(\bm{p}) = \pi'(\bm{q})$ and, in particular, $\pi_\#(\bm{p}) = \pi_\#(\bm{q})$.
    \end{proof}
    
    This allows us now to show that the \DecProblem{ePCP} instance has a solution if the monoid generated by $\mathcal{T}$ is not isomorphic to $(Q \setminus \{ e \})^*$.
    The proof is again a simplified version of the one for \autoref{lem:sharpRelationImpliesSolution}.
    \begin{lemma}\label{prop:monoiNotIsomorphicImpliesSolution}
      If $\mathscr{M}(\mathcal{T})$ is not isomorphic to $(Q \setminus \{ e \})^*$, the \DecProblem{ePCP} instance has a solution.
    \end{lemma}
    \begin{proof}
      If $\mathscr{M}(\mathcal{T})$ is not isomorphic to $(Q \setminus \{ e \})^*$, then, by \autoref{lem:monoidIsomorphicIfRelationsSharp}, there must be $\bm{p}, \bm{q} \in Q^*$ with $\bm{p} =_{\mathcal{T}} \bm{q}$ but $\pi_\#(\bm{p}) \neq \pi_\#(\bm{q})$. We factorize these $\bm{p}$ and $\bm{q}$ in the same way as in \autoref{def:monoidCompatible} and observe that $\bm{p}$ and $\bm{q}$ are compatible by \autoref{lem:monoidRelationIsCompatible}. We may assume that there is some $1 \leq \mu_0 \leq s = t$ with $\#_{x_{\mu_0}} = \#_1$ but $\#_{y_{\mu_0}} = \#_2$ (due to symmetry).
      
      As before, we may assume $\mu_0 = 1$ by \autoref{fct:monoidRemoveBlocks} since we have
      \begin{align*}
        \bm{p} \cdot a^{\mu_0 - 1} ={}& \bm{p}' e e^{|\bm{p}_0|}\\
        &\text{for } \bm{p}' = \left( \bm{p}_s \#_{x_s} \right) \dots \left( \bm{p}_{\mu_0} \#_{x_{\mu_0}} \right) \bm{p}_{\mu_0 - 1} \text{ and}\\
        \bm{q} \cdot a^{\mu_0 - 1} ={}& \bm{q}' e e^{|\bm{q}_0|}\\
        &\text{for } \bm{q}' = \left( \bm{q}_t \#_{x_t} \right) \dots \left( \bm{q}_{\mu_0} \#_{x_{\mu_0}} \right) \bm{q}_{\mu_0 - 1}
      \end{align*}
      and, thus, may replace $\bm{p}$ by $\bm{p}'$ and $\bm{q}$ by $\bm{q}'$ (for which we still have $\bm{p}' =_{\mathcal{T}} \bm{q}'$).

      With this assumptions, we apply $\bm{p}$ and $\bm{q}$ to $\iota$ and obtain the cross diagrams (see \autoref{fig:monoidT2})
      \begin{center}
        \begin{tikzpicture}[baseline=(m-6-1.base)]
          \matrix[matrix of math nodes, text height=1.25ex, text depth=0.25ex] (m) {
                     & \iota  &                \\
            \bm{p}_0 &        & e^{|\bm{p}_0|} \\
                     & \iota  &                \\
            \#_1     &        & e              \\
                     & \alpha_0 &                \\
            \bm{p}_1 &        & \bm{p}_1'      \\
                     & c_1    &                \\
            \#_{x_2} &        & p_2'           \\
                     & c_2    &                \\
            \tilde{\bm{p}} &  & \tilde{\bm{p}}'\\
                     & c      &                \\
          };
          \foreach \i in {1,3,5,7,9} {
            \foreach \j in {1} {
              \draw[->] let
                \n1 = {int(2+\i)},
                \n2 = {int(1+\j)}
              in
                (m-\i-\n2) -> (m-\n1-\n2);
              \draw[->] let
                \n1 = {int(1+\i)},
                \n2 = {int(2+\j)}
              in
                (m-\n1-\j) -> (m-\n1-\n2);
            };
          };
        \end{tikzpicture}
        and
        \begin{tikzpicture}[baseline=(m-6-1.base)]
          \matrix[matrix of math nodes, text height=1.25ex, text depth=0.25ex] (m) {
                     & \iota &                \\
            \bm{q}_0 &       & e^{|\bm{q}_0|} \\
                     & \iota &                \\
            \#_2     &       & e              \\
                     & \beta_0 &                \\
            \bm{q}_1 &       & \bm{q}_1'      \\
                     & d_1   &                \\
            \#_{y_2} &       & q_2'           \\
                     & d_2   &                \\
            \tilde{\bm{q}} & & \tilde{\bm{q}}'\\
                     & d     &                \\
          };
          \foreach \i in {1,3,5,7,9} {
            \foreach \j in {1} {
              \draw[->] let
                \n1 = {int(2+\i)},
                \n2 = {int(1+\j)}
              in
                (m-\i-\n2) -> (m-\n1-\n2);
              \draw[->] let
                \n1 = {int(1+\i)},
                \n2 = {int(2+\j)}
              in
                (m-\n1-\j) -> (m-\n1-\n2);
            };
          };
        \end{tikzpicture}
      \end{center}
      for $\tilde{\bm{p}} = \bm{p}_s \#_{x_s} \ldots \bm{p}_3 \#_{x_3} \bm{p}_2$, $\tilde{\bm{q}} = \bm{q}_t \#_{y_t} \ldots \bm{q}_3 \#_{y_3} \bm{q}_2$ and some $\bm{p}_1', \tilde{\bm{p}}', \bm{q}_1', \tilde{\bm{q}}' \in Q^*$, $p_2', q_2' \in Q$ and $c_1, c_2, c, d_1, d_2, d \in \Gamma$. Since we have $\bm{p} =_{\mathcal{T}} \bm{q}$, we must have $c = d$ and, by the construction of $\mathcal{T}$, this is only possible if $c = f = d$ (see \autoref{sfig:monoidT2dual}: two paths starting in $\iota$ with one passing through $\alpha_0$ and the other one passing through $\beta_0$ can only re-join in $f$). This, in turn, is only possible if we have $\bm{p}_1 =_e L(\bm{i})$ and $\bm{q}_1 =_e L(\bm{j})$ for some $\bm{i}, \bm{j} \in I^+$. Since $\bm{p}$ and $\bm{q}$ are compatible, we must even have $L(\bm{i}) =_e \bm{p}_1 =_e \bm{q}_1 =_e L(\bm{j})$, which implies $\bm{i} = \bm{j}$. Additionally, we obtain (from \autoref{fct:monoidPhiPsiCrossDiagrams}) $\bm{p}_1' =_e \varphi(\bm{i})$, $c_1 = \alpha_L$, $p_2' = e$, $c_2 = f$, $\bm{q}_1' =_e \psi(\bm{i})$, $d_1 = \beta_L$, $q_2' = e$ and $d_2 = f$. Finally, we also get $\tilde{\bm{p}}' = e^{|\tilde{\bm{p}}'|}$ and $\tilde{\bm{q}}' = e^{|\tilde{\bm{q}}'|}$ from the construction of $\mathcal{T}$.

      Thus (and because $e$ acts as the identity), we have $\varphi(\bm{i}) =_\mathcal{T} \psi(\bm{i})$, which implies $\varphi(\bm{i}) =_e \psi(\bm{i})$ by \autoref{lem:monoidRelationIsCompatible} and, therefore, that $\bm{i}$ is a solution for the \DecProblem{ePCP} instance.
    \end{proof}
    
    \paragraph{Free Presentation of Automaton Monoids.}
    By \autoref{prop:monoidSolutionImpliesRelation} and \autoref{prop:monoiNotIsomorphicImpliesSolution}, we have now reduced \DecProblem{ePCP} to \DecProblem{Free Monoid Presentation} and obtain:
    \begin{theorem}\label{thm:monoidFreePresentationIsUndecidable}
      The free presentation problem for automaton monoids
      \FreeMonoidPresentation
      is undecidable.
    \end{theorem}
  \end{section}
  
  \begin{section}{Open Problems}\label{sec:openProblems}
    In \autoref{thm:mainTheorem}, we have shown that the freeness problems for automaton semigroups and for automaton monoids are undecidable even if we restrict the alphabet to a size of $25$. Since free semigroups and monoids (and, in fact, even groups) can be generated already with a binary alphabet (compare also \cite{klimann2016automaton} for freeness and binary alphabet/two-state automata), there is an immediate question of how far the alphabet size can be reduced further. In particular, it does not seem unlikely that the problem is already undecidable over binary alphabets:
    \begin{openproblem}
      Are the problems
      \problem
        [a binary alphabet $\Sigma$]
        {a (complete) \SAut $\mathcal{T} = (Q, \Sigma, \delta)$}
        {if $\mathscr{S}(\mathcal{T})$/$\mathscr{M}(\mathcal{T})$ free?}\noindent
      decidable?
    \end{openproblem}
    
    A similar question also poses itself for \autoref{thm:monoidFreePresentationIsUndecidable}:
    \begin{openproblem}
      Does the problem in \autoref{thm:monoidFreePresentationIsUndecidable} (the free presentation problem for automaton monoids) remain undecidable if we fix the underlying alphabet? What about the binary alphabet case?
    \end{openproblem}
    
    Similarly, we also have \autoref{thm:monoidFreePresentationIsUndecidable} only for the monoid case and it is natural to ask whether the analogue for automaton semigroups also holds:
    \begin{openproblem}
      Is the problem
      \problem{%
        a (complete) \SAut $\mathcal{T} = (Q, \Sigma, \delta)$%
      }{%
        is $\mathscr{S}(\mathcal{T}) \simeq Q^+$?
      }\noindent
      decidable?
      What happens if we restrict the alphabet?
    \end{openproblem}
    
    In \autoref{thm:cancellativityEquidivisibilityIsUndecidable}, we have also shown that it is not possible to test whether a given automaton semigroup (or monoid) is equidivisible. By Levi's lemma (\autoref{fct:leviLemma}) this is one of the two conditions that, together, are equivalent to a semigroup (monoid) being free while the other one is the existence of a (proper) length function. So, a natural question is whether we can test if a given automaton semigroup or monoid admits a (proper) length function:
    \begin{openproblem}
      Is the problem
      \problem{%
        a (complete) \SAut $\mathcal{T}$
      }{%
        does $\mathscr{S}(\mathcal{T})$ ($\mathscr{M}(\mathcal{T})$) admit a (proper) length function?
      }\noindent
      decidable?
    \end{openproblem}
    \noindent{}We highly suspect this problem to be undecidable and it seems likely that our construction can be adapted to show this.
    
    Finally, all of the above questions can also be investigated with respect to the activity of the given automaton (see \cite{sidki2000automorphisms} or e.\,g.\ \cite{waechter2024word} for definitions). We suspect that the freeness problem is decidable for bounded automaton monoids but that it is already undecidable for linear activity. Please note in this context that, while automata of polynomial activity cannot generate free groups \cite{sidki2004automata}, examples of bounded automaton semigroups generating free monoids do exist (see e.\,g.\ \cite{cavaleri2023class}).
    \begin{openproblem}
      At which level of the activity hierarchy (as defined in \cite{bartholdi2020hierarchy}; see also \cite{waechter2024word}) does the freeness problem for automaton monoids become undecidable?
    \end{openproblem}
    
    Of course, it also remains open whether the freeness problem for automaton groups \cite[7.2 b)]{grigorchuk2000automata} is decidable. As an intermediate step between semigroups/monoids and groups, one can also ask whether the semigroup/monoid generated by an invertible automaton is free. A complete \SAut $\mathcal{T} = (Q, \Sigma, \delta)$ is \emph{invertible} if, for every $p \in Q$ and $b \in \Sigma$, there are exactly one $a \in \Sigma$ and $q \in Q$ with $\trans{p}{a}{b}{q} \in \delta$. A complete and invertible \SAut is called a \emph{\GAut}.
    \begin{openproblem}
      Is the problem
      \problem{%
        a \GAut $\mathcal{T}$
      }{%
        is $\mathscr{S}(\mathcal{T})$ ($\mathscr{M}(\mathcal{T})$) a free semigroup (monoid)?
      }\noindent
      decidable?
    \end{openproblem}
  \end{section}
  
  \begin{section}{Acknowledgement}
    The first and the second author are members of the National Research Group GNSAGA (Gruppo Nazionale per le Strutture Algebriche, Geometriche e le loro Applicazioni) of Indam.
    
    Major parts of this word were conducted while the third author was affiliated with Politecnico di Milano where he was funded by the Deutsche Forschungsgemeinschaft (DFG, German Research Foundation) – 492814705. During later stages of this work, he moved to Universität des Saarlandes, partly funded by ERC grant 101097307. The listed affiliation is his current one where the extensions in this journal version were supported by the Engineering and Physical Sciences Research Council [grant number EP/Y008626/1].
  \end{section}
  
  \bibliographystyle{plain}
  \bibliography{references}
\end{document}